\documentclass[letterpaper, 11pt]{article}
\usepackage{graphicx}
\usepackage{etoolbox}
\usepackage{amsmath}
\usepackage{amsfonts}
\usepackage{hyperref}
\usepackage{amsthm}
 \usepackage{amscd}
\usepackage{mathrsfs}
\usepackage{caption}
\usepackage{subcaption}
\usepackage{float}
\usepackage[top=1in,left=1in,right=1in,bottom=1in]{geometry}

\newtheorem{theorem}{Theorem}[section]
\newtheorem{proposition}[theorem]{Proposition}

\newtheorem{example}[theorem]{Example}

\newtheorem{remark}[theorem]{Remark}

\hypersetup{
	colorlinks = true,
	linkcolor  = red,
	citecolor  = green,
	filecolor  = cyan,
	urlcolor   = magenta,}
\newcommand \bR {\mathbb{R}}
\newcommand{\bC}{\mathbb{C}}
\newcommand{\bCp}{\mathbb{C}^+}

\newcommand {\bCpb} {\overline{\mathbb{C}^+}}

\newcommand{\ds}{\displaystyle}

\numberwithin{equation}{section}

\title{The transformations to remove or add bound states for the half-line matrix Schr\"odinger operator
\thanks{Research partially supported by the project PAPIIT-DGAPA, UNAM IN100321.}}
\author{Tuncay Aktosun\thanks{aktosun@uta.edu}\\
Department of Mathematics\\
University of Texas at Arlington\\
Arlington, TX 76019-0408, USA\\
\\
Ricardo Weder\thanks{weder@unam.mx. Emeritus fellow Sistema Nacional de Investigadores, CONAHCYT, M\'exico.}\\
Departamento de F\'\i sica Matem\'atica\\
Instituto de Investigaciones en
Matem\'aticas Aplicadas y en Sistemas\\
Universidad Nacional Aut\'onoma de M\'exico\\
Apartado Postal 20-126, IIMAS-UNAM\\
 Ciudad de M\'exico, CP 01000, M\'exico}

\date{}

\begin{document}

\maketitle

\begin{abstract}
We present the transformations to remove or add bound states or to decrease or increase the multiplicities of any existing bound states  for the half-line matrix-valued Schr\"odinger 
operator with the general selfadjoint boundary condition, without changing the continuous spectrum of the operator. When the matrix-valued potential is selfadjoint, is integrable, and has a finite first moment, the relevant 
transformations are constructed through the development of the Gel'fand--Levitan method for the corresponding Schr\"odinger 
operator. In particular, the bound-state normalization matrices are constructed at each bound state for any multiplicity. The transformations are obtained 
for all relevant quantities, including the matrix potential, the Jost solution, the regular solution, the Jost matrix, the scattering matrix, and the
boundary condition. For each bound state, the corresponding dependency matrix is introduced by connecting the normalization matrix
used in the Gel'fand--Levitan method and the normalization matrix used in the Marchenko method of inverse scattering.
Various estimates are provided to describe the large spacial asymptotics for the change in the potential when the bound states
are removed or added or their multiplicities are modified. An explicit example is provided showing that an asymptotic estimate available in
 the literature in the scalar case for the potential increment
is incorrect.
\end{abstract}



\newpage

\section{Introduction}
\label{section1}

In this paper we consider the half-line matrix-valued Schr\"odinger 
operator with the general selfadjoint boundary condition under the assumption that the matrix-valued potential is selfadjoint, is integrable, and has a finite first moment.
We analyze the change in all relevant quantities associated with this Schr\"odinger operator when we change its discrete spectrum by removing or adding a finite number of
eigenvalues or decreasing or increasing the multiplicities of the existing eigenvalues, without changing the
continuous spectrum. In other words, we determine the transformations of all relevant quantities associated with this Schr\"odinger operator when a bound state of any multiplicity is removed from or added to the discrete spectrum or the multiplicity of a bound state is changed.
Since the boundary condition directly affects the definition of the Schr\"odinger operator, in general we expect that the selfadjoint boundary condition also undergoes
a change, with the only exception that a purely Dirichlet boundary condition remains a purely  Dirichlet boundary condition. By extending the
Gel'fand--Levitan method of inverse spectral theory  from the scalar case to the matrix case, we are able to 
obtain the transformations for all relevant quantities when the bound states of any multiplicities are removed or added or when the bound-state multiplicities are
modified. The aforementioned relevant quantities include the potential, the regular solution, the Jost solution, the Jost matrix, the scattering
matrix, and the boundary matrices used in the description of the selfadjoint  boundary condition.

The use of a matrix-valued Schr\"odinger equation finds its origin in the beginning of the development of quantum mechanics.
The matrix Schr\"odinger operator plays a central role to take into account the internal structure, such as spin, 
of quantum mechanical particles and also to consider collections of particles such as nuclei, atoms, and molecules. An important example of a matrix-valued Schr\"odinger equation
is the Pauli equation describing particles of spin equal to one half.
We refer the reader to \cite{CS1989} and \cite{LL1989} for the applications
of  matrix-valued Schr\"odinger equation in quantum mechanics and in nuclear physics.
In the nonlinear case of a matrix-valued time-dependent Schr\"odinger equation, we refer the reader to \cite{NW2023} for the analysis of
long-time asymptotics of the corresponding solutions. We remark that
matrix-valued Schr\"odinger equations have important applications in the context of quantum graphs \cite{BK2006,K2024} in various fields including nanotechnology, quantum wires, and quantum computing.
In the analysis of quantum graphs the use of a general selfadjoint boundary condition, rather than a purely Dirichlet boundary condition, 
is crucial. For example, a star graph, which is a quantum graph with one vertex and a finite number of semi-infinite edges
meeting at the vertex, is mathematically described by a matrix Schr\"odinger operator with the general boundary condition.
For a further discussion on applications of matrix-valued Schr\"odinger equations with the general
boundary condition, we refer the reader to \cite{AW2021}.

A simple but useful method to remove or add bound states for scalar Sturm--Liouville  equations has a long history dating back to Jacobi \cite{J1837} and  Darboux \cite{D1882},  
with later contributions by  Crum \cite{C1955}, Deift \cite{D1978}, Deift and Trubowitz \cite{DT1979}, and Schminke \cite{S1979}. This method is often called 
the Darboux transformation, the Darboux--Crum transformation, or the single commutation method.
It involves factoring the corresponding Sturm--Liouville operator into two noncommuting operators.  However, it has the disadvantage that it requires the positivity of certain eigensolutions to the relevant Sturm--Liouville problem, and consequently it can only be applied to remove the bound state with the lowest energy 
 or to add a bound state 
at an energy below the existing bound-state energy levels. An alternate method is the double commutation method \cite{G1993,GT1996,GST1996,S2003}. In this 
alternate method a second commutation is used, which allows to remove or add bound states at any energy in the spectrum
of the Sturm--Liouville operator. It is mentioned in \cite{GT1996} that the second commutation method can be found in the seminal work of Gel'fand
and Levitan \cite{GL1955}. The literature on methods
to remove or add bound states is extensive. Besides the references we have quoted already, we mention
\cite{EK1982,L1987,M2011,AU2022} and the references therein. Such methods have applications in others areas as well. For example, 
they are used to obtain explicit solution to integrable evolution equations \cite{L1987,MS1991,M2011,RS2002}.

The Gel'fand--Levitan method was developed by Gel'fand and Levitan in \cite{GL1955} (see also \cite{LG1964}) to recover a scalar-valued potential 
when the corresponding spectral function is used as the input data set. 
The method was extended to the matrix case by Newton and Jost \cite{NJ1955} by considering the matrix-valued
Schr\"odinger operator when the boundary condition is purely Dirichlet. In our paper we extend the Gel'fand--Levitan method in the matrix case
by allowing the general selfadjoint boundary condition, which includes the purely Dirichlet boundary condition as a specific case.
We provide the transformations of all relevant quantities by using the Gel'fand--Levitan method for the matrix-valued  Schr\"odinger 
operator with the general selfadjoint boundary condition. We construct the bound-state normalization matrices by taking into consideration the multiplicities
of the bound states. We refer to the constructed normalization matrices as the Gel'fand--Levitan normalization matrices. This is contrast with the normalization
matrices in the Marchenko method of inverse scattering, which is a method used to recover the matrix-valued potential from the scattering data set consisting of the scattering
matrix, the bound-state energies, and the bound-state normalization matrices. We refer to the normalization matrices used in the Marchenko method
as the Marchenko normalization matrices. The Marchenko method for the half-line matrix-valued Schr\"odinger 
operator with the purely Dirichlet boundary condition was developed by Agranovich and Marchenko in their seminal work
\cite{AM1963}. We refer the reader to our own monograph \cite{AW2021} for the Marchenko method for the half-line matrix-valued Schr\"odinger 
operator with the general selfadjoint boundary condition.
For some of our results in this paper, the existence of the first moment of the potential is not needed. On the other hand, for some other results we impose
a stronger moment condition on the potential than just the first moment condition.

The analysis of the half-line Schr\"odinger operator in the matrix case is
much more involved than in the scalar case. The Gel'fand--Levitan normalization constants
in the scalar case are defined in a straightforward manner by normalizing the
regular solution, whereas in the matrix case
the definition of a Gel'fand--Levitan normalization matrix is more elaborate. This is because each bound state is simple in the scalar case whereas a bound state
in the matrix case may have a multiplicity.
Furthermore, in the
scalar case a boundary condition is either Dirichlet or non-Dirichlet, whereas in the
matrix case a boundary condition may be purely Dirichlet, purely non-Dirichlet, or
a mixture of Dirichlet and non-Dirichlet. 
Under a transformation to remove or to add bound states in the scalar case, a Dirichlet boundary condition remains Dirichlet and 
a non-Dirichlet boundary condition remains non-Dirichlet, although the boundary
parameter changes in the latter case.
In the matrix case, we observe the following. The boundary matrix $A$ appearing in \eqref{2.5} and \eqref{2.10} remains unchanged
under the transformation. In the matrix case, the purely Dirichlet case occurs when $A=0,$ and hence
a purely Dirichlet case remains purely Dirichlet under the  transformation, which resembles the scalar case.
In the matrix case, the boundary matrix $B$ appearing in \eqref{2.5} and \eqref{2.10} changes under the transformation
unless $A=0,$ and this also has a resemblance in the scalar case.
In our paper we obtain the  transformations to remove or add bound states or to decrease or increase the multiplicity of a bound state
when the potential is matrix valued and the boundary condition is the general selfadjoint boundary
condition, and this is done without restricting ourselves to the purely Dirichlet case or to the purely
non-Dirichlet case.
We deal with bound-state multiplicities by using the appropriate orthogonal projections onto the kernels of the Jost matrix and of the matrix adjoint of  the Jost matrix,
respectively,
evaluated at the bound-state energies. Our analysis also involves the introduction of
a dependency matrix at each bound state and the use of a Moore--Penrose inverse.
The bound-state dependency matrix connects the normalized matrix-valued bound-state wavefunction 
used in the Gel'fand-Levitan method to
the normalized matrix-valued bound-state wavefunction 
used in the Marchenko method.

In the half-line scalar case, when  remove or add a bound state, the Gel'fand--Levitan normalization constants for the remaining bound states are not affected
\cite{L1987,M2011}.
We observe the same in the matrix case even with the general selfadjoint boundary condition. In other words, we find that the
Gel'fand--Levitan normalization matrices remain unchanged for the
bound states not affected by the transformation. We also observe that, 
when we change the multiplicity of
a bound state,
the Gel'fand--Levitan normalization matrices for the remaining bound states  are not affected. 
As a result, in the half-line case, a removal or addition of
a bound state or a decrease or  increase in the multiplicity of a bound state does not affect the remaining bound states. One important consequence of
this result is that the mathematical formulation of  the transformation does not require that the bound-states 
must be ordered in an increasing or a decreasing manner with respect to the bound-state energies.

Our paper is organized as follows.
In Section~\ref{section2} we present the preliminaries
needed for our analysis, by introducing the boundary matrices associated with a selfadjoint boundary condition, the Jost solution
and the regular solution to the matrix Schr\"odinger equation on the half line, the Jost matrix, the scattering matrix, and the 
physical solution.
In Section~\ref{section3} we deal with the bound states, by introducing the orthogonal projections onto the kernels of
the Jost matrix and of the adjoint of the Jost matrix, respectively, at each bound state. We then introduce
the Gel'fand--Levitan normalization matrices and
the matrix-valued Gel'fand--Levitan normalized bound-state solutions.
This is done in a similar manner the Marchenko normalization matrices and the matrix-valued
Marchenko normalized bound state solutions
are introduced \cite{AW2021} for the corresponding matrix Schr\"odinger operator with the general selfadjoint boundary condition.
In Section~\ref{section4} we introduce the dependency matrices that relate the Marchenko normalized bound-state solutions to the  Gel'fand--Levitan normalized bound-state solutions, and we present the basic properties of the dependency matrices.
We show that the matrix adjoint of a dependency matrix is the same as its Moore--Penrose inverse, and we also
show that the dependency matrices help us treat bound states with  multiplicities not only
in a rather simple and mathematically elegant way but also in
a practical way suitable for computations.
In Section~\ref{section5} we introduce the Gel'fand--Levitan system of integral equations
and Parseval's equality related to the matrix-valued Schr\"odinger operator
with the general selfadjoint boundary condition.
We show how the input to the Gel'fand--Levitan system is related
to the change in the spectral measure. We also describe how the change in the potential,
the change in the regular solution, and the change in the boundary matrices are expressed in terms of the
solution to the Gel'fand--Levitan system of integral equations.
In Section~\ref{section6} we present the  transformations of all relevant quantities when a bound state is removed. In particular, we show how the regular solution,  the potential, the boundary matrices, the Jost matrix, the scattering matrix, the Jost solution, and other relevant
quantities change.  We also show that the Gel'fand--Levitan normalization matrices for the bound states
not involved in the  transformation remain unchanged. Further, we show that the continuous part of the spectral measure does not change.  In Section~\ref{section7} we present the transformations of all relevant
quantities when the multiplicity of a bound state is decreased, without totally removing the bound state. Note that this can only happen in the matrix case since in the scalar case the multiplicity of the bound states is one. The proofs of the results in this section are similar to the proofs in Section~\ref{section6}, and hence 
the results in Section~\ref{section7} are stated without their proofs. 
In Section~\ref{section8} we present the  transformations of all relevant quantities when a new bound state is added
to the discrete spectrum of the operator.
In (IV.1.30) of \cite{CS1989} it is claimed that, in the scalar case with the Dirichlet boundary condition when adding a bound state with the energy $-\kappa^2,$  the increment of the potential given by $\tilde V(x)- V(x)$ decays as $F e^{-2\kappa x}$ as $x\to+\infty,$ where $F$
is a nonzero constant. In Example~\ref{example8.9}, by expressing the relevant quantities and their transformations
explicitly, we show that the result of (IV.1.30) in \cite{CS1989} is incorrect.
Finally, in Section~\ref{section9} we present the transformations of all relevant quantities
when we increase the multiplicity of a bound state
without changing the rest of
the spectrum. The results in Section~\ref{section9} are presented without their proofs
because those proofs are similar to the proofs in Section~\ref{section8}.

\section{The half-line matrix Schr\"odinger equation}
\label{section2}

In this section we present the preliminaries needed for the
remaining sections. In particular, we introduce the boundary matrices $A$ and $B,$ the Jost solution
$f(k,x),$ the regular solution $\varphi(k,x),$ the Jost matrix $J(k),$ the scattering matrix $S(k),$ and the physical solution $\Psi(k,x)$
associated with the half-line matrix Schr\"odinger operator. We refer the reader to \cite{AW2021}
for further elaborations on the material presented here.

Consider the half-line Schr\"odinger equation
\begin{equation}
\label{2.1}
-\psi''+V(x)\,\psi=k^2\psi,\qquad x\in\bR^+,
\end{equation}
where the prime denotes the $x$-derivative, we let $\bR^+$ denote the interval $(0,+\infty),$ and the
potential $V$ is an $n\times n$ selfadjoint matrix-valued function
of $x,$ with $n$ being a fixed positive integer.
The wavefunction
$\psi$ is either an $n\times n$ matrix or
a column vector with
$n$ components.
Note that
the selfadjointness is expressed as
\begin{equation}
\label{2.2}
V(x)^\dagger=V(x),\qquad x\in\bR^+,
\end{equation}
with the dagger denoting the matrix adjoint
(matrix transpose and complex conjugation).
We always assume that the potential $V$ is integrable, i.e. we have
\begin{equation}\label{2.3}
\int_0^\infty dx\,\,|V(x)|<+\infty,
\end{equation}
where $|V(x)|$ denotes the
operator norm of the $n\times n$ matrix $V(x).$
At times we impose further restrictions on the potential $V$ so that it satisfies
\begin{equation}\label{2.4}
\int_0^\infty dx\, (1+x)^\epsilon   \,|V(x)|<+\infty,
\end{equation}
for some fixed positive constant $\epsilon.$ We say that $V$
belongs to $L^1_\epsilon(\bR^+)$ if \eqref{2.4} is satisfied.
Since all matrix norms are equivalent
in finite-dimensional vector spaces, any other matrix norm can be used in \eqref{2.3} and \eqref{2.4} instead of 
the operator norm. The adjectives selfadjoint and hermitian are used interchangeably. 
We use $\bCp$ to denote the upper-half complex plane, use
$\bR$ for the real axis, and let $\bCpb:=\bCp\cup\bR.$
We use both $V$ and $V(x)$ to refer to the potential $V,$ where the latter usage emphasizes
the dependence on $x.$

We obtain a selfadjoint Schr\"odinger operator on the half line
by supplementing \eqref{2.1} with the
general selfadjoint boundary condition at $x=0,$
which can be written as
\begin{equation}
\label{2.5}
-B^\dagger \psi(0)+A^\dagger \psi'(0)=0,
\end{equation}
where the constant $n\times n$ matrices $A$ and $B$ satisfy
\begin{equation}
\label{2.6}
-B^\dagger A+A^\dagger B=0,\end{equation}
\begin{equation}
\label{2.7}
A^\dagger A+B^\dagger B>0,\end{equation}
and we refer to $A$ and $B$ as the boundary matrices.
We use the positivity and the positive definiteness for a matrix interchangeably.
Thus, the positivity for the matrix appearing on the left-hand side of \eqref{2.7}
is the same as the positive definiteness of that matrix. We recall that an $n\times n$
matrix is positive if and only if
all its eigenvalues are positive.
Similarly, a matrix is nonnegative if and only if its eigenvalues
are real and nonnegative.
We remark that that \eqref{2.7} is equivalent \cite{AW2021} to
having the matrix-rank relation given by
\begin{equation*}
\text{rank}\begin{bmatrix} A\\
B\end{bmatrix}=n,
\end{equation*}
where $\begin{bmatrix} A\\
B\end{bmatrix}$ is the $2n\times n$ matrix formed from $A$ and $B.$
Multiplying the boundary matrices $A$ and $B$ on the right by an invertible
$n\times n$ matrix $T$ does not change \eqref{2.5}. Thus, even though the boundary condition \eqref{2.5} itself is uniquely
determined by
the boundary-matrix pair $(A,B),$ 
the matrix pair $(AT,BT)$ with any invertible $T$
also yields the
same boundary condition \eqref{2.5}. In fact, this is the only freedom we have in choosing the boundary matrices $A$ and $B.$ It is possible to 
use a single constant $n\times n$ matrix to describe the boundary condition \eqref{2.5} instead of a pair of constant
$n\times n$ matrices. However, it is more advantageous \cite{AW2021} to describe \eqref{2.5} by using the pair of
boundary matrices $A$ and $B$ appearing in the initial conditions \eqref{2.10}.
By using $A=0$ and $B=I$ in \eqref{2.10} we obtain the Dirichlet boundary condition
$-\psi(0)=0,$ which can also be written as $\psi(0)=0$ by letting $T=-I.$
At times we use the term purely Dirichlet for emphasis to make a distinction
with the boundary condition consisting of a mixture of Dirichlet and non-Dirichlet components.
The notation $0$ is used for the scalar zero, the zero vector with $n$ components, and
the $n\times n$ zero matrix, depending on the context, We use $I$ to denote the $n\times n$ identity matrix.

When the potential $V$ satisfies \eqref{2.2} and \eqref{2.3}, there are two relevant particular $n\times n$ matrix solutions to
\eqref{2.1}. One of them is the Jost solution
$f(k,x),$ and it satisfies the spacial
asymptotics
\begin{equation}
\label{2.9}
f(k,x)=e^{ikx}\left[ I+o(1)\right], \quad f'(k,x)= e^{ikx} \left[ik I+o(1)\right],
\qquad x\to +\infty,
\end{equation}
for $ k \in \overline{\mathbb C^+}\setminus\{0\}.$
In fact the Jost solution $f(k,x)$ is defined as the unique solution satisfying
 the first asymptotics in \eqref{2.9}, and the second asymptotics there is automatically satisfied. If $ V\in L^1_1(\mathbb R^+)$ 
then the Jost solution exists also at $k=0.$
The second relevant particular solution to \eqref{2.1} is the regular solution $\varphi(k,x)$
satisfying
the initial conditions
\begin{equation}
\label{2.10}
\varphi(k,0)=A,\quad \varphi'(k,0)=B,
\end{equation}
where $A$ and $B$ are the boundary matrices appearing in
\eqref{2.5}. When the potential $V$ satisfies \eqref{2.2} and \eqref{2.3}, 
the regular solution $\varphi(k,x)$ is entire in 
$k\in\mathbb C$ for each fixed $x\in[0,+\infty).$
We remark that in general the Jost solution $f(k,x)$ does not satisfy
the boundary condition \eqref{2.5}. However, the regular solution
$\varphi(k,x)$  indeed satisfies \eqref{2.5} at any $k\in\mathbb C.$

When the potential $V$ satisfies \eqref{2.2} and \eqref{2.3}, the Jost matrix $J(k)$ associated with
\eqref{2.1} and \eqref{2.5} is defined as
\begin{equation}
\label{2.11}
J(k):=f(-k^*,0)^\dagger\,B-f'(-k^*,0)^\dagger\,A,\qquad k\in\mathbb R\setminus \{0\},
\end{equation}
where the asterisk denotes complex conjugation. If we further have $V \in L^1_1(\mathbb R^+)$ then
$J(k)$ exists also at $k=0.$ The $n\times n$ matrix-valued
 $J(k)$ has an extension from $ k\in\mathbb R\setminus \{0\}$
 to $k\in\bCp,$ and in fact the asterisk in
\eqref{2.11} is used to indicate how that
extension occurs.
The scattering matrix $S(k)$ associated with \eqref{2.1} and \eqref{2.5}
is defined as
\begin{equation}
\label{2.12}
S(k):=-J(-k)\,J(k)^{-1},\qquad k\in\bR\setminus\{0\}.
\end{equation}
If $V \in L^1_1(\mathbb R^+)$ then the domain of the scattering matrix, by continuity, extends also to $k=0.$
We remark that the definition of the scattering matrix given in
\eqref{2.12} differs by a minus sign from the definition used in \cite{AM1963,NJ1955} 
in the matrix case with the purely Dirichlet boundary condition as well as the traditional definition
used in the scalar case \cite{CS1989,L1987,M2011} with the Dirichlet boundary condition.
On the other hand, in the scalar case with a non-Dirichlet boundary condition a minus sign
as in \eqref{2.12} is used \cite{L1987,M2011} to define the scattering matrix.
We refer the reader to \cite{AW2018,AW2021} for the elaboration on the issue of defining the
scattering matrix for the half-line matrix Schr\"odinger operator
and the troubling consequences of defining the scattering matrix differently
in the Dirichlet and non-Dirichlet cases.

The physical solution $\Psi(k,x)$ associated with
\eqref{2.1} and \eqref{2.5} is defined as
\begin{equation}
\label{2.13}
\Psi(k,x):=f(-k,x)+f(k,x)\,S(k).
\end{equation}
For each fixed $x\in\mathbf R^+$ the physical
solution $\Psi(k,x)$ has a meromorphic extension from $k\in\mathbb R\setminus\{0\} $  (or from $k \in \bR$ if $V\in L^1_1(\mathbb R^+)$) to
$k\in\mathbb C^+$  and the regular solution $\varphi(k,x)$ is entire in $k.$ The regular solution
can be expressed in terms of the physical solution, the Jost solution, 
and the Jost matrix in various equivalent forms
such as
\begin{equation}
\label{2.14}
\varphi(k,x)=-\ds\frac{1}{2ik}\,\Psi(k,x)\,J(k),\qquad k\in\bCpb\setminus\{0\},
\end{equation}
\begin{equation}
\label{2.15}
\varphi(k,x)=\ds\frac{1}{2ik}\left[f(k,x)\,J(-k)-f(-k,x)\,J(k)\right],
\qquad k\in\mathbb R\setminus\{0\}.
\end{equation}
The regular solution is even in $k,$ i.e. we have
\begin{equation*}
\varphi(-k,x)=\varphi(k,x),
\qquad k\in\mathbb C.
\end{equation*}

\section{The bound states of the Schr\"odinger operator}
\label{section3}

In this section we introduce the relevant material related to the
bound states of the half-line matrix Schr\"odinger operator.
For the treatment of bound states and the Marchenko normalization matrices, we refer the reader to \cite{AW2021}. In this paper, since we use the Gel'fand--Levitan
method to establish the relevant transformations
for the half-line matrix Schr\"odinger operator, we introduce
the Gel'fand--Levitan normalization matrices by using a procedure
similar to the construction of the Marchenko normalization matrices.

A bound-state solution
corresponds to a square-integrable column-vector solution
to \eqref{2.1} satisfying the boundary condition \eqref{2.5}. We let
$\lambda:=k^2.$
As stated in Theorem~3.11.1 of \cite{AW2021}, there are no bound states when $\lambda >0$ but  a bound state at $\lambda=0$ is possible. The  bound states when $\lambda<0$ occur at the $k$-values on the positive imaginary axis 
of the complex $k$-plane 
corresponding to the zeros of $\det[J(k)].$
We use $\det[J(k)]$ to denote the determinant of the
Jost matrix $J(k)$ defined in \eqref{2.11}. If there are an infinite number of bound states when $\lambda<0,$ the corresponding bound-state
$k$-values only accumulate at $k=0.$ As indicated in Theorem~4.3.3 of \cite{AW2021}, the selfadjoint Schr\"odinger operator related to 
\eqref{2.1} and \eqref{2.5} has no singular continuous spectrum and  its absolutely continuous spectrum is the $\lambda$-interval $[0,+\infty).$  
If $V \in L^1_1(\mathbb R^+),$ then from
Theorem~3.11.1 of \cite{AW2021} we conclude that the zero value of $\lambda$ does not correspond to
a bound state and that the number of bound states
when $\lambda<0$ is finite. 
We assume that there are
$N$ zeros of $\det[J(k)]$ occurring when $k=i\kappa_j$ for
$1\le j\le N,$ with
$\kappa_j$ being distinct positive
constants. Note that $N$ is a nonnegative integer or it is equal to $+\infty,$ and
it represents the number of negative-energy bound states
without counting multiplicities.
In case $N=0,$ there are no negative-energy bound states.
Thus, $\det[J(i\kappa_j)]=0$ and we use
$m_j$ to denote the number of linearly independent vectors
in $\bC^n$ spanning the kernel $\text{\rm{Ker}}[J(i\kappa_j)].$
We remark that $m_j$ is a positive integer
satisfying $1\le m_j\le n,$ and it corresponds to
the multiplicity of the bound state at $k=i\kappa_j.$
Thus, the total number of negative-energy bound states
including the multiplicities, denoted by $\mathcal N,$ is given by
\begin{equation*}
\mathcal N:=\ds\sum_{j=1}^N m_j.
\end{equation*}
Note that $\mathcal N$ can be zero, a positive integer, or infinity.
At $k=i\kappa_j$ there are $2n$ linearly independent column-vector
solutions to \eqref{2.1}. We remark that $n$ of such solutions are
exponentially decreasing as $x\to+\infty$ and hence square integrable in $x,$ and $n$ of
such solutions are exponentially increasing as $x\to+\infty$ and hence
not square integrable in $x.$ For further elaborations on this issue,
we refer the reader to Propositions~3.2.1, 3.2.2, and 3.2.3 of \cite{AW2021}.
 It turns out that, by using linear combinations of $n$ square-integrable solutions, we can construct exactly $m_j$ 
 square-integrable column-vector solutions that also satisfy the boundary condition.
Thus, the multiplicity $m_j$ of the bound state at $k=i\kappa_j$
cannot exceed $n.$  Consequently, there are exactly $m_j$ linearly
 independent square-integrable column-vector
 solutions satisfying the boundary condition, and those solutions make
 up the bound states.

The bound-state solutions to
\eqref{2.1} can be analyzed either via
the Marchenko theory or the Gel'fand--Levitan theory.
In the Marchenko theory the normalization
matrices for the bound states are obtained by
normalizing the Jost solution $f(k,x)$ at the bound states,
whereas in the Gel'fand--Levitan theory
the normalization
matrices for the bound states are obtained by
normalizing the regular solution $\varphi(k,x)$ at the bound states.
Thus, the Marchenko normalization matrices and
the corresponding normalized matrix solutions
to the Schr\"odinger equation
are different from
the Gel'fand--Levitan normalization matrices and
the corresponding normalized matrix solutions.
We use 
$M_j$ and $\Psi_j(x)$ to denote the $n\times n$
Marchenko normalization matrix and
the corresponding $n\times n$ normalized matrix solution
at the bound state with $k=i\kappa_j,$ respectively.
Similarly, we use
$C_j$ and $\Phi_j(x)$ to denote the
Gel'fand--Levitan normalization matrix and
the corresponding normalized matrix solution
at the bound state with $k=i\kappa_j,$ respectively.
In this section we
introduce the $n\times n$ matrices $C_j$ and $\Phi_j(x)$
in a similar way $M_j$ and $\Psi_j(x)$ are
defined \cite{AW2021}.

We remark that the $n\times n$ matrices
$J(i\kappa_j)$ and $J(i\kappa_j)^\dagger$ have the same rank, and their
common rank is equal to $n-m_j.$ We see that
when $n=1$ we must have $m_j=1,$ in which case
$J(i\kappa_j)$ and $J(i\kappa_j)^\dagger$ are both equal to the
scalar zero. If $n\ge 2,$
then the kernels
$\text{\rm{Ker}}[J(i\kappa_j)]$ and
$\text{\rm{Ker}}[J(i\kappa_j)^\dagger]$
are in general different, but they have the same
dimension equal to $m_j.$
We use $Q_j$ and $P_j$ to denote the orthogonal projections
onto $\text{\rm{Ker}}[J(i\kappa_j)]$ and
$\text{\rm{Ker}}[J(i\kappa_j)^\dagger],$ respectively.
By the definition of an orthogonal projection, we have
\begin{equation}
\label{3.2}
P_j^2=P_j^\dagger=P_j,\quad Q_j^2=Q_j^\dagger=Q_j.
\end{equation}

Let us use $\mathbf M^+$ to denote the Moore--Penrose inverse
of a matrix
$\mathbf M.$ Even though the Moore--Penrose inverse is uniquely defined for any matrix \cite{BG2003,CM2009}, in
our paper we only deal with Moore--Penrose inverses of square matrices.
As indicated in Definitions~1.12 and 1.13 and Theorem~1.1.1 of \cite{CM2009},
the matrix $\mathbf M^+$ is the Moore--Penrose inverse of the matrix $\mathbf M$ if the matrices
$\mathbf M$ and $\mathbf M^+$ satisfy the four equalities given by
\begin{equation}
\label{3.3}
\mathbf M \mathbf M^+ \mathbf M=\mathbf M,\quad \mathbf M^+ \mathbf M \mathbf M^+=\mathbf M^+,\quad 
 (\mathbf M \mathbf M^+)^\dagger=\mathbf M \mathbf M^+,\quad
(\mathbf M^+\mathbf  M)^\dagger=\mathbf M^+ \mathbf M.
\end{equation}

We list the basic properties of the Moore--Penrose inverse of a square matrix in the following theorem.
We refer the reader to
Theorems~1.1.1 and 1.2.1 of
\cite{CM2009} for a proof.

\begin{theorem}
\label{theorem3.1} The Moore--Penrose inverse $\mathbf M^+$ 
of an $n\times n$ matrix $\mathbf M$ has the following properties:

\begin{enumerate}

\item[\text{\rm(a)}] The matrix $\mathbf M^+$ exists and is unique.

\item[\text{\rm(b)}] If $\mathbf M$ is invertible, then $\mathbf M^+=\mathbf M^{-1}.$

\item[\text{\rm(c)}] The Moore--Penrose inverse commutes with the matrix  adjoint, i.e. we have
\begin{equation}
\label{3.4}
 (\mathbf M^+)^\dagger=(\mathbf M^\dagger)^+,
\end{equation}
where we recall that the dagger denotes the matrix adjoint.

\item[\text{\rm(d)}] The matrix $\mathbf M^+ \mathbf M$ is the  orthogonal projection onto $\text{\rm{Ran}}[\mathbf M^+],$ where we use
$\text{\rm{Ran}}$ to denote the range.

\item[\text{\rm(e)}] The matrix $\mathbf M \mathbf M^+$ is the orthogonal projection onto $\text{\rm{Ran}}[\mathbf M].$

\item[\text{\rm(f)}] The Moore--Penrose inverse of $\mathbf M^\dagger \mathbf M$ is expressed in terms of the Moore--Penrose inverse $\mathbf M^+$ as
in the case of the standard matrix inversion, i.e. we have
\begin{equation}
\label{3.5}
(\mathbf M^\dagger \mathbf M)^+=\mathbf M^+(\mathbf M^\dagger)^+.
\end{equation}

\item[\text{\rm(g)}] For any $n\times n$ unitary matrices $\mathbf U_1$ and $\mathbf U_2,$ the matrices 
$\mathbf M$ and $\mathbf M^+$ satisfy
\begin{equation*}
\left( \mathbf U_1 \,\mathbf M \,\mathbf U_2\right)^+=\mathbf  U_2^\dagger\, \mathbf M^+\, \mathbf U_1^\dagger.
\end{equation*}

\end{enumerate}
\end{theorem}

At times we need to use the $x$-derivative of the Moore--Penrose inverse of a square matrix when that square matrix 
contains the independent variable $x$ in a special way. The following proposition provides the appropriate derivative formula
that we need in our paper in that special case.

\begin{proposition}
\label{proposition3.2}
Assume that $\mathbb C^n$ has the direct-sum decomposition into any two of its subspaces $\mathbb D_1$ and $\mathbb D_2$ as
\begin{equation*}
\mathbb C^n= \mathbb D_1 \oplus \mathbb D_2.
\end{equation*}
Let $\mathbf M(x)$ be an $n\times n$ matrix-valued differentiable function of $x\in\mathbb R^+.$
Suppose that $\mathbf M(x)$ has the direct-sum decomposition given by
\begin{equation*}
\mathbf M(x)=\mathbf M_1(x)\oplus \mathbf M_2, \qquad x \in \mathbb R^+,
\end{equation*}
where the components $\mathbf M_1(x)$ and $\mathbf M_2$ belong to $\mathbb D_1$ and $\mathbb D_2,$ respectively,
in such a way that 
$\mathbf M_1(x)$ is invertible on $\mathbb D_1$ and that $\mathbf M_2$ is independent of $x.$
We have the following:

\begin{enumerate}

\item[\text{\rm(a)}] 

The Moore--Penrose inverse $\mathbf M(x)^+$ has the direct-sum decomposition given by
\begin{equation}\label{3.9}
\mathbf M(x)^+=\mathbf M_1(x)^{-1}\oplus \mathbf M_2^+,
\end{equation}
where $\mathbf M_1(x)^{-1}$denotes the inverse of the matrix $\mathbf M_1(x)$ on $\mathbb D_1$ and
$\mathbf M_2^+$ denotes the Moore--Penrose inverse of the constant matrix $\mathbf M_2.$

\item[\text{\rm(b)}] 
The $x$-derivative $[\mathbf M(x)^+]'$ is expressed in terms of 
$\mathbf M(x)^+$ and the $x$-derivative $\mathbf M'(x)$ as 
\begin{equation}
\label{3.10}
[\mathbf M(x)^+]'=-\mathbf M(x)^+\, \mathbf M'(x)\, \mathbf M(x)^+,
\end{equation}
which can also be written as
\begin{equation}
\label{3.11}
[\mathbf M(x)^+]'= -\left[\mathbf M_1(x)^{-1}\oplus \mathbf M_2^+\right]\left[\mathbf M_1'(x)\oplus 0 \right]  
\left[\mathbf M_1(x)^{-1}\oplus \mathbf M_2^+\right].
\end{equation}

\end{enumerate}

\end{proposition}

\begin{proof} The proof of \eqref{3.9} is obtained by verifying that the right-hand side of \eqref{3.9} satisfies 
the four equalities in \eqref{3.3} used in the definition of the Moore--Penrose inverse of 
$\mathbf M(x).$ Hence, the proof of (a) is complete. The proof of (b) is obtained as follows.
Since $\mathbf M_1(x)$ is invertible on $\mathbb D_1,$ we use the standard derivative formula for the matrix inverse of
$\mathbf M_1(x)$ and get
 \begin{equation*}
[\mathbf M_1(x)^{-1}]'= -\mathbf M_1(x)^{-1}\, \mathbf M_1'(x)\,\mathbf  M_1(x)^{-1}.
\end{equation*}
Furthermore, since $\mathbf M_2$ is independent of $x,$ its
$x$-derivative is equal to the zero matrix $0.$
Thus, \eqref{3.11} holds and also its equivalent \eqref{3.10} is valid.
Hence, the proof of (b) is also complete.
\end{proof}

Comparing \eqref{3.2} and \eqref{3.3}
we see that $P_j$ is equal to its own Moore--Penrose inverse and that $Q_j$ is also equal to its own Moore--Penrose inverse, i.e. we have
\begin{equation*}
P_j^+=P_j,\quad Q_j^+=Q_j,\qquad 1\le j\le N.
\end{equation*}
We can explicitly construct $P_j$ and $Q_j$ in terms of
any
orthonormal basis $\{v_j^{(l)}\}_{l=1}^{m_j}$ for $\text{\rm{Ker}}[J(i\kappa_j)^\dagger]$
and any
orthonormal basis $\{w_j^{(l)}\}_{l=1}^{m_j}$ for $\text{\rm{Ker}}[J(i\kappa_j)],$ respectively.
We recall that $m_j$ corresponds to the common dimension of
those two kernels.
Thus, we have
\begin{equation}
\label{3.14}
P_j=\ds\sum_{j=1}^{m_j} v_j^{(l)}\,(v^{(l)})^\dagger,\quad
Q_j=\ds\sum_{j=1}^{m_j} w_j^{(l)}\, (w_j^{(l)})^\dagger.
\end{equation}

The kernels $\text{\rm{Ker}}[J(i\kappa_j)^\dagger]$ and $\text{\rm{Ker}}[J(i\kappa_j)]$
play a key role in the analysis of bound states for the matrix Schr\"odinger operator
associated with \eqref{2.1} and \eqref{2.5}. 
Those two kernels satisfy
\begin{equation}
\label{3.15}
\bC^n=\text{\rm{Ker}}[J(i\kappa_j)^\dagger]\oplus \left(\text{\rm{Ker}}[J(i\kappa_j)^\dagger]\right)^\perp,\qquad 1\le j\le N,
\end{equation}
\begin{equation}
\label{3.16}
\bC^n=\text{\rm{Ker}}[J(i\kappa_j)]\oplus \left(\text{\rm{Ker}}[J(i\kappa_j)]\right)^\perp,\qquad 1\le j\le N,
\end{equation}
\begin{equation*}
\left(\text{\rm{Ker}}[J(i\kappa_j)]\right)^\perp=\text{\rm{Ran}}[J(i\kappa_j)^\dagger],\qquad 1\le j\le N,
\end{equation*}
\begin{equation*}
\left(\text{\rm{Ker}}[J(i\kappa_j)^\dagger]\right)^\perp=\text{\rm{Ran}}[J(i\kappa_j)],\qquad 1\le j\le N,
\end{equation*}
where we use the
superscript $\perp$ to denote the orthogonal complement.

Let us consider the Schr\"odinger equation \eqref{2.1} at the bound state $k=i\kappa_j,$ which is given by
\begin{equation} 
\label{3.19}
-\psi''+V(x)\,\psi=-\kappa_j^2\,\psi,\qquad x\in\mathbb R^+,\quad 1\le j\le N.
\end{equation}
We remark that in general the solution 
$f(i\kappa_j,x)$ to \eqref{3.19} does not satisfy the boundary condition \eqref{2.5}. On the other hand, $f(i\kappa_j,x)$ is exponentially decaying
asymptotically, i.e. we have
\begin{equation}
\label{3.20}
f(i\kappa_j,x)=e^{-\kappa_j x}\left[I+o(1)\right],\qquad x\to+\infty,\quad 1\le j\le N,
\end{equation}
which is seen from \eqref{2.9}. Hence, for large $x$-values
the $n$ columns of
$f(i\kappa_j,x)$ are linearly independent exponentially
decaying solutions to \eqref{3.19} although those columns do not, in general, 
satisfy the boundary condition \eqref{2.5}.
It is desirable to construct an $n\times n$ matrix solution to
\eqref{3.19} satisfying the boundary condition \eqref{2.5}. 
Such a matrix solution, denoted by $\Psi_j(x),$ is obtained by constructing \cite{AW2021}
the $n\times n$ 
Marchenko normalization matrix $M_j$ and by letting
\begin{equation}
\label{3.21}
\Psi_j(x):=f(i\kappa_j,x)\,M_j,\qquad 1\le j\le N,
\end{equation}
where we refer to the $n\times n$ matrix
$\Psi_j(x)$ as the Marchenko
normalized  bound-state solution
to the Schr\"odinger equation \eqref{3.19}
at $k=i\kappa_j.$ The Marchenko normalization matrices $M_j$ are constructed
 in such a way that the matrices $\Psi_j(x)$ defined in \eqref{3.21} for $1\le j\le N$
satisfy the boundary condition \eqref{2.5} and also satisfy
the respective normalization and orthogonalization conditions given by
\begin{equation}
\label{3.22}
\int_0^\infty dx\,\Psi_j(x)^\dagger\,\Psi_j(x)=P_j,\qquad 1\le j\le N,
\end{equation}
\begin{equation}
\label{3.23}
\int_0^\infty dx\,\Psi_j(x)^\dagger\,\Psi_l(x)=0,\qquad j\ne l,
\end{equation} 
where we have $1\le j\le N$ and $1\le l\le N.$
We will see
that each $M_j$  is a nonnegative matrix of rank $m_j.$
In other words, exactly $m_j$ eigenvalues of the matrix
$M_j$ are positive and the remaining $n-m_j$ eigenvalues are zero. Hence, 
although each column of the matrix $\Psi_j(x)$ is a solution to
\eqref{3.19} and satisfies \eqref{2.5}, among the $n$ columns of
$M_j$ we only have exactly $m_j$ linearly independent column-vector solutions to \eqref{3.19} satisfying
the boundary condition \eqref{2.5}. We observe that any nontrivial
square-integrable column-vector solution to \eqref{3.19} satisfying \eqref{2.5} can be written as $\Psi_j(x)\,\nu$ for a constant nonzero
column vector $\nu \in \mathbb C^n.$ 

For the construction of the $n\times n$ Marchenko normalization matrix $M_j$, we refer the reader
to \cite{AM1963} when the boundary condition \eqref{2.5} is purely Dirichlet and 
to Section~3.11 of \cite{AW2021} when the boundary condition \eqref{2.5} is the general selfadjoint
boundary condition. We remark that the construction of the matrix $M_j$
in Section~3.11 of \cite{AW2021} is given by assuming that the potential
$V$ in \eqref{3.19} belongs to $L^1_1(\mathbb R^+).$ However, the same proof applies by replacing
that assumption by the weaker assumption $V \in L^1(\mathbb R^+).$ 

The construction of $M_j$ can be summarized as follows.
With the help of the orthogonal projection matrix
$P_j$ and the Jost solution $f(i\kappa_j,x),$ we form the $n\times n$ matrix $\mathbf A_j$ defined as
\begin{equation}
\label{3.24}
\mathbf A_j:=\int_0^\infty dx\,P_j\,f(i\kappa_j,x)^\dagger\,f(i\kappa_j,x)\,P_j,\qquad 1\le j\le N.
\end{equation}
We then define the $n\times n$ matrix $\mathbf B_j$ as
\begin{equation}
\label{3.25}
\mathbf B_j:=I-P_j+\mathbf A_j,\qquad 1\le j\le N.
\end{equation}
From \eqref{3.2}, \eqref{3.24}, and \eqref{3.25} it follows that $\mathbf A_j$ and $\mathbf B_j$ are both
hermitian. Furthermore, it can be shown \cite{AW2021} that $\mathbf B_j$ is a positive matrix.
Thus, there exists a unique positive hermitian matrix $\mathbf B_j^{1/2}$ so that
\begin{equation}
\label{3.26}\mathbf B_j^{1/2}\,\mathbf B_j^{1/2}=\mathbf B_j,\qquad 1\le j\le N.
\end{equation}
Since $\mathbf B_j^{1/2}$ is positive, its inverse exists and is denoted by
$\mathbf B_j^{-1/2}.$ It can be shown \cite{AW2021} that each of
$\mathbf B_j,$ $\mathbf B_j^{1/2},$ and $\mathbf B_j^{-1/2}$ commutes with
$P_j.$ Then, the Marchenko normalization matrix $M_j$ is defined
as
\begin{equation}
\label{3.27}M_j:=\mathbf B_j^{-1/2}\,P_j,\qquad 1\le j\le N.
\end{equation}

We construct the $n\times n$ matrix-valued Gel'fand--Levitan normalized bound-state solutions
in a manner similar to the construction of  the $n\times n$ matrix-valued Marchenko normalized bound-state solutions.
In analogy with \eqref{3.21} we construct the $n\times n$ 
Gel'fand--Levitan normalization matrix $C_j$ and let
\begin{equation}
\label{3.28}
\Phi_j(x):=\varphi(i\kappa_j,x)\,C_j,\qquad 1\le j\le N,
\end{equation}
where we refer to the $n\times n$ matrix
$\Phi_j(x)$ as the Gel'fand--Levitan normalized bound-state solution
 to the Schr\"odinger equation \eqref{3.19} at $k=i\kappa_j.$
In the construction of the matrix $C_j$ presented below, we observe that
each $C_j$  is a nonnegative matrix of rank $m_j.$ We construct the matrices $C_j$ for $1\le j\le N$ in such a way that the matrices
$\Phi_j(x)$ are square integrable and satisfy the respective normalization and orthogonalization
conditions given by
\begin{equation}
\label{3.29}
\int_0^\infty dx\,\Phi_j(x)^\dagger\,\Phi_j(x)=Q_j,\qquad 1\le j\le N,
\end{equation}
\begin{equation}
\label{3.30}
\int_0^\infty dx\,\Phi_j(x)^\dagger\,\Phi_l(x)=0,\qquad j\ne l,
\end{equation} 
which are the analogs of \eqref{3.22} and \eqref{3.23}, respectively.

Although the regular solution $\varphi(i\kappa_j,x)$ to \eqref{3.19} satisfies the boundary condition \eqref{2.5}, unless we have $m_j=n$ that solution becomes unbounded
as $x\to+\infty,$ and we have \cite{AW2021}
\begin{equation}
\label{3.31}
\varphi(i\kappa_j,x)=O(e^{\kappa_j x}),\qquad x\to+\infty.
\end{equation}
We recall that $m_j$ corresponds to the multiplicity of the bound state at $k=i\kappa_j.$ Thus, in the scalar case, i.e. when $n=1$
the regular solution $\varphi(i\kappa_j,x)$ to \eqref{3.19} is exponentially decaying like
$O(e^{-\kappa_j x})$ as $x\to+\infty.$ In the matrix case, i.e. when $n\ge 2,$
that regular solution in general cannot be square integrable in $x\in\mathbb R^+$ unless $m_j=n.$
We know from Propositions~3.2.1, 3.2.2, and 3.2.3 of \cite{AW2021} that, in general, each column of $\varphi(i\kappa_j,x)$ consists of two components
behaving like $O(e^{\kappa_j x})$ and $O(e^{-\kappa_j x}),$
respectively, as $x\to +\infty.$ The normalization matrix 
$C_j$ allows us to form some appropriate linear combinations of the
columns of $\varphi(i\kappa_j,x)$ in such a way that those linear combinations
behave like $O(e^{-\kappa_j x})$ as $x\to+\infty,$ and that there are
exactly $m_j$ such
linearly independent  combinations of  the columns. 
Unless $m_j=n,$ it is possible that none of the $n$ columns of
$\varphi(i\kappa_j,x)$ behave like
$O(e^{-\kappa_j x})$ as $x\to+\infty.$
On the other hand, we have
\begin{equation*}
\Phi_j(x)=O(e^{-\kappa_j x}),\qquad x\to+\infty.
\end{equation*}
Furthermore, $\Phi_j(x)$ satisfies the boundary condition \eqref{2.5}
as a consequence of \eqref{2.10}. We remark that any nontrivial
square-integrable column-vector solution to \eqref{3.19} satisfying \eqref{2.5} can be written as $\Phi_j(x)\,\nu,$ 
where $\nu$ is a constant nonzero
column vector in $\mathbb C^n.$

We present the construction of the $n\times n$ Gel'fand--Levitan normalization matrix $C_j$ in a similar way
to the construction of the Marchenko normalization matrix $M_j$ outlined in \eqref{3.24}--\eqref{3.27}. With the help of the orthogonal projection matrix
$Q_j$ and the regular solution $\varphi(i\kappa_j,x)$ to \eqref{3.19}, we form the $n\times n$ matrix $\mathbf G_j$ as
\begin{equation}
\label{3.33}\mathbf G_j:=\int_0^\infty dx\,Q_j\,\varphi(i\kappa_j,x)^\dagger\,\varphi(i\kappa_j,x)\,Q_j,\qquad 1\le j\le N.
\end{equation}
It follows from Theorem~3.11.1(e) of \cite{AW2021} that the integral on the right-hand side of \eqref{3.33} is finite.
Next, we introduce the matrix $\mathbf H_j$ as
\begin{equation}
\label{3.34}
\mathbf H_j:=I-Q_j+\mathbf G_j,\qquad 1\le j\le N.
\end{equation}
From \eqref{3.2}, \eqref{3.33}, and \eqref{3.34} it follows that the matrices $\mathbf G_j$ and $\mathbf H_j$ are both
hermitian. Moreover, $\mathbf H_j$ is a positive matrix, and the proof for this is similar to the proof \cite{AW2021}
that $\mathbf B_j$ defined in \eqref{3.25} is a positive matrix, which is given in the proof of Proposition~3.11.10 of \cite{AW2021}.
Thus, there exists a unique positive hermitian matrix $\mathbf H_j^{1/2}$ so that
\begin{equation}
\label{3.35}\mathbf H_j^{1/2}\,\mathbf H_j^{1/2}=\mathbf H_j,\qquad 1\le j\le N.
\end{equation}
Since the matrix $\mathbf H_j^{1/2}$ is positive, its inverse exists and is denoted by
$\mathbf H_j^{-1/2}.$ We note that each of
$\mathbf H_j,$ $\mathbf H_j^{1/2},$ and $\mathbf H_j^{-1/2}$ commutes with
$Q_j,$ and the proof for this is similar to the proof \cite{AW2021} regarding the commutation of each of
$\mathbf B_j,$ $\mathbf B_j^{1/2},$ and $\mathbf B_j^{-1/2}$ with
$P_j,$ which is given in the proof of Proposition~3.11.10 of \cite{AW2021}.
The Gel'fand--Levitan normalization matrix $C_j$ is defined
as
\begin{equation}
\label{3.36}
C_j:=\mathbf H_j^{-1/2}\,Q_j,\qquad 1\le j\le N,
\end{equation}
which is analogous to the definition of the Marchenko normalization matrix $M_j$ given in \eqref{3.27}.
 From \eqref{3.36} it follows that $C_j$ is hermitian and nonnegative and has rank
 equal to $m_j,$ which is the same as the rank of $Q_j.$

We recall that $\varphi(i\kappa_j,x)$ for each $j$ with $1\le j\le N$ satisfies the boundary condition \eqref{2.5}. It follows from the second equality in 
\eqref{3.14} and Theorem~3.11.1 of \cite{AW2021} that each column of $\varphi (i\kappa_j,x) Q_j$ decreases exponentially as $O(e^{-\kappa_j x})$
as $x\to +\infty.$ Thus, the integral on the right-hand side of \eqref{3.33} is finite.
The normalization condition
\eqref{3.29} follows from \eqref{3.33}, \eqref{3.34}, and the fact that $Q_j$  is an orthogonal projection that commutes with $\mathbf H_j,$ 
$\mathbf H_j^{1/2},$ and $\mathbf H^{-1/2}.$ The orthogonality in \eqref{3.30} can be established
in a similar way \eqref{3.23} is established in the proof of Proposition~5.10.5 in \cite{AW2021}. 
Since $\varphi(i\kappa_j,x)\,C_j$ satisfies the boundary condition \eqref{2.5} and decays exponentially as 
$O(e^{-\kappa_j x})$ as $x\to+\infty,$ it is square integrable and hence the column vector $\varphi(i\kappa_j,x)\,C_j\, v,$ where 
$v$ is any column vector in $\mathbb C^n,$ is a bound state 
with the energy $-\kappa_j^2.$ From Theorem~3.11.1 in \cite{AW2021} we know that any bound state with the energy $-\kappa_j^2$ is of the form 
$\varphi(i\kappa_j,x)\,C_j \,v,$ where  $v$ is a column vector in $\mathbb C^n.$ Moreover, 
Theorem~3.11.1 of \cite{AW2021} indicates that there are precisely $m_j$ linearly independent such column vectors, and hence $m_j$ is the multiplicity of the bound state 
at $k=i\kappa_j.$ We recall that $m_j$ is the positive integer equal to the dimension of the kernel of $J(i\kappa_j).$ 

We remark that the integral in \eqref{3.24} remains finite even in the absence
of $P_j$ there, i.e. we have
\begin{equation*}
\int_0^\infty dx\,f(i\kappa_j,x)^\dagger\,f(i\kappa_j,x)<+\infty,
\end{equation*}
which can be proved with the help of \eqref{3.20}. On the other hand,
if we remove $Q_j$ in the integral in \eqref{3.33} then the resulting integral
is in general not convergent when $n\ge 2,$ i.e. the integral
\begin{equation*}
\int_0^\infty dx\,\varphi(i\kappa_j,x)^\dagger\,\varphi(i\kappa_j,x),
\end{equation*}
may not be finite because of \eqref{3.31}, unless $m_j=n,$ in which case we have $Q_j=I.$

With the help of \eqref{3.27} and \eqref{3.36}, let us write \eqref{3.21} and \eqref{3.28} as
\begin{equation}
\label{3.39}\Psi_j(x)=f(i\kappa_j,x)\,\mathbf B_j^{-1/2}\,P_j,\quad
\Phi_j(x)=\varphi(i\kappa_j,x)\,\mathbf H_j^{-1/2}\,Q_j,
\end{equation}
where we recall that $\mathbf B_j^{-1/2}$ and $P_j$ commute with each other and that
$\mathbf H_j^{-1/2}$ and $Q_j$ commute with each other.
We can write the orthonormality relations \eqref{3.22}, \eqref{3.23}, \eqref{3.29}, and \eqref{3.30} in a compact form as
\begin{equation*}
\ds\int_0^\infty dx\, \Psi_j(x)^\dagger\,\Psi_j(x)=\delta_{jl}\,P_l,\quad
\ds\int_0^\infty dx\, \Phi_j(x)^\dagger\,\Phi_j(x)=\delta_{jl}\,Q_l,
\end{equation*}
where $\delta_{jl}$ denotes the Kronecker delta,
and we have $1\le j\le N$ and $1\le l\le N.$

\section{The bound-state dependency matrices}
\label{section4}

In this section we assume that
the potential $V(x)$ appearing in \eqref{2.1} satisfies \eqref{2.2} and \eqref{2.3}
and that the corresponding matrix-valued Schr\"odinger operator is associated
with the selfadjoint boundary condition described in \eqref{2.5}--\eqref{2.7}.
Although the Schr\"odinger operator is selfadjoint, the number of bound states
may be infinite. For each bound state, we introduce the concept of the
bound-state dependency matrix, which relates the Gel'fand--Levitan normalized bound-state solution $\Phi_j(x)$ 
to the Marchenko normalized bound-state solution $\Psi_j(x).$
We describe the basic
properties of the bound-state dependency matrix. In the scalar case, i.e. when $n=1,$ it is appropriate to refer to dependency matrices 
as dependency constants.
We refer the reader to \cite{AE2022,AEU2023,AEU2023a} for the use of dependency constants in some scattering problems associated with nonselfadjoint systems,
especially when the bound states are not necessarily simple, and for the use of the time evolutions of dependency constants
arising in the analysis of integrable systems.

We recall that a bound-state solution at $k=i\kappa_j$ is a square-integrable column-vector 
solution to \eqref{3.19} satisfying the boundary condition \eqref{2.5}.
We know that there are exactly $m_j$ linearly independent
bound-state solutions at $k=i\kappa_j,$ where $m_j$ is a positive integer corresponding to the 
multiplicity of the bound state and we have $1\le m_j\le n.$
At the bound state with $k=i\kappa_j,$
 from \eqref{3.21} we know that the $n\times n$ Marchenko normalized bound-state solution
$\Psi_j(x)$ has $m_j$ linearly independent columns and that each column of $\Psi_j(x)$
satisfies \eqref{3.19} and the boundary condition \eqref{2.5}.
Similarly, we know from \eqref{3.28} that that the $n\times n$ Gel'fand--Levitan normalized bound-state solution
$\Phi_j(x)$ has $m_j$ linearly independent columns and that each column of $\Phi_j(x)$
satisfies \eqref{3.19} and the boundary condition \eqref{2.5}.
Thus, there exists a constant $n\times n$ matrix relating the columns of 
the matrix $\Psi_j(x)$ to the columns of the matrix $\Phi_j(x).$
Let us introduce the $n\times n$ matrix $D_j$
in such a way that
\begin{equation}
\label{4.1}
\Phi_j(x)=\Psi_j(x)\,D_j, \qquad 1\le j\le N,
\end{equation}
where we recall that the number $N$ of bound states may be infinite.
Since $D_j$ is an $n\times n$ matrix relating 
the two $n\times n$  matrices with equal ranks $m_j,$ 
the matrix $D_j$ in \eqref{4.1} cannot be unique
unless we impose a further restriction on it.
Without loss of generality, we can choose that
additional restriction as
\begin{equation}
\label{4.2}
D_j=P_j\,D_j,
\end{equation}
where we recall that the $n\times n$ matrix $P_j$ is the orthogonal projection onto $\text{\rm{Ker}}[J(i\kappa_j)^\dagger]$ and
appearing in \eqref{3.2}. We refer to 
$D_j$ appearing in \eqref{4.1} and \eqref{4.2} as the dependency matrix 
associated with the bound state at $k=i\kappa_j.$

We construct the dependency matrix $D_j$ satisfying \eqref{4.1} and \eqref{4.2} as follows. 
From Theorem~3.11.1 of \cite{AW2021} we know that
there exists a bijection $\alpha\mapsto\beta$ from the kernel of $J(i\kappa_j)$ onto the kernel of $J(i\kappa_j)^\dagger$ 
in such a way that
\begin{equation}\label{4.3}
\varphi(i\kappa_j,x)\,\alpha= f(i\kappa_j,x)\,\beta,\qquad 1\le j\le N.
\end{equation}
Let  $\{w_j^{(l)}\}_{l=1}^{m_j}$ be an orthonormal basis for the subspace $\text{\rm{Ker}}[J(i\kappa_j)],$ and let $\beta_j^{(l)}$ be the
column vector in the subspace $\text{Ker}[J(i\kappa_j)^\dagger]$ corresponding to
the column vector $w_j^{(l)}$ under the aforementioned bijection. 
Thus, \eqref{4.3} implies that for $1\le j\le N$ we have
\begin{equation*}
\varphi(i\kappa_j,x)\,w_j^{(l)}= f(i\kappa_j,x)\,\beta_j^{(l)},\qquad 1\le l\le m_j,
\end{equation*}
from which we get
\begin{equation}\label{4.5}
\varphi(i\kappa_j,x)\,w_j^{(l)}\,(w_j^{(l)})^\dagger= f(i\kappa_j,x)\,\beta_j^{(l)}\,(w_j^{(l)})^\dagger,\qquad 1\le l\le m_j.
\end{equation}
Summing over $l$ in \eqref{4.5}, with the help of the second equality of \eqref{3.14}, we obtain
\begin{equation}\label{4.6}
\varphi(i\kappa_j,x)\,Q_j= f(i\kappa_j,x)\ds\sum_{l=1}^{m_j}\beta_j^{(l)}\,(w_j^{(l)})^\dagger,\qquad 1\le j\le N.
\end{equation}
Postmultiplying each side of \eqref{4.6} by the invertible matrix $\mathbf H_j^{-1/2}$ appearing in \eqref{3.36}, we get
\begin{equation}\label{4.7}
\varphi(i\kappa_j,x)\,Q_j\,\mathbf H_j^{-1/2}= f(i\kappa_j,x)\ds\sum_{l=1}^{m_j}\beta_j^{(l)}\,(w_j^{(l)})^\dagger\,\mathbf H_j^{-1/2},\qquad 1\le j\le N.
\end{equation}
We already know that the matrices $Q_j$ and $\mathbf H_j^{-1/2}$ commute, and hence \eqref{4.7} is equivalent to
\begin{equation}\label{4.8}
\varphi(i\kappa_j,x)\,\mathbf H_j^{-1/2}\,Q_j= f(i\kappa_j,x)\ds\sum_{l=1}^{m_j}\beta_j^{(l)}\,(w_j^{(l)})^\dagger\,\mathbf H_j^{-1/2},\qquad 1\le j\le N.
\end{equation}
Postmultiplying both sides of \eqref{4.8} by $Q_j$ and then using $Q_j^2=Q_j,$ we obtain
\begin{equation}\label{4.9}
\varphi(i\kappa_j,x)\,\mathbf H_j^{-1/2}\,Q_j= f(i\kappa_j,x)\ds\sum_{l=1}^{m_j}\beta_j^{(l)}\,(w_j^{(l)})^\dagger\,\mathbf H_j^{-1/2}\,Q_j,\qquad 1\le j\le N.
\end{equation}
Using the second equality of \eqref{3.39} on the left-hand side of \eqref{4.9}, we get
\begin{equation}\label{4.10}
\Phi_j(x)= f(i\kappa_j,x)\ds\sum_{l=1}^{m_j}\beta_j^{(l)}\,(w_j^{(l)})^\dagger\,\mathbf H_j^{-1/2}\,Q_j,\qquad 1\le j\le N.
\end{equation}
Hence, from \eqref{4.1} and \eqref{4.10} we have
\begin{equation}\label{4.11}
\Psi_j(x)\,D_j= f(i\kappa_j,x)\ds\sum_{l=1}^{m_j}\beta_j^{(l)}\,(w_j^{(l)})^\dagger\,\mathbf H_j^{-1/2}\,Q_j,\qquad 1\le j\le N.
\end{equation}
Using \eqref{3.21} and \eqref{3.27} on the left-hand side of \eqref{4.11}, we obtain
\begin{equation}\label{4.12}
 f(i\kappa_j,x)\,\mathbf B_j^{-1/2} \,P_j\,D_j= f(i\kappa_j,x)\ds\sum_{l=1}^{m_j}\beta_j^{(l)}\,(w_j^{(l)})^\dagger\,\mathbf H_j^{-1/2}\,Q_j,\qquad 1\le j\le N.
\end{equation}
From \eqref{3.20} we see that the matrix $f(i\kappa_j,x)$ is invertible for large $x$-values. Hence, \eqref{4.12} yields
\begin{equation}\label{4.13}
\mathbf B_j^{-1/2} \,P_j\,D_j= \ds\sum_{l=1}^{m_j}\beta_j^{(l)}\,(w_j^{(l)})^\dagger\,\mathbf H_j^{-1/2}\,Q_j,\qquad 1\le j\le N.
\end{equation}
We recall that the matrix $\mathbf B_j^{-1/2}$ is invertible and its inverse is given by  $\mathbf B_j^{1/2}.$
Premultiplying both sides of \eqref{4.13} by
$P_j\,\mathbf B_j^{1/2}$ and then using \eqref{4.2} in the resulting equation, we obtain
\begin{equation}\label{4.14}
D_j= P_j \,\mathbf B_j^{1/2} \sum_{l=1}^{m_j} \beta_j^{(l)}  \,(w_j^{(l)})^\dagger \,\mathbf H_j^{-1/2}\, Q_j,\qquad 1\le j\le N.
\end{equation}

The next theorem shows that the dependency matrix
$D_j$ satisfying \eqref{4.1} and \eqref{4.2} is unique, and it presents some basic properties of $D_j.$
We remark that \eqref{4.2} alone cannot uniquely determine
$D_j$ because any constant scalar multiple of $D_j$
satisfying \eqref{4.2} also
satisfies \eqref{4.2}.

 \begin{theorem}
\label{theorem4.1}
Assume that
the potential $V$ appearing in \eqref{2.1} satisfies \eqref{2.2} and \eqref{2.3}
and that the corresponding matrix-valued Schr\"odinger operator is associated
with the selfadjoint boundary condition described in \eqref{2.5}--\eqref{2.7}.
Let $D_j$ be the $n\times n$ dependency matrix satisfying
\eqref{4.1} and \eqref{4.2}
associated with the bound state at $k=i\kappa_j.$ We have the following:

\begin{enumerate}
\item[\text{\rm(a)}] The matrix $D_j$ is uniquely determined.

\item[\text{\rm(b)}] The matrix
 $D_j$ satisfies
\begin{equation}
\label{4.15}
D_j=D_j\,Q_j,
\end{equation}
where we recall that the $n\times n$ matrix $Q_j$ is the orthogonal projection
onto $\text{\rm{Ker}}[J(i\kappa_j)]$ and
satisfies the second set of equalities in \eqref{3.2}.

 \item[\text{\rm(c)}]
The matrix product $D_j^\dagger D_j$
is an orthogonal projection, and we have
\begin{equation}
\label{4.16}
D_j^\dagger D_j=Q_j
.\end{equation}

\item[\text{\rm(d)}]
The matrix product $D_j D_j^\dagger$
is also an orthogonal projection, and we have
\begin{equation}
\label{4.17}
D_j D_j^\dagger=P_j,\end{equation}
where we recall that the $n\times n$ matrix $P_j$ is the orthogonal projection
onto $\text{\rm{Ker}}[J(i\kappa_j)^\dagger]$ and
satisfies the first set of equalities in \eqref{3.2}.

\item[\text{\rm(e)}]
The matrices $D_j$ and $D_j^\dagger$
each have rank $m_j,$ which is the same as the common rank
of the orthogonal projection matrices $P_j$ and $Q_j.$

\item[\text{\rm(f)}] The Moore--Penrose inverse of the matrix $D_j$ is
equal to $D_j^\dagger,$ i.e. we have
\begin{equation}
\label{4.18}
D_j^+=D_j^\dagger.\end{equation}

\item[\text{\rm(g)}] The $n\times n$ Marchenko normalized bound-state solution $\Psi_j(x)$ is related to
the $n\times n$ Gel'fand--Levitan normalized bound-state solution $\Phi_j(x)$ as
\begin{equation}
\label{4.19}
\Psi_j(x)=\Phi_j(x) D_j^\dagger.\end{equation}

\item[\text{\rm(h)}] The matrix $D_j$ satisfies
\begin{equation}
\label{4.20}
D_j D_j^\dagger D_j=D_j,\quad
D_j^\dagger D_j D_j^\dagger=D_j^\dagger.
\end{equation}

\end{enumerate}
\end{theorem}

\begin{proof}
The proof of (a) can be given as follows. The existence of $D_j$ follows from the explicit expression
given in \eqref{4.14}. To prove its uniqueness, 
let us assume that $D_j$ and $\hat D_j$ both satisfy
\eqref{4.1} and \eqref{4.2} so that we have
\begin{equation}
\label{4.21}
\Phi_j(x)=\Psi_j(x)\,D_j,\quad \Phi_j(x)=\Psi_j(x)\,\hat D_j,\end{equation}
\begin{equation}
\label{4.22}
D_j=P_j\,D_j,\quad \hat D_j=P_j\,\hat D_j.\end{equation}
 From \eqref{4.21} we get
\begin{equation}
\label{4.23}
\Psi_j(x)\,(D_j-\hat D_j)=0.\end{equation}
Using the first equality of \eqref{3.39} in \eqref{4.23} we obtain
\begin{equation}
\label{4.24}
f(i\kappa_j,x)\,\mathbf B_j^{-1/2}\,P_j\,(D_j-\hat D_j)=0.\end{equation}
We already know that the matrix 
$\mathbf B_j^{-1/2}$ is invertible.
Furthermore, from \eqref{3.20} it follows that
the matrix $f(i\kappa_j,x)$ for large $x$-values is also invertible.
Hence, \eqref{4.24} yields
\begin{equation}
\label{4.25}
P_j\,(D_j-\hat D_j)=0.\end{equation}
Using \eqref{4.22} in \eqref{4.25} we get
\begin{equation*}
D_j-\hat D_j=0,
\end{equation*}
which yields
$D_j=\hat D_j$ and hence $D_j$ is uniquely determined by \eqref{4.1} and \eqref{4.2}. Thus, the proof of (a) is complete.
For the proof of (b), we proceed as follows. We remark that \eqref{4.14} and the uniqueness proved in (a) imply that (b) holds. Note that (b) can alternatively be proved as follows. 
From the second set of
equalities in \eqref{3.2} and the second equality in \eqref{3.39}, we see that
$\Phi_j(x)\,Q_j=\Phi_j(x).$ Hence, postmultiplying
\eqref{4.1} by $Q_j$ we obtain
\begin{equation}
\label{4.27}
\Phi_j(x)=\Psi_j(x)\,D_j\,Q_j.
\end{equation}
Subtracting \eqref{4.27} from \eqref{4.1} we get
\begin{equation}
\label{4.28}
\Psi_j(x)\left(D_j-D_j\,Q_j\right)=0.\end{equation}
Proceeding as in the steps involving \eqref{4.23}--\eqref{4.25}, from \eqref{4.28} we have
\begin{equation}
\label{4.29}
P_j\left(D_j-D_j\,Q_j\right)=0.\end{equation}
Finally, using \eqref{4.2} in \eqref{4.29} we obtain
\eqref{4.15}, which completes the proof of (b).
For the proof of (c) we proceed as follows.
Using \eqref{4.1} in \eqref{3.29}, with the help of
\eqref{3.22} we get
\begin{equation}
\label{4.30}
D_j^\dagger\,P_j\,D_j=Q_j.\end{equation}
Using \eqref{4.2} in \eqref{4.30} we
obtain \eqref{4.16}. Since $Q_j$ is an orthogonal
projection, we see that the proof of (c) is complete.
To prove (d) we first need to show that the matrix product $D_j D_j^\dagger$ is an orthogonal projection, i.e.
we need to show that $D_j D_j^\dagger$ is selfadjoint and it satisfies
\begin{equation}
\label{4.31}
\left(D_j D_j^\dagger\right)\left(D_j D_j^\dagger\right)=D_j D_j^\dagger.\end{equation}
The selfadjointness of $D_j D_j^\dagger$ is self evident, and
\eqref{4.31} is verified by using the set of equalities
\begin{equation*}
D_j\left(D_j^\dagger D_j\right)D_j^\dagger=D_j\,Q_j D_j^\dagger=D_j D_j^\dagger,
\end{equation*}
which are obtained by using \eqref{4.16} and then \eqref{4.15}.
Since the matrices $D_j D_j^\dagger$ and $P_j$ are orthogonal projections and have the same domain $\mathbb C^n,$ in order to
complete the proof of (d) we need to show that the matrices
$D_j D_j^\dagger$ and $P_j$ have the same range.
For this we proceed as follows. With the help of \eqref{4.2} we see
that 
\begin{equation*}
P_j D_j D_j^\dagger=D_j D_j^\dagger,
\end{equation*}
from which we obtain the range relations given by
\begin{equation*}
\text{\rm{Ran}}[D_j D_j^\dagger]=\text{\rm{Ran}}[P_j D_j D_j^\dagger]\subset \text{\rm{Ran}}[P_j].
\end{equation*}
We have
\begin{equation*}
\text{\rm{Ker}}[D_j D_j^\dagger]=\text{\rm{Ker}}[D_j^\dagger],
\end{equation*}
and hence by using the rank-nullity theorem we conclude that
\begin{equation}
\label{4.36}
\text{\rm{rank}}[D_j\,D_j^\dagger]=\text{\rm{rank}}[D_j^\dagger].\end{equation}
Since
$D_j$ and $D_j^\dagger$ have the same rank, from
\eqref{4.36} we get
\begin{equation}
\label{4.37}
\text{\rm{rank}}[D_j D_j^\dagger]=\text{\rm{rank}}[D_j].\end{equation}
 From \eqref{4.37} we conclude that
 the two matrices $D_j D_j^\dagger$
 and $D_j$ have the same nullity.
On the other hand, \eqref{4.16} implies that
\begin{equation}
\label{4.38}
\text{\rm{Ker}}[Q_j]=\text{\rm{Ker}}[D_j^\dagger D_j]=\text{\rm{Ker}}[D_j],\end{equation}
and hence we conclude
that the rank of $D_j$ is the same as the rank of
$Q_j,$ which is also equal to
the rank of $P_j.$ Then, from
\eqref{4.37} and \eqref{4.38} we get
\begin{equation}
\label{4.39}
\text{\rm{rank}}[D_j D_j^\dagger]=\text{\rm{rank}}[P_j],\end{equation}
and hence \eqref{4.39} implies \eqref{4.17}.
Thus, the proof of (d) is complete.
The proof of (e) directly follows from
\eqref{4.37}--\eqref{4.39}.
The proof of (f) is obtained by using
$\mathbf M=D_j$ and $\mathbf M^+=D_j^\dagger$
in \eqref{3.3} and by verifying that the four equalities listed
there are satisfied.
Note that the last two equalities in \eqref{3.3} are satisfied
because we have $\mathbf M^+=\mathbf M^\dagger$ there. The first two equalities in
\eqref{3.3} are also satisfied, and this can be proved as follows.
Using \eqref{4.16} and then \eqref{4.15}, we have
\begin{equation}
\label{4.40}
D_j (D_j^\dagger D_j)=D_j Q_j=D_j,\end{equation}
and hence the first equality in \eqref{3.3} is satisfied.
Similarly, using \eqref{4.15}, \eqref{4.16}, and the property
$Q_j^\dagger=Q_j$ listed in \eqref{3.2}, we obtain
\begin{equation}
\label{4.41}
(D_j^\dagger D_j)\,D_j^\dagger=Q_j D_j^\dagger=(D_j Q_j)^\dagger=D_j^\dagger,\end{equation}
which verifies the second equality in
\eqref{3.3}. Consequently, with the help of the uniqueness of
the Moore--Penrose inverse, we confirm \eqref{4.18}. 
This  completes the proof of (f). Note that \eqref{4.19} follows from \eqref{4.1}, \eqref{4.17}, and the first equality in \eqref{3.39}, which completes
the proof of (g).
Note also that
\eqref{4.20} directly follows from \eqref{4.40} and \eqref{4.41},
and hence the proof of (h) is also complete.
\end{proof}

In the previous theorem we have established the uniqueness of the dependency matrix $D_j$ satisfying \eqref{4.1}
and \eqref{4.2}. We have presented an explicit formula 
for $D_j$ in \eqref{4.14}. We would like to show that $D_j$ can also be expressed explicitly in terms of the Jost solution
$f(i\kappa_j,x),$ the Marchenko normalization matrix $M_j$ appearing in \eqref{3.27}, and the Gel'fand--Levitan bound-state normalized
solution defined in \eqref{3.28} as
\begin{equation}
\label{4.42}
D_j= M^+_j\, f(i\kappa_j,x)^{-1}\,   \Phi_j(x),
\end{equation}
where we evaluate the right-hand side at any $x$-value at which the matrix $f(i\kappa_j,x)$ is invertible and 
we recall that  $M^+_j$ denotes the Moore--Penrose inverse of $M_j.$ The invertibility of  $f(i\kappa_j,x)$
is assured for large $x$-values, as seen from \eqref{3.20}. 
We can establish
\eqref{4.42} as follows. As seen from \eqref{3.2} the orthogonal projection matrix $P_j$ satisfies $P_j^2=P_j,$
and we already know that $P_j$ commutes with $\mathbf B_j^{-1/2}$ appearing in \eqref{3.27}. Hence, from \eqref{3.27} we have
\begin{equation}
\label{4.43}
M_j= P_j M_j P_j,
\end{equation}
Using \eqref{3.15} we obtain the decomposition for $M_j$ given by
\begin{equation}
\label{4.44}
M_j=[M_j]_1\oplus 0,
\end{equation}
where $[M_j]_1$ is defined as the restriction of $M_j$ to $\textrm{Ker}[ J(i\kappa_j)^\dagger].$ 
From \eqref{3.27} it follows that the matrix $[M_j]_1$ is invertible in $\textrm{Ker}[ J(i\kappa_j)^\dagger].$ Thus, from \eqref{3.9} we obtain
\begin{equation}\label{4.45}
M_j^+=\left([M_j]_1\right)^{-1}\oplus 0.
\end{equation}
Using \eqref{3.21} and the invertibility of $f(i\kappa_j,x)$ assured for large $x,$ from \eqref{4.1} we get
\begin{equation}\label{4.46}
M^+_j M_j D_j= M^+_j f(i\kappa_j,x)^{-1}  \Phi_j(x).
\end{equation}
From \eqref{4.44} and \eqref{4.45} we have
\begin{equation}
\label{4.47}
M^+_j M_j= I\oplus 0,
\end{equation}
which indicates that 
\begin{equation}
\label{4.48}
M^+_j M_j=P_j.
\end{equation}
Using \eqref{4.48} in \eqref{4.46} we get
\begin{equation}\label{4.49}
P_j D_j= M^+_j f(i\kappa_j,x)^{-1}\,  \Phi_j(x),
\end{equation}
and using \eqref{4.2} on the right-hand side of \eqref{4.49}
we obtain \eqref{4.42}.

If the matrix $f(i\kappa_j,x)$ is not invertible
at a particular $x$-value
but the matrix
$f'(i\kappa_j,x)$ is invertible there, then we can derive an expression for
$D_j$ in an analogous manner we have obtained \eqref{4.42}, and we get
\begin{equation}
\label{4.50}
D_j= M^+_j f'(i\kappa_j,x)^{-1} \,  \Phi'_j(x).
\end{equation}
The derivation of \eqref{4.50} is straightforward and can be established by taking the
$x$-derivative of both sides of \eqref{4.1} and then by using the analogs of the 
steps involving \eqref{4.43}--\eqref{4.49}.

The dependency matrix $D_j$ appearing in
\eqref{4.1} and \eqref{4.2} and its adjoint $D_j^\dagger$ are related to
certain partial isometries. In order to see this we proceed as follows.
We recall \cite{K1980} that a linear map $L$ on $\mathbb C^n$ is an isometry if 
$|Lv|=|v|$ for all vectors $v\in\mathbb C^n,$ where we use $|\cdot|$ to denote the standard norm on
$\mathbb C^n.$ On the other hand, $L$ is a partial isometry if 
$|Lv|=|v|$ for all vectors $v$ in $\left(\text{\rm{Ker}}[L]\right)^\perp,$
where we recall that the superscript $\perp$ denotes the orthogonal complement.
It is known \cite{K1980} that the linear map
$L$ is a partial isometry if and only if
$\langle Lv,Lw\rangle=\langle v,w\rangle$ for all vectors $v$ and
$w$ in $\left(\text{\rm{Ker}}[L]\right)^\perp,$
where we use $\langle\cdot,\cdot\rangle$ to denote the standard inner product on
$\mathbb C^n.$
If the linear map $L$
is a partial isometry, then the subspace
$\left(\text{\rm{Ker}}[L]\right)^\perp$
is called the initial subspace of $\bC^n$ and the subspace
$\text{\rm{Ran}}[L]$ is called the final subspace of $\bC^n.$
It is also known \cite{K1980} that a linear map $L$ is partially isometric if and only if 
the adjoint linear map $L^\dagger$ is a partial isometry.

The relevant properties of a partial isometry are known \cite{K1980},
and we list them
in the following theorem without a proof.

\begin{theorem}
\label{theorem4.2}
Let $L$
be a linear map on $\bC^n.$ Then, the following are equivalent:

\begin{enumerate}
\item[\text{\rm(a)}]
The map $L$ is a partial isometry,
with the understanding that the initial subspace is
$\left(\text{\rm{Ker}}[L]\right)^\perp$ and the final subspace is
$\text{\rm{Ran}}[L].$

\item[\text{\rm(b)}]
The adjoint map $L^\dagger$ is a partial isometry, with
the understanding that the initial subspace is
$\left(\text{\rm{Ker}}[L^\dagger]\right)^\perp$ and the final subspace is
$\text{\rm{Ran}}[L^\dagger].$

\item[\text{\rm(c)}]
The map 
$L L^\dagger$ is the orthogonal projection
onto $\text{\rm{Ran}}[L].$

\item[\text{\rm(d)}]
The map $L^\dagger L$ is the orthogonal projection onto
$\left(\text{\rm{Ker}}[L]\right)^\perp$.

\item[\text{\rm(e)}]
We have $L L^\dagger L=L.$

\item[\text{\rm(f)}]
We have $L^\dagger L L^\dagger=L^\dagger.$

\end{enumerate}
\end{theorem}

In the next theorem, with the help of Theorem~\ref{theorem4.2},
we  relate the bound-state dependency matrix $D_j$ appearing in
\eqref{4.1} and \eqref{4.2} to a certain partial
isometry and to a bijection from
$\text{\rm{Ker}}[J(i\kappa_j)]$
 onto $\text{\rm{Ker}}[J(i\kappa_j)^\dagger].$

  \begin{theorem}
\label{theorem4.3}
Assume that the potential $V$ appearing in \eqref{2.1} satisfies \eqref{2.2} and \eqref{2.3} and that
the corresponding half-line matrix Schr\"odinger operator is associated with the
selfadjoint boundary condition described in \eqref{2.5}--\eqref{2.7}.
Let $D_j$ be the $n\times n$ dependency matrix satisfying
\eqref{4.1} and \eqref{4.2} associated with
the bound state at $k=i\kappa_j.$
We have the following:

\begin{enumerate}
\item[\text{\rm(a)}]
 The matrix $D_j$ is a partial isometry
with the initial subspace $\text{\rm{Ker}}[J(i\kappa_j)]$ and the
final subspace $\text{\rm{Ker}}[J(i\kappa_j)^\dagger].$

\item[\text{\rm(b)}] The adjoint matrix $D_j^\dagger$ is a partial isometry
with the initial subspace $\text{\rm{Ker}}[J(i\kappa_j)^\dagger]$ and the
final subspace $\text{\rm{Ker}}[J(i\kappa_j)].$

\end{enumerate}
\end{theorem}

\begin{proof}
From the first equality of \eqref{4.20} we see that
Theorem~\ref{theorem4.2}(e) holds with $L=D_j.$ Hence, all items in Theorem~\ref{theorem4.2} hold when
$L=D_j.$ In particular, Theorem~\ref{theorem4.2}(a) implies that
$D_j$ is a partial isometry with the initial subspace $\left(\text{\rm{Ker}}[D_j]\right)^\perp$ and the final subspace $\text{\rm{Ran}}[D_j].$
On the other hand, from \eqref{4.16} we know that $D_j^\dagger D_j$ is equal to $Q_j.$ 
By Theorem~\ref{theorem4.2}(d) we know that
$D_j^\dagger\, D_j$
is the orthogonal projection onto $\left(\text{\rm{Ker}}[D_j]\right)^\perp$ and we already know  that
$Q_j$ is the orthogonal projection onto
the kernel of $J(i\kappa_j).$ Thus, we conclude that
\begin{equation*}
\left(\text{\rm{Ker}}[D_j]\right)^\perp
=\text{\rm{Ker}}[J(i\kappa_j)].
\end{equation*}
Moreover, by Theorem~\ref{theorem4.2}(c) we know that $D_j D_j^\dagger$ is the orthogonal projection
onto $\text{\rm{Ran}}[D_j],$
which is equal to $\text{\rm Ker}[D_j^\dagger]^\perp.$ On the other hand, \eqref{4.17} shows that $D_j D_j^\dagger$ is the 
orthogonal projection onto $\textrm{Ker}[J(i\kappa_j)^\dagger].$
Hence, we have
\begin{equation*}
\text{\rm{Ran}}[D_j]=\left(\text{\rm{Ker}}[D_j^\dagger]\right)^\perp
=\text{\rm{Ker}}[J(i\kappa_j)^\dagger].
\end{equation*}
Consequently, with the help of Theorem~\ref{theorem4.2}(a) we conclude that
$D_j$ is a partial isometry with the initial subspace
$\text{\rm{Ker}}[J(i\kappa_j)]$ and the final subspace
$\text{\rm{Ker}}[J(i\kappa_j)^\dagger],$
which completes the proof of (a).
With the help of the equivalences stated in Theorem~\ref{theorem4.2},
we see that (b) is equivalent to (a), and hence the proof of (b) is also complete. 
\end{proof}

\section{The Gel'fand--Levitan system of integral equations}
\label{section5}

In this section we introduce the spectral measure related to the half-line matrix
Schr\"odinger operator described by \eqref{2.1} and \eqref{2.5},
provide Parseval's equality involving the regular solution and
the Gel'fand--Levitan normalized bound-state solutions, and
derive the corresponding matrix-valued Gel'fand--Levitan system of linear equations.

 When the potential $V$ in \eqref{2.1} satisfies \eqref{2.2} and belongs to $L^1_1(\mathbb R^+),$
as indicated in Section~\ref{section3}, the
 Schr\"odinger operator associated with \eqref{2.1} and \eqref{2.5} has at most a finite number of distinct bound states
occurring at $k=i\kappa_j$ for $1\le j\le N.$ It is understood that $N=0$ if there are no
bound states.
The spectral measure $d\rho$ associated with the corresponding selfadjoint Schr\"odinger operator
 related to \eqref{2.1} and \eqref{2.5} is given by
\begin{equation}
\label{5.1}
d\rho=\begin{cases}
\ds\frac{\sqrt{\lambda}}{\pi}\,\left(J(k)^\dagger\,J(k)\right)^{-1}\,d\lambda,\qquad \lambda\ge 0,
\\
\noalign{\medskip}
\ds\sum_{j=1}^N C_j^2\,\delta(\lambda-\lambda_j)\,d\lambda,
\qquad \lambda<0,\end{cases}
\end{equation}
where $\delta(\cdot)$ denotes the Dirac delta distribution and we recall that 
$J(k)$ is the Jost matrix
given in \eqref{2.11}, the quantity $C_j$ 
is the Gel'fand--Levitan normalization matrix defined in \eqref{3.36}, and we let
$\lambda:=k^2$ and $\lambda_j:=-\kappa_j^2$ for $1\le j\le N.$

In the next theorem we present Parseval's equality for the half-line matrix
Schr\"odinger operator associated with \eqref{2.1} and \eqref{2.5}.

\begin{theorem}
\label{theorem5.1} We assume that the potential $V$ appearing in \eqref{2.1} satisfies \eqref{2.2} and
belongs to $L^1_1(\mathbb R^+).$ Associated with the 
 half-line selfadjoint
Schr\"odinger operator described by \eqref{2.1} and \eqref{2.5},
we have the completeness relation known as Parseval's equality, which is given by
\begin{equation}
\label{5.2}
\ds\int_{\lambda\in\mathbb R}  \varphi(k,x)\,d\rho\,
\varphi(k,y)^\dagger=\delta(x-y)\,I,
\end{equation}
with the integration  from $\lambda=-\infty$ to $\lambda=+\infty$
covering the domain of the spectral measure $d\rho.$
Parseval's equality can also be written as
\begin{equation}
\label{5.3}
\ds\int_{\lambda\in\mathbb R^+} \varphi(k,x)\,d\rho\,
\varphi(k,y)^\dagger+\ds\sum_{j=1}^N \Phi_j(x)\,\Phi_j(y)^\dagger=\delta(x-y)\,I,
\end{equation}
where the integration is from $\lambda=0$ to $\lambda=+\infty,$ and we recall that
$\varphi(k,x)$ is the regular solution to \eqref{2.1} satisfying the initial
conditions \eqref{2.10}, $d\rho$ is the spectral measure appearing
in \eqref{5.1}, $\Phi_j(x)$ is the $n\times n$ Gel'fand--Levitan normalized bound-state  solution at $k=i\kappa_j$ and appearing in \eqref{3.28}, and $N$ is the
number of bound states without counting the multiplicities.

\end{theorem}

\begin{proof} We already know from (5.9.1) of \cite{AW2021} that Parseval's equality can be expressed in terms of the
physical solution $\Psi(k,x)$ appearing in \eqref{2.13}
and the $n\times n$ Marchenko normalized bound-state  solutions $\Psi_j(x)$
appearing in \eqref{3.21}, and we have
\begin{equation}
\label{5.4}
\ds\frac{1}{2\pi} \ds\int_0^\infty
dk\, \Psi(k,x)\,
\Psi(k,y)^\dagger+\ds\sum_{j=1}^N \Psi_j(x)\,\Psi_j(y)^\dagger=\delta(x-y)\,I.
\end{equation}
From \eqref{2.14} we obtain
\begin{equation}
\label{5.5}
\Psi(k,x)=-2ik\,\varphi(k,x)\,J(k)^{-1},\qquad
k\in\mathbb R,
\end{equation}
where the existence of
$J(k)^{-1}$
for $x\in\mathbb R\setminus \{0\}$ and the continuity
of $k\,J(k)^{-1}$ at $k=0$ are assured \cite{AW2021} under the stated assumptions on the potential
and the selfadjoint boundary condition. From \eqref{5.5} we get
\begin{equation*}
\Psi(k,x)\, \Psi(k,y)^\dagger=4 k^2\,\varphi(k,x)\,J(k)^{-1}
\left[J(k)^{-1}\right]^\dagger \varphi(k,y)^\dagger,\qquad
k\in\mathbb R,
\end{equation*}
which yields
\begin{equation}
\label{5.7}
\ds\frac{dk}{2\pi}\,\Psi(k,x)\, \Psi(k,y)^\dagger=\ds\frac{2 k^2\,dk}{\pi}\,\varphi(k,x)\left[J(k)^\dagger\,J(k)\right]^{-1}\varphi(k,y)^\dagger,\qquad
k\in\mathbb R.\end{equation}
Since $\lambda=k^2,$ we can write \eqref{5.7} in an equivalent form as
\begin{equation}
\label{5.8}
\ds\frac{dk}{2\pi}\,\Psi(k,x)\, \Psi(k,y)^\dagger=
\ds\frac{\sqrt{\lambda}\,d\lambda} {\pi}\,\varphi(k,x)\left[J(k)^\dagger\,J(k)\right]^{-1}\varphi(k,y)^\dagger,\qquad
\lambda\ge 0.\end{equation}
Using
the first line of \eqref{5.1} on the right-hand side of
\eqref{5.8} and then integrating both sides of the resulting equality,
we obtain
\begin{equation}
\label{5.9}
\ds\frac{1}{2\pi} \ds\int_0^\infty
dk\, \Psi(k,x)\,
\Psi(k,y)^\dagger=
\ds\int_{\lambda\in\mathbb R^+} \varphi(k,x)\,d\rho\,
\varphi(k,y)^\dagger.\end{equation}
On the other hand, from \eqref{4.1} we have
\begin{equation}
\label{5.10}
\Phi_j(x)\,\Phi_j(y)^\dagger=\Psi_j(x)\,D_j\,D_j^\dagger \Psi_j(y)^\dagger.\end{equation}
Using \eqref{4.17} in \eqref{5.10} we obtain
\begin{equation}
\label{5.11}
\Phi_j(x)\,\Phi_j(y)^\dagger=\Psi_j(x) P_j \Psi_j(y)^\dagger.\end{equation}
 From the first equality in \eqref{3.39} we see
 that 
\begin{equation*}
\Psi_j(x)\,P_j=\Psi_j(x),
\end{equation*}
 and hence \eqref{5.11} yields
\begin{equation*}
\Phi_j(x)\,\Phi_j(y)^\dagger=\Psi_j(x)\,\Psi_j(y)^\dagger,
\end{equation*}
which implies that
\begin{equation}
\label{5.14}
\ds\sum_{j=1}^N
\Phi_j(x)\,\Phi_j(y)^\dagger=\ds\sum_{j=1}^N
\Psi_j(x)\,\Psi_j(y)^\dagger.\end{equation}
Adding \eqref{5.9} and \eqref{5.14} we obtain the equivalent
of the left-hand side of \eqref{5.4}, and hence
we get \eqref{5.3}.
Using \eqref{3.28}  and \eqref{5.1} on the
left-hand side of \eqref{5.3}, we observe
that \eqref{5.3} is equivalent to \eqref{5.2}.
\end{proof}

Let us view $V(x),$ $\varphi(k,x),$ $(A,B),$ and $d\rho$
as the potential, the regular solution, the pair of boundary
matrices, and the spectral measure corresponding to the unperturbed problem.
Thus, the regular solution $\varphi(k,x)$ 
satisfies the unperturbed Schr\"odinger equation
\begin{equation}
\label{5.15}
\varphi''(k,x)=\left[V(x)-k^2\right]\varphi(k,x),\end{equation}
and the unperturbed initial conditions \eqref{2.10}.
Let us view $\tilde V(x),$
$\tilde\varphi(k,x),$ $(\tilde A,\tilde B),$ and
$d\tilde\rho$ as the potential, the regular solution, the pair of boundary
matrices, and the spectral measure corresponding to the perturbed problem.
Thus, the perturbed regular solution $\tilde\varphi(k,x)$ 
satisfies the perturbed Schr\"odinger equation
\begin{equation}
\label{5.16}
\tilde\varphi''(k,x)=\left[\tilde V(x)-k^2\right]\tilde\varphi(k,x),\end{equation}
and the perturbed initial conditions
\begin{equation}
\label{5.17}
\tilde\varphi(k,0)=\tilde A,\quad \tilde\varphi'(k,0)=\tilde B.\end{equation}
Let us assume that the unperturbed and perturbed potentials 
satisfy \eqref{2.2} and  belong to $L^1_1(\mathbb R^+),$ and let us also assume that
the unperturbed and perturbed boundary conditions 
expressed by the respective unperturbed and perturbed boundary matrices
are compatible with the selfadjoint boundary condition described in \eqref{2.5}--\eqref{2.7}.

We can  specify the perturbation relating the unperturbed problem and the perturbed problem 
in various different ways. In the rest of this section we specify the perturbation by relating the perturbed regular solution $\tilde\varphi(k,x)$
 to the unperturbed regular solution $\varphi(k,x)$ through a matrix-valued quantity $\mathcal A(x,y)$
via the relation given by 
\begin{equation}
\label{5.18}
\tilde\varphi(k,x)=\varphi(k,x)+\int_0^x dy\,\mathcal A(x,y)\,\varphi(k,y).
\end{equation}
We would like to find out how the perturbed potential $\tilde V(x)$ is related
to the unperturbed potential $V(x)$ through the quantity $\mathcal A(x,y).$
Similarly, we would like to relate the perturbed boundary matrices
$\tilde A$ and $\tilde B$ to the unperturbed 
boundary matrices
$A$ and $B$ through the quantity $\mathcal A(x,y).$
In the next theorem we present 
the appropriate restrictions on $\mathcal A(x,y)$ to obtain the aforementioned relations.

 \begin{theorem}
\label{theorem5.2}
Assume that the $n\times n$ matrix-valued potential $V$ appearing in \eqref{2.1} satisfies \eqref{2.2} and
is locally integrable on $x\in [0,+\infty).$ For $k\in\mathbb C,$ let $\varphi(k,x)$ be the $n\times n$ matrix-valued solution to \eqref{2.1}
with the initial values specified in \eqref{2.10}, and assume that $\varphi(k,x)$ is continuously differentiable in $x \in [0,+\infty).$ Suppose
 that the quantity $\mathcal A(x,y)$ appearing in \eqref{5.18} is continuously differentiable on the set $\mathbb D$ defined as
\begin{equation}
\label{5.19}
\mathbb D:=\left\{(x,y):  0\le y \le x<+\infty\right\},
\end{equation}
and that the second derivatives $\mathcal A_{xx}(x,y)$
and $\mathcal A_{yy}(x,y)$ are locally integrable on $\mathbb D,$ 
where the subscripts denote the appropriate partial derivatives.
Further, assume that the 
quantity $\mathcal A(x,y)$ satisfies the second-order matrix-valued  partial differential
equation
\begin{equation}
\label{5.20}
\mathcal A_{xx}(x,y)-\mathcal A_{yy}(x,y)=\tilde V(x)\,\mathcal A(x,y)-
\mathcal A(x,y)\,V(y),
\end{equation}
  with the quantity $\tilde V(x)$ given by
\begin{equation}
\label{5.21}
\tilde V(x):=V(x)+2\,\ds\frac{d\, \mathcal A(x,x)}{dx},
\end{equation}
 where we use $\mathcal A(x,x)$ to denote $\mathcal A(x,x^-).$ 
Then, we have the following:

\begin{enumerate}

\item[\text{\rm(a)}]
The  $n\times n$ matrix-valued function $\tilde\varphi(k,x)$ appearing in \eqref{5.18} satisfies \eqref{2.1} with the potential $\tilde V(x)$
specified in \eqref{5.21}.

\item[\text{\rm(b)}]
The $n\times n$ constant matrices $A$ and $B$ appearing in \eqref{2.10}
are related to each other via the matrix $\mathcal A(x,y)$ as
\begin{equation}
\label{5.22}
-\mathcal A_{y}(x,0)\,A
+
\mathcal A(x,0)\,B=0.\end{equation}

\item[\text{\rm(c)}]
The $n\times n$ constant matrices $\tilde A$ and $\tilde B$ appearing in 
\eqref{5.17} are related to
the constant matrices $A$ and $B$ appearing in \eqref{2.10} as
\begin{equation}
\label{5.23}
\tilde A=A,\quad 
\tilde B=B+\mathcal A(0,0)\,A.
\end{equation}
\end{enumerate}
\end{theorem}

\begin{proof} 
Using the quantity $\tilde\varphi(k,x)$
specified in \eqref{5.18}, by taking the first and second $x$-derivatives, respectively, we obtain
\begin{equation*}
\tilde\varphi'(k,x)=\varphi'(k,x)+\mathcal A(x,x)\,\varphi(k,x)+\int_0^x dy\,\mathcal A_x(x,y)\,\varphi(k,y),
\end{equation*}
\begin{equation}
\label{5.25}
\tilde\varphi''(k,x)=q_1(k,x),
\end{equation}
where we have defined
\begin{equation}
\label{5.26}
q_1(k,x):=
\varphi''(k,x)+\mathcal A_x(x,x)\,\varphi(k,x)+\mathcal A_y(x,x)\,\varphi(k,x)+q_2(x),
\end{equation}
\begin{equation*}
q_2(x):=
\mathcal A(x,x)\,\varphi'(k,x)+
\mathcal A_x(x,x)\,\varphi(k,x)+\int_0^x dy\,\mathcal A_{xx}(x,y)\,\varphi(k,y).
\end{equation*}
Using \eqref{5.15}, we write \eqref{5.26} as
\begin{equation}
\label{5.28}
q_1(k,x)=q_3(x)+q_2(x)
-\int_0^x dy\,\mathcal A_{yy}(x,y)\,\varphi(k,y)+
\int_0^x dy\,\mathcal A_{yy}(x,y)\,\varphi(k,y),
\end{equation}
where we have defined
\begin{equation*}
q_3(x):=\left[V(x)-k^2\right]\varphi(k,x)
+\mathcal A_x(x,x)\,\varphi(k,x)+\mathcal A_y(x,x)\,\varphi(k,x).
\end{equation*}
The integration by parts in the last term
of \eqref{5.28} yields
\begin{equation}
\label{5.30}
\int_0^x dy\,\mathcal A_{yy}(x,y)\,\varphi(k,y)=
\mathcal A_{y}(x,x)\,\varphi(k,x)-
\mathcal A_{y}(x,0)\,\varphi(k,0)-
\int_0^x dy\,\mathcal A_{y}(x,y)\,\varphi'(k,y).\end{equation}
Using integration by parts in the last integral on the right-hand side of \eqref{5.30}, we have
\begin{equation}
\label{5.31}
\int_0^x dy\,\mathcal A_{y}(x,y)\,\varphi'(k,y)=
\mathcal A(x,x)\,\varphi'(k,x)-
\mathcal A(x,0)\,\varphi'(k,0)\\
-
\int_0^x dy\,\mathcal A(x,y)\,\varphi''(k,y).\end{equation}
Using \eqref{5.15} in the last term on the right-hand side of \eqref{5.31}, we can write \eqref{5.31} as
\begin{equation}
\label{5.32}
\int_0^x dy\,\mathcal A_{y}(x,y)\,\varphi'(k,y)=q_4(x),
\end{equation}
where we have defined
\begin{equation*}
q_4(x):=
\mathcal A(x,x)\,\varphi'(k,x)-
\mathcal A(x,0)\,\varphi'(k,0)\\
-
\int_0^x dy\,\mathcal A(x,y)\left[V(y)-k^2\right]\varphi(k,y).
\end{equation*}
Then,
using \eqref{5.30}, \eqref{5.31}, and \eqref{5.32} on the right-hand side of
\eqref{5.28}, we obtain
\begin{equation}
\label{5.34}
q_1(k,x)=q_3(x)+q_2(x)+q_5(x)-q_4(x),
\end{equation}
where we have defined
\begin{equation*}
q_5(x):=-\ds\int_0^x dy\,\mathcal A_{yy}(x,y)\,\varphi(k,y)+
\mathcal A_y(x,x)\,\varphi(k,x)-
\mathcal A_y(x,0)\,\varphi(k,0).
\end{equation*}
We have the identity
\begin{equation}
\label{5.36}
\mathcal A_x(x,x)+\mathcal A_y(x,x)=\ds\frac{d\, \mathcal A(x,x)}{dx},\end{equation}
where we recall that the values evaluated at $y=x$ are obtained
by using the values at $y=x^-.$ Using \eqref{5.36} in \eqref{5.34}, with the help of \eqref{5.25} we get
\begin{equation}
\label{5.37}
\aligned\tilde\varphi''(k,x)=&
\left[V(x)+2\,\ds\frac{d\, \mathcal A(x,x)}{dx}-k^2\right]\varphi(k,x)
-
\mathcal A_{y}(x,0)\,\varphi(k,0)
+
\mathcal A(x,0)\,\varphi'(k,0)\\&
+
\int_0^x dy\,\left[\mathcal A_{xx}(x,y)-\mathcal A_{yy}(x,y)+
\mathcal A(x,y)\,\left[V(y)-k^2\right]\right]\varphi(k,y).\endaligned
\end{equation}
From \eqref{5.18}, \eqref{5.20}, \eqref{5.25}, and \eqref{5.37} we obtain 
\begin{equation*}
\tilde\varphi''(k,x)= \left[V(x)+2\,\ds\frac{d\, \mathcal A(x,x)}{dx}-k^2\right]\tilde\varphi(k,x)-
\mathcal A_{y}(x,0)\,\varphi(k,0)
+
\mathcal A(x,0)\,\varphi'(k,0).
\end{equation*}
From \eqref{5.37} we see that the proof of (a) is complete provided we have
\begin{equation}
\label{5.39}
-
\mathcal A_{y}(x,0)\,\varphi(k,0)
+
\mathcal A(x,0)\,\varphi'(k,0)=0.
\end{equation}
We establish \eqref{5.39} as follows.
Because of \eqref{2.10} we observe that
\eqref{5.39} is equivalent to \eqref{5.22}.
We next show that \eqref{5.22} is automatically satisfied.
By multiplying \eqref{5.18} on the right by $d\rho\,\varphi(k,z)^\dagger,$
we obtain
\begin{equation}
\label{5.40}
\tilde\varphi(k,x)\,d\rho\,\varphi(k,z)^\dagger=\varphi(k,x)\,
d\rho\,\varphi(k,z)^\dagger\\
+\int_0^x dy\,\mathcal A(x,y)\,\varphi(k,y)\,d\rho\,\varphi(k,z)^\dagger.
\end{equation}
Integrating both sides of \eqref{5.40} over $\lambda\in\bR,$ we get
\begin{equation}
\label{5.41}
\ds\int_{\lambda\in\mathbb R}\tilde\varphi(k,x)\,d\rho\,\varphi(k,z)^\dagger=
\delta(x-z)\,I+\int_0^x dy\,\mathcal A(x,y)\,\delta(y-z)\,I,\end{equation}
where we have used Parseval's equality given in
\eqref{5.2}.
 From \eqref{5.41} we have
\begin{equation*}
\mathcal A(x,y)=\ds\int_{\lambda\in\mathbb R}\tilde\varphi(k,x)\,d\rho\,\varphi(k,y)^\dagger,
\qquad 0\le y<x,
\end{equation*}
 which yields
\begin{equation}
\label{5.43}
\mathcal A(x,0)=\ds\int _{\lambda\in\mathbb R}\tilde\varphi(k,x)\, d\rho\,\varphi(k,0)^\dagger,\end{equation}
\begin{equation*}
\mathcal A_y(x,y)=\ds\int_{\lambda\in\mathbb R} \tilde\varphi(k,x)\, d\rho\,\varphi'(k,y)^\dagger,
\end{equation*}
\begin{equation}
\label{5.45}
\mathcal A_y(x,0)=\ds\int _{\lambda\in\mathbb R}\tilde\varphi(k,x)\, d\rho\,\varphi'(k,0)^\dagger.\end{equation}
Using \eqref{2.10} in \eqref{5.43} and \eqref{5.45}, we have
\begin{equation}
\label{5.46}
\mathcal A(x,0)=\ds\int_{\lambda\in\mathbb R} \tilde\varphi(k,x)\, d\rho\,A^\dagger,\end{equation}
\begin{equation}
\label{5.47}
\mathcal A_y(x,0)=\ds\int _{\lambda\in\mathbb R}\tilde\varphi(k,x)\, d\rho\,B^\dagger.\end{equation}
Hence, from \eqref{5.46} and \eqref{5.47} we obtain
\begin{equation}
\label{5.48}
-\mathcal A_y(x,0) A+\mathcal A(x,0) B=\ds\int_{\lambda\in\mathbb R} \tilde\varphi(k,x)\, d\rho
\,(-B^\dagger A+A^\dagger B).\end{equation}
Using \eqref{2.6} in \eqref{5.48} we get \eqref{5.22}. This completes the proofs of (a) and (b).
Since $\varphi(k,x)$ satisfies the initial conditions \eqref{2.10}
and $\tilde\varphi(k,x)$ satisfies the initial conditions \eqref{5.17}, we obtain \eqref{5.23}
by evaluating \eqref{5.18} and its $x$-derivative at $x=0.$
Thus, the proof of (c) is also complete.
 \end{proof}

We have the following remark.
Recall that an input data set $\mathbf D$ consists of
the potential and the boundary condition, or equivalently
it is described by the set $\{V,A,B\}.$ The perturbation of the
potential is expressed in \eqref{5.21}, and the perturbation of
the boundary matrices $A$ and $B$ are expressed in
\eqref{5.23}. The perturbation of the wavefunction
is given by \eqref{5.18}, and as the Gel'fand--Levitan method indicates
it is appropriate to describe the perturbation at the wavefunction level
by providing the perturbation of the regular solution given in \eqref{5.18}.
This is because the regular solution is the wavefunction directly related to
the input data set, as it is related to the potential $V(x)$
through the Schr\"odinger equation \eqref{5.15} and the boundary matrices
$A$ and $B$ through the initial conditions \eqref{2.10}. In particular, in the evaluation of the transformations
on the half line it is appropriate to
describe how the regular solution $\varphi(k,x)$
is transformed. For the half-line
Schr\"odinger equation, the Jost solution $f(k,x)$ is not the appropriate
wavefunction to express the perturbation at the wavefunction level.
This is because $f(k,x)$ is only affected by the potential $V(x)$ but not
by the boundary matrices $A$ and $B.$

In the next theorem we show that
the quantity $\mathcal A(x,y)$ appearing in
\eqref{5.18},  \eqref{5.20}, and \eqref{5.21} satisfies the matrix-valued
Gel'fand--Levitan system of integral equations.

 \begin{theorem}
\label{theorem5.3} Assume that the unperturbed and perturbed potentials $V$ and $\tilde V$ appearing in 
\eqref{5.15} and \eqref{5.16}, respectively, both satisfy
\eqref{2.2} and belong to $L^1_1(\mathbb R^+).$ Furthermore, assume that the unperturbed boundary 
matrix pair $(A,B)$ and
the perturbed boundary matrix pair $(\tilde A,\tilde B)$
appearing in \eqref{2.10} and \eqref{5.17}, respectively, both yield
a selfadjoint boundary condition described in \eqref{2.5}--\eqref{2.7}.
Let the perturbed regular solution $\tilde\varphi(k,x)$ be related to the
unperturbed regular solution $\varphi(k,x)$ through the $n\times n$ matrix $\mathcal A(x,y)$
as in \eqref{5.18}. Then, $\mathcal A(x,y)$ satisfies the Gel'fand--Levitan system of integral equations given by
\begin{equation}
\label{5.49}
\mathcal A(x,y)+G(x,y)+\int_0^x dz\,\mathcal A(x,z)\,G(z,y)=0,\qquad 0\le y< x,\end{equation}
where we have defined
\begin{equation}
\label{5.50}
G(x,y):=\int_{\lambda\in\mathbb R} \varphi(k,x)\left[d\tilde\rho-d\rho\right]\varphi(k,y)^\dagger,
\end{equation}
with $d\rho$ and $d\tilde\rho$ denoting the unperturbed and perturbed spectral measures, respectively.
\end{theorem}

\begin{proof} Let us multiply \eqref{5.18} on the right by
$d\tilde\rho\,\varphi(k,z)^\dagger$ and integrate over $\lambda\in\bR,$
where we recall that $\lambda:=k^2.$ We get
\begin{equation}
\label{5.51}
\ds\int _{\lambda\in\mathbb R}\tilde\varphi(k,x)\, d\tilde\rho\,\varphi(k,z)^\dagger
\\
=
\ds\int_{\lambda\in\mathbb R} \varphi(k,x)\,d\tilde\rho\,\varphi(k,z)^\dagger
+\int_0^x dy\,\mathcal A(x,y)\,\int_{\lambda\in\mathbb R} \varphi(k,y)\,d\tilde\rho\,\varphi(k,z)^\dagger.
\end{equation}
We have
\begin{equation}
\label{5.52}
\ds\int _{\lambda\in\mathbb R}\tilde\varphi(k,x)\,d\tilde\rho\,\varphi(k,z)^\dagger=0,
\qquad 0\le z<x,\end{equation}
which is obtained with the help of
Parseval's equality given by
\begin{equation*}
\ds\int_{\lambda\in\mathbb R} \tilde\varphi(k,x)\,d\tilde\rho\,\tilde\varphi(k,z)^\dagger=\delta(x-z)\,I,
\end{equation*}
and by using the fact that for each fixed $k\in\bC$ and fixed
$x\in\bR^+$ we can express
$\varphi(k,z)$ in terms of
the functions in the set
$\{ \tilde\varphi(k,y): y\in [0,z]\}.$ To  prove this we remark that we can consider \eqref{5.18} as an integral 
equation where $\varphi(k,x)$ is the unknown quantity and $\tilde\varphi(k,x)$ is the nonhomogeneous term.
We further observe that the integral operator on the right-hand side of \eqref{5.18} is a Volterra integral operator,
and hence the integral equation \eqref{5.18} can be inverted into an integral equation of a similar form, which is given by
\begin{equation*}
\varphi(k,x)= \tilde\varphi(k,x)+ \ds\int_0^x dy\, \tilde{\mathcal A}(x,y)\, \tilde\varphi(k,y),
\end{equation*}
where the quantity $\tilde{\mathcal A}(x,y)$ corresponds to the kernel, $\varphi(k,x)$
is the nonhomogeneous term, and $\tilde\varphi(k,x)$ is the unknown quantity.
We write the first integral on the right-hand side of \eqref{5.51} as
\begin{equation}
\label{5.55}
\ds\int_{\lambda\in\mathbb R} \varphi(k,x)\,
 d\tilde\rho\,\varphi(k,z)^\dagger\\
=
\ds\int_{\lambda\in\mathbb R} \varphi(k,x)\left[d\tilde\rho-d\rho\right]\varphi(k,z)^\dagger
+
\ds\int _{\lambda\in\mathbb R}\varphi(k,x)\,d\rho\,\varphi(k,z)^\dagger.\end{equation}
Using \eqref{5.2} and \eqref{5.50} on the right-hand side of \eqref{5.55}, we obtain
\begin{equation}
\label{5.56}
\ds\int_{\lambda\in\mathbb R} \varphi(k,x)\,d\tilde\rho\,\varphi(k,z)^\dagger=
G(x,z)+\delta(x-z)\,I.\end{equation}
Thus, from \eqref{5.56} we get
\begin{equation}
\label{5.57}
\ds\int _{\lambda\in\mathbb R}\varphi(k,x)\,d\tilde\rho\,\varphi(k,z)^\dagger=
G(x,z),\qquad 0\le z<x.\end{equation}
The second integral on the right-hand side
of \eqref{5.51} can be written as
\begin{equation}
\label{5.58}
\begin{split}
\int_0^x dy\,\mathcal A(x,y)\ds\int_{\lambda\in\mathbb R} \varphi(k,y)\,
d\tilde\rho\,\varphi(k,z)^\dagger
=&
\int_0^x dy\,\mathcal A(x,y)\ds\int_{\lambda\in\mathbb R} \varphi(k,y)\left[d\tilde\rho-d\rho\right]\varphi(k,z)^\dagger\\
&+
\int_0^x dy\,\mathcal A(x,y)\int_{\lambda\in\mathbb R} \varphi(k,y)\,d\rho\,\varphi(k,z)^\dagger.
\end{split}\end{equation}
Using \eqref{5.2} in the second integral and using \eqref{5.50}
in the first integral on the
right-hand side of
\eqref{5.58}, we obtain
\begin{equation}
\label{5.59}
\int_0^x dy\,\mathcal A(x,y)\int_{\lambda\in\mathbb R} \varphi(k,y)\,d\tilde\rho\,\varphi(k,z)^\dagger=
\int_0^x dy\,\mathcal A(x,y)\,G(y,z)
+
\mathcal A(x,z),\qquad 0\le z<x.
\end{equation}
Thus, using \eqref{5.52}, \eqref{5.57}, and \eqref{5.59} in \eqref{5.51} we get
\eqref{5.49}. \end{proof}

 \section{The transformation to remove a bound state}
\label{section6}

We recall that we consider the half-line matrix
Schr\"odinger operator described by \eqref{2.1} and \eqref{2.5}, and 
we change the unperturbed
problem into the perturbed problem without changing the continuous spectrum but only
by adding or removing a finite number of
bound states or by decreasing or increasing the multiplicities of the bound states. 
We are interested in finding the transformations describing the change in the potential, the boundary matrices, 
the regular solution, and in other relevant quantities.
In this case, it is convenient to use 
\eqref{5.18}, \eqref{5.21}, and \eqref{5.23}
to describe the transformations for the regular solution, the potential, and the boundary matrices, respectively, 
 when the spectral measures
$d\tilde\rho$ and $d\rho$ differ from each other only in the bound states.
Since those transformations involve the quantity $\mathcal A(x,y)$ appearing in 
the Gel'fand--Levitan system \eqref{5.49} with the input given in \eqref{5.50},
it is natural to use the Gel'fand--Levitan method to determine the relevant
transformations.

Since the removal or addition of bound states or changing their multiplicities can be performed in succession, there is no loss of generality in considering only the case when one bound state is removed or added or when its multiplicity is decreased or increased.
In this section, we consider the case where a bound state is completely removed from the spectrum of the half-line Schr\"odinger operator associated with
\eqref{2.1} and \eqref{2.5}.

Let us start with the unperturbed Schr\"odinger operator described by \eqref{2.1} and \eqref{2.5},
where the potential $V$ satisfies \eqref{2.2} and belongs to  $L^1_1(\mathbb R^+)$ and the
boundary matrices $A$ and $B$ satisfy \eqref{2.5}--\eqref{2.7}.
Thus, for the unperturbed problem, we have the potential
$V(x),$ the boundary matrices $A$ and $B,$ the regular solution
$\varphi(k,x),$ the spectral measure $d\rho,$ 
the Jost matrix $J(k),$ the scattering matrix $S(k),$ and
$N$ bound states with energies $-\kappa_j^2,$ 
the Gel'fand--Levitan normalization matrices $C_j,$ the orthogonal projections
$Q_j$ onto $\text{\rm{Ker}}[J(i\kappa_j)],$ the orthogonal projections
$P_j$ onto $\text{\rm{Ker}}[J(i\kappa_j)^\dagger],$ 
the Gel'fand--Levitan normalized bound-state solutions $\Phi_j(x),$ the Marchenko normalized
bound-state solutions $\Psi_j(x),$ the dependency matrices $D_j,$ and the multiplicities $m_j$ of the bound states
for $1\le j\le N.$

We would like to remove any one of the bound states at $\lambda=\lambda_j,$ or equivalently at $k=i\kappa_j,$ with
the Gel'fand--Levitan normalization matrix $C_j.$
Since we do not order the distinct positive constants $\kappa_j$ in any increasing or decreasing manner, without loss
of generality we can assume that we remove the bound state
with $k=i\kappa_N$ and the Gel'fand--Levitan normalization matrix $C_N.$
We choose the Gel'fand--Levitan kernel $G(x,y)$ given in \eqref{5.50} as
\begin{equation}
\label{6.1}
G(x,y)=- \varphi(k,x) \,C_N^2\, \varphi(k,y)^\dagger.
\end{equation}
Next, we use the input \eqref{6.1} in the Gel'fand--Levitan system \eqref{5.49} and we obtain
its solution $\mathcal A(x,y).$ With the help of $\mathcal A(x,y)$ we then form
the perturbed Schr\"odinger operator with 
the potential
$\tilde V(x),$ the boundary matrices $\tilde A$ and $\tilde B,$ the regular solution
$\tilde\varphi(k,x),$ the spectral measure $d\tilde\rho,$ 
the Jost matrix $\tilde J(k),$ the scattering matrix $\tilde S(k),$ and
$N-1$ bound states with energies $-\kappa_j^2,$ 
the Gel'fand--Levitan normalization matrices $\tilde C_j,$ the orthogonal projections
$\tilde Q_j$ onto $\text{\rm{Ker}}[\tilde J(i\kappa_j)],$ the orthogonal projections
$\tilde P_j$ onto $\text{\rm{Ker}}[\tilde J(i\kappa_j)^\dagger],$ 
the Gel'fand--Levitan normalized bound-state solutions $\tilde\Phi_j(x),$ the Marchenko normalized
bound-state solutions $\tilde\Psi_j(x),$ the dependency matrices $\tilde D_j,$ and the multiplicities $m_j$ of the bound states
for $1\le j\le N-1.$
Our goal is to express the perturbed quantities distinguished with a tilde in their notations
in terms of the unperturbed quantities not containing a tilde and the perturbation identified by $\kappa_N$ and $C_N.$

In the next theorem, we present the solution to the Gel'fand--Levitan system \eqref{5.50} with the input \eqref{6.1}.

\begin{theorem}\label{theorem6.1}
Consider the unperturbed Schr\"odinger operator associated with \eqref{2.1} and \eqref{2.5},
where the potential
$V$ satisfies \eqref{2.2} and belongs to
$L^1_1(\mathbb R^+)$
and the boundary matrices $A$ and $B$ appearing in \eqref{2.5} satisfy 
\eqref{2.6} and \eqref{2.7}. Let $\varphi(k,x)$ be the corresponding regular solution
satisfying the initial conditions
 \eqref{2.10},
and each of $k=i\kappa_j$ for $1\le j\le N$ correspond to the bound state with the energy $-\kappa_j^2,$ 
the Gel'fand--Levitan normalization matrix $C_j,$ 
the orthogonal projection
$Q_j$ onto $\text{\rm{Ker}}[J(i\kappa_j)],$ 
and
the Gel'fand--Levitan normalized bound-state solution $\Phi_j(x).$
If the kernel $G(x,y)$ of the matrix-valued Gel'fand--Levitan system \eqref{5.49} is chosen as in \eqref{6.1},
then the solution $\mathcal A(x,y)$ to \eqref{5.49} is given by 
\begin{equation}\label{6.2}
\mathcal A(x,y)= \Phi_N(x) \,W_N(x)^+\,\Phi_N(y)^\dagger, \qquad 0 \le y < x, 
\end{equation}
with $W_N(x)$ defined as
\begin{equation}
\label{6.3}
W_N(x):=\int_x^\infty dz\, \Phi_N(z)^\dagger\,\Phi_N(z),
\end{equation}
where we recall that $W_N(x)^+$ denotes the Moore--Penrose inverse of $W_N(x).$
\end{theorem}

\begin{proof}
From  \eqref{3.28} we see that \eqref{6.1} is equivalent to
\begin{equation}
\label{6.4}
G(x,y)=- \Phi_N(x)\,\Phi_N(y)^\dagger.\end{equation}
Using the input \eqref{6.4} in \eqref{5.49}, we write the Gel'fand--Levitan system of integral equations as
\begin{equation}
\label{6.5}
\mathcal A(x,y)-\Phi_N(x)\,\Phi_N(y)^\dagger-\int_0^x dz\,\mathcal A(x,z)\,\Phi_N(z)\,\Phi_N(y)^\dagger=0,
\qquad 0\le y\le x.
\end{equation}
The specific $y$-dependence of the second and third terms in \eqref{6.5} indicates that
the solution $\mathcal A(x,y)$ to \eqref{6.5} must have the form
\begin{equation}
\label{6.6}
\mathcal A(x,y)= \alpha_N(x)\,\Phi_N(y)^\dagger,\end{equation}
for some $\alpha_N(x)$ satisfying
\begin{equation}
\label{6.7}
\alpha_N(x)=\alpha_N(x)\,Q_N.\end{equation}
Using \eqref{6.6} in \eqref{6.5}, we get
\begin{equation*}
\alpha_N(x)-\Phi_N(x)-\alpha_N(x)\int_0^x dz\, \Phi_N(z)^\dagger\,\Phi_N(z)=0,
\end{equation*}
or equivalently
\begin{equation}
\label{6.9}
\alpha_N(x)\left[Q_N-\int_0^x dz\, \Phi_N(z)^\dagger\,\Phi_N(z)\right]=
\Phi_N(x).\end{equation}
From \eqref{3.29} and \eqref{6.3} we see that we can write $W_N(x)$ also as
\begin{equation}
\label{6.10}
W_N(x)=Q_N-\int_0^x dz\, \Phi_N(z)^\dagger\,\Phi_N(z).\end{equation}
Hence, \eqref{6.9} is equivalent to
\begin{equation*}
\alpha_N(x)\,W_N(x)=\Phi_N(x),\end{equation*}
 from which we get
\begin{equation}
\label{6.12}
\alpha_N(x)\, W_N(x)\,W_N(x)^+ =\Phi_N(x)\, W_N(x)^+.\end{equation}
Since $Q_N$ is the orthogonal projection onto $\text{\rm{Ker}}[J(i\kappa_N)],$ with the help of \eqref{3.16} and \eqref{3.28} we can
express $W_N(x)$ as a direct sum and obtain
\begin{equation}\label{6.13}
W_N(x)= [W_N(x)]_1\oplus 0,
\end{equation}
where we use $[W_N(x)]_1$ to denote the restriction of $W_N(x)$ to $\text{\rm{Ker}}[J(i\kappa_N)].$ 
We would like to prove that $[W_N(x)]_1$ is invertible for every $x\ge 0,$ 
and for this we proceed as follows.  If $[W_N(x)]_1$ were not invertible, then we would have a nonzero column vector $\nu$
in $\text{\rm{Ker}}[J(i\kappa_N)]$ satisfying
$W_N(x)\,\nu=0,$ which would yield $\nu^\dagger\,W_N(x)\,\nu=0.$
From \eqref{6.3} we would obtain
\begin{equation}
\label{6.14}
\nu^\dagger\,W_N(x)\,\nu=\ds\int_x^\infty dy\, |\Phi_N(y)\,\nu|^2,
\end{equation}
and hence  we see that the left-hand side of \eqref{6.14} would be equal to zero
if and only 
\begin{equation}
\label{6.15}
\Phi_N(y)\,\nu=0,\qquad y \ge x.
\end{equation}
On the other hand, since  $\Phi_N(x)\,\nu$ is a column-vector solution to \eqref{2.1}, the equality in
\eqref{6.15} would hold if and only if we had
\begin{equation}
\label{6.16}
\Phi_N(x)\,\nu=0,\qquad x \ge 0.
\end{equation}
However, 
\eqref{6.16} would imply that
\begin{equation}
\label{6.17}
\Phi_N(0)\,\nu=0,\quad \Phi'_N(0)\,\nu=0.
\end{equation}
From \eqref{2.10}, \eqref{3.36}, and the second equality of \eqref{3.39}, we observe that
\eqref{6.17} would imply
\begin{equation*}
A\,\mathbf H_N^{-1/2} Q_N \nu=0,\quad B\,\mathbf H_N^{-1/2} Q_N \nu=0,
\end{equation*}
which in turn would yield
\begin{equation}
\label{6.19}
\left(A^\dagger A+B^\dagger B\right)\mathbf H_N^{-1/2} Q_N \nu=0.
\end{equation}
From \eqref{2.7} we know that the matrix
$A^\dagger A+B^\dagger B$ is invertible and from Section~\ref{section3} we know that
$\mathbf H_N^{-1/2}$ is invertible. Thus, \eqref{6.19} would yield
$Q_N \nu=0.$ Since $Q_N$ is the orthogonal projection onto
$\text{\rm{Ker}}[J(i\kappa_N)]$ the equality $Q_N \nu=0$ would imply that $\nu=0.$
This would contradict the assumption that $\nu$ is a nonzero vector in $\text{\rm{Ker}}[J(i\kappa_N)].$ 
Thus, we have proved that $[W_N(x)]_1$ is invertible for every $x\ge 0.$ 
We note that \eqref{3.9} and \eqref{6.13} imply that
\begin{equation}\label{6.20}
W_N(x)^+= \left([W_N(x)]_1\right)^{-1}\oplus 0,
\end{equation}
\begin{equation}\label{6.21}
W_N(x)\,W_N(x)^+=W_N(x)^+\,W_N(x)= I\oplus 0,
\end{equation}
and hence from \eqref{6.21} we conclude that
\begin{equation}\label{6.22}
W_N(x)\,W_N(x)^+= W_N(x)^+\,W_N(x)=Q_N.
\end{equation}
Using \eqref{6.22} in \eqref{6.12}, with the help of \eqref{6.7}, we get
\begin{equation}
\label{6.23}
\alpha_N(x)=\Phi_N(x)\, W_N(x)^+.\end{equation}
Finally, using \eqref{6.23} in \eqref{6.6}, we obtain \eqref{6.2}.
\end{proof}

In the following theorem we present some relevant properties of the quantity $\mathcal A(x,y)$ appearing in \eqref{6.2}.

\begin{theorem} \label{theorem6.2}
Consider the unperturbed Schr\"odinger operator associated with \eqref{2.1} and \eqref{2.5},
where the potential
$V$ satisfies \eqref{2.2} and belongs to
$L^1_1(\mathbb R^+)$
and the boundary matrices $A$ and $B$ appearing in \eqref{2.5} satisfy 
\eqref{2.6} and \eqref{2.7}. Let $\varphi(k,x)$ be the corresponding regular solution
satisfying the initial conditions
 \eqref{2.10}, 
each of $k=i\kappa_j$ for $1\le j\le N$ correspond to the bound state with the energy $-\kappa_j^2,$ 
the Gel'fand--Levitan normalization matrix $C_j,$ 
the orthogonal projection
$Q_j$ onto $\text{\rm{Ker}}[J(i\kappa_j)],$ 
and
the Gel'fand--Levitan normalized bound-state solution $\Phi_j(x).$
Let  $\mathcal A(x,y)$ be the matrix-valued quantity given in \eqref{6.2}.
Let the matrix-valued perturbed potential $\tilde V(x)$ be given by
\begin{equation}
\label{6.24}
\tilde V(x):=V(x)+2\,\ds\frac{d}{dx}\left[
\Phi_N(x) \,W_N(x)^+\,\Phi_N(x)^\dagger
\right],\end{equation}
where $W_N(x)$ is the quantity defined in \eqref{6.3}.
We have the following:

\begin{enumerate}

\item[\text{\rm(a)}] The quantity  $\mathcal A(x,y)$
is continuously differentiable in the set $\mathbb D$ defined in \eqref{5.19}.

\item[\text{\rm(b)}] The second partial derivatives $\mathcal A_{xx}(x,y)$
and $\mathcal A_{yy}(x,y)$ are locally integrable in $\mathbb D.$

\item[\text{\rm(c)}] 
The quantity $\mathcal A(x,y)$ satisfies the second-order matrix-valued partial differential equation \eqref{5.20}.

\item[\text{\rm(d)}] The perturbed potential $\tilde V$ defined in \eqref{6.24} is related to
the unperturbed potential $V$ and the quantity $\mathcal A(x,y)$ as
in \eqref{5.21}.
\end{enumerate}

\end{theorem}
 
 \begin{proof} From Proposition~3.2.9 of \cite{AW2021} and \eqref{3.28} we see that $\Phi_N(x)$ is continuously differentiable on $[0,+\infty).$
 Furthermore, $\Phi''_N(x)$ is locally integrable
 on $[0,+\infty)$ because $\Phi_N(x)$ is a solution to \eqref{3.19}.
Using \eqref{3.10}, \eqref{6.13}, and \eqref{6.20}, we obtain the $x$-derivative of $W_N(x)^+$ as
\begin{equation}
\label{6.25}
\left[W_N(x)^+\right]'=W_N(x)^+\,\Phi_N(x)^\dagger\,\Phi_N(x)\,W_N(x)^+.
\end{equation} 
With the help of \eqref{6.20} and \eqref{6.25}, we see that $W_N(x)^+$ is continuously differentiable on $[0,+\infty).$ Hence, $\mathcal A(x,y)$ 
given in \eqref{6.2} is continuously differentiable  
on $\mathbb D.$
This completes the proof of (a).
For the proof of (b), we proceed as follows.
By taking the derivative of both sides of \eqref{6.25}, and using  \eqref{6.25} in the resulting equation, we prove that $[W_N^+(x)]''$ is continuous on $[0,+\infty).$
Consequently, using \eqref{6.2}, we conclude that 
the second derivatives $\mathcal A_{xx}(x,y)$
and $\mathcal A_{yy}(x,y)$ are locally integrable in 
$\mathbb D.$ Thus, the proof of (b) is complete. The proof of (c) is obtained as follows.
From \eqref{6.2}, by taking the
appropriate derivatives we obtain
\begin{equation*}
\mathcal A_x(x,y)=\left[\Phi_N'(x) \,W_N(x)^++\Phi_N(x) \,(W_N(x)^+)'\right]\Phi_N(y)^\dagger,\end{equation*}
\begin{equation*}
\mathcal A_y(x,y)= \Phi_N(x) \,W_N(x)^+\,\Phi_N'(y)^\dagger,\end{equation*}
\begin{equation}
\label{6.28}
\mathcal A_{xx}(x,y)=\left[\Phi_N''(x) \,W_N(x)^++2\, \Phi_N'(x) \,(W_N(x)^+)'
+ \Phi_N(x) \,(W_N(x)^+)''\right]\Phi_N(y)^\dagger,\end{equation}
\begin{equation}
\label{6.29}
\mathcal A_{yy}(x,y)= \Phi_N(x) \,W_N(x)^+\,{\Phi_N''}(y)^\dagger.\end{equation}
Using \eqref{5.15}, we can express each of \eqref{6.28} and \eqref{6.29} in their equivalent forms as
\begin{equation}
\label{6.30}
\begin{split}
\mathcal A_{xx}(x,y)
=&\left[V(x)+\kappa_N^2\right]\Phi_N(x) \,W_N(x)^+\,\Phi_N(y)^\dagger
\\
&+\left[2\, \Phi_N'(x) \,(W_N(x)^+)'
+ \Phi_N(x) \,(W_N(x)^+)''\right]\Phi_N(y)^\dagger,\end{split}\end{equation}
\begin{equation}
\label{6.31}
\mathcal A_{yy}(x,y)= \Phi_N(x) \,W_N(x)^+\,\Phi_N'(y)^\dagger\left[V(y)+\kappa_N^2\right].\end{equation}
From \eqref{6.2} and \eqref{6.24} we also get
\begin{equation}
\label{6.32}
\tilde V(x)\,\mathcal A(x,y)=\left[V(x)+2\left(
\Phi_N(x) \,W_N(x)^+\,\Phi_N(x)^\dagger\right)'\right]\Phi_N(x) \,W_N(x)^+\,\Phi_N(y)^\dagger,
\end{equation}
\begin{equation}
\label{6.33}
\mathcal A(x,y)\,V(y)=\Phi_N(x) \,W_N(x)^+\,\Phi_N(y)^\dagger\,V(y).\end{equation}
From \eqref{6.31} and \eqref{6.33} we obtain
\begin{equation*}
-\mathcal A_{yy}(x,y)+\mathcal A(x,y)\,V(y)= -\kappa_N^2\,\Phi_N(x) \,W_N(x)^+\,\Phi_N(y)^\dagger.\end{equation*}
Using \eqref{6.30}, \eqref{6.32}, and \eqref{6.33}, we observe that \eqref{5.20} holds provided we have
\begin{equation}
\label{6.35}
2 \Phi_N'(x) \,(W_N(x)^+)'+ \Phi_N(x) \,(W_N(x)^+)''=2\,[
\Phi_N(x) \,W_N(x)^+\,\Phi_N(x)^\dagger]' \,\Phi_N(x) \,W_N(x)^+.\end{equation}
We prove \eqref{6.35} by using \eqref{6.25} and the identity
\begin{equation}
\label{6.36}
\Phi_N'(x)^\dagger\,\Phi_N(x)-\Phi_N(x)^\dagger\,\Phi_N'(x)\equiv 0.
\end{equation}
The proof of \eqref{6.36} can be given as follows. From \eqref{3.28} and \eqref{5.15} we get
\begin{equation}\label{6.37} 
\Phi''_N(x)= \left[V(x)+\kappa_N^2\right] \Phi_N(x).
\end{equation}
By taking the $x$-derivative of the left-hand side of \eqref{6.36} and using \eqref{6.37}, we establish that
the left-hand side of \eqref{6.36} is constant, and hence that constant value can be evaluated at $x=0.$ 
For this, we proceed as follows. With the help of \eqref{2.10} and \eqref{3.28} 
we get
\begin{equation}
\label{6.38}
\Phi_N(0)=A\, C_N,\quad \Phi_N'(0)=B \,C_N.
\end{equation}
Next, using \eqref{6.38} and the fact that $C_N$ is selfadjoint, we evaluate the left-hand side of \eqref{6.36} as
\begin{equation}
\label{6.39}
\Phi_N'(0)^\dagger\,\Phi_N(0)-\Phi_N(0)^\dagger\,\Phi_N'(0)=C_N\left(B^\dagger A-A^\dagger B\right) C_N.
\end{equation}
Using \eqref{2.6} on the right-hand side of \eqref{6.39} we obtain
\begin{equation*}
\Phi_N'(0)^\dagger\,\Phi_N(0)-\Phi_N(0)^\dagger\,\Phi_N'(0)=0,
\end{equation*}
which ensures that \eqref{6.36} holds.
With the help of \eqref{6.25} and \eqref{6.36}, we verify that \eqref{6.35} indeed holds, and hence $\mathcal A(x,y)$ satisfies \eqref{5.20}.
Thus, the proof of (c) is complete. The proof of (d) follows by comparing \eqref{5.21} and \eqref{6.24}
and by using the fact that the two expressions agree when we evaluate
$\mathcal A(x,x)$ by letting $y=x^-$ in \eqref{6.2}. Hence, the proof is complete.
\end{proof}

In the next theorem we show that the quantity $\tilde\varphi(k,x)$ in \eqref{5.18} with
$\mathcal A(x,y)$ as in \eqref{6.2} is the perturbed regular solution to \eqref{2.1}
with the potential $\tilde V(x)$ in \eqref{6.24} and
with an appropriate pair of boundary matrices $\tilde A$ and $\tilde B.$

\begin{theorem}\label{theorem6.3}
Consider the unperturbed Schr\"odinger operator associated with \eqref{2.1} and \eqref{2.5},
where the potential
$V$ satisfies \eqref{2.2} and belongs to
$L^1_1(\mathbb R^+)$
and the boundary matrices $A$ and $B$ appearing in \eqref{2.5} satisfy 
\eqref{2.6} and \eqref{2.7}. Let $\varphi(k,x)$ be the corresponding regular solution
satisfying the initial conditions
 \eqref{2.10}, 
each of $k=i\kappa_j$ for $1\le j\le N$ correspond to the bound state with the energy$-\kappa_j^2,$ 
the Gel'fand--Levitan normalization matrix $C_j,$ 
the orthogonal projection
$Q_j$ onto $\text{\rm{Ker}}[J(i\kappa_j)],$ 
and
the Gel'fand--Levitan normalized bound-state solution $\Phi_j(x).$
Let $\tilde\varphi(k,x)$ be the quantity given in \eqref{5.18} with $\mathcal A(x,y)$ expressed as in \eqref{6.2}.  
Then, we have the following:

\begin{enumerate}

\item[\text{\rm(a)}] The quantity  
$\tilde\varphi(k,x)$ is a solution to the perturbed Schr\"odinger equation \eqref{2.1}
 with the potential $\tilde V(x)$ given in \eqref{6.24}. We can express $\tilde\varphi(k,x)$ also as
 \begin{equation}\label{6.41}
\tilde\varphi(k,x)=\varphi(k,x)+ \Phi_{N}(x)\, W^+(x) \int_0^x dy\,\Phi_{N}(y)^\dagger\,\varphi(k,y).
\end{equation}

\item[\text{\rm(b)}] The perturbed potential $\tilde V(x)$ given in \eqref{6.24} 
satisfies \eqref{2.2}. Moreover, the potential increment $\tilde V(x)-V(x)$ has the asymptotic behavior
\begin{equation}\label{6.42}
\tilde V(x)-V(x)= O\left(q_6(x)\right), \qquad x \to +\infty,
\end{equation}
where we have defined $q_6(x)$ as
\begin{equation}
\label{6.43}
q_6(x):=\int_x^\infty dy \,|V(y)|.
\end{equation}
If the unperturbed potential $V$ is further restricted to $L^1_{1+\epsilon}(\mathbb R^+)$ for some fixed $\epsilon \ge 0,$ then 
the perturbed potential $\tilde V$ belongs to $L^1_\epsilon(\mathbb R^+).$

\item[\text{\rm(c)}] If the unperturbed potential $V(x)$ is further restricted to satisfy
\begin{equation}\label{6.44}
|V(x)| \le c\,e^{-\alpha x}, \qquad x \ge x_0,
\end{equation}
for some positive constants $\alpha$ and $x_0$ and with
$c$ denoting a generic constant, then the potential increment $\tilde V(x)-V(x)$ also satisfies
\begin{equation}\label{6.45}
|\tilde V(x)-V(x)| \le c\,e^{-\alpha x}, \qquad x \ge x_0.
\end{equation}

\item[\text{\rm(d)}]
If the support of the unperturbed potential $V$ is contained in the interval $[0,x_0]$
for some positive $x_0,$ then the support of 
the perturbed potential $\tilde V$ is also contained in $[0,x_0].$ 

\item[\text{\rm(e)}] For $ k \neq i \kappa_N,$  the perturbed quantity $\tilde\varphi(k,x)$ can be expressed as
\begin{equation}
\label{6.46}
\tilde\varphi(k,x)=\varphi(k,x)+\ds\frac{1}{k^2+\kappa_N^2}
\,\Phi_N(x) \,W_N(x)^+\left[\Phi'_N(x)^\dagger\,\varphi(k,x)-\Phi_N(x)^\dagger
\,\varphi'(k,x)\right].
\end{equation}

\item[\text{\rm(f)}] The perturbed quantity $\tilde\varphi(k,x)$ satisfies the initial conditions \eqref{5.17}
where the matrices $\tilde A$ and $\tilde B$
are expressed in terms of the unperturbed
boundary matrices $A$ and $B$ and
the Gel'fand--Levitan normalization matrix $C_N$ for the bound state at $k=i \kappa_N$ as
\begin{equation}\label{6.47}
\tilde A=A, \quad \tilde B=B+A\,C_N^2\, A^\dagger A.
\end{equation}

\item[\text{\rm(g)}] 
The matrices $\tilde A$ and $\tilde B$ appearing in \eqref{6.47} satisfy \eqref{2.6} and \eqref{2.7}. Hence, 
as a consequence of (a) and (f), 
the quantity $\tilde\varphi(k,x)$ corresponds to the regular solution 
to the matrix Schr\"odinger equation with the potential $\tilde V(x)$ in \eqref{6.24} 
and with the selfadjoint
boundary condition \eqref{2.5} with $A$ and $B$ there replaced with
$\tilde A$ and $\tilde B,$ respectively.

\end{enumerate}

\end{theorem}

\begin{proof} The proof of (a) is obtained by using
Theorems~\ref{theorem5.2} and \ref{theorem6.2}. 
For the proof of (b) we proceed as follows.
From \eqref{2.2} we know that $V(x)$ satisfies \eqref{2.2}.
Furthermore, from \eqref{6.3} we know that $W_N(x)$ is hermitian. Hence,
from \eqref{6.4} we see that $\tilde V(x)$ also satisfies \eqref{2.2}. The proof of \eqref{6.42}
is obtained as follows. From (3.2.3) and (3.2.30) of \cite{AW2021} we conclude that the unperturbed
Jost solution $f(k,x)$ satisfies
\begin{equation}\label{6.48}
 f(k,x)= e^{ikx }\left[I+ O\left(q_6(x)\right)\right],\qquad x \to +\infty.
 \end{equation}
Similarly, with the help of (3.2.7) of \cite{AW2021} we obtain
 \begin{equation}\label{6.49}
 f'(k,x)=ik\,  e^{ikx }\left[I+ O\left(q_6(x)\right)\right],\qquad x \to +\infty.
 \end{equation}
With the help of
\eqref{4.1}, \eqref{4.2}, \eqref{4.15}, \eqref{6.48}, \eqref{6.49},
and the first equality in \eqref{3.39},
we get
 \begin{equation}\label{6.50}
 \Phi_N(x)= e^{- \kappa_N x }\left[I+ O\left(q_6(x)\right) \right]\mathbf B_N^{-1/2}\, D_N\, Q_N ,\qquad x \to +\infty, \end{equation}
 \begin{equation}\label{6.51}
 \Phi'_N(x)= -\kappa_N \, e^{-\kappa_N x }\left[(I+ O\left(q_6(x)\right) \right] \mathbf B_N^{-1/2}\, D_N\, Q_N,\qquad x \to +\infty.
 \end{equation} 
We have the estimate
 \begin{equation}\label{6.52}
 \int_x^\infty dy\,  e^{-2\kappa_N y} \, O\left(q_6(y)\right) \le c  \int_x^\infty dy \, e^{-2\kappa_N y} 
q_6(y),\qquad x\to+\infty,
\end{equation}
for some generic constant $c.$
Using integration by parts on the right-hand side of \eqref{6.52}, we get
\begin{equation}
\label{6.53}
\int_x^\infty dy \, e^{-2\kappa_N y} 
q_6(y)
= \ds\frac{1}{2\,\kappa_N} \left[ e^{-2 \kappa_N  x}q_6(x)-
  \int_x^\infty dy\, e^{-2 \kappa_N y} \, |V(y)|\right]. 
\end{equation}
Using \eqref{6.52} and \eqref{6.53} we conclude that
  \begin{equation}\label{6.54}
 \int_x^\infty dy\,  e^{-2\kappa_N y} \, O\left(q_6(y)\right)\le
 \ds\frac{c}{2\,\kappa_N} \, e^{-2 \kappa_N  x} q_6(x),\qquad x\to+\infty.\end{equation}
With the help of \eqref{6.3}, \eqref{6.13}, \eqref{6.50}, and \eqref{6.54}, we obtain
 \begin{equation}\label{6.55}
 W_N(x)=\ds \frac{1}{2\,\kappa_N} e^{-2\kappa_N x}Q_N\left[D_N^\dagger \mathbf B_N^{-1} D_N +O\left(q_6(x)\right) \right]Q_N \oplus 0,
\qquad x\to+\infty,
 \end{equation}
which yields
\begin{equation}\label{6.56}
 W_N(x)=\ds \frac{1}{2\,\kappa_N} e^{-2\kappa_N x}Q_N\left[D_N^\dagger \mathbf B_N^{-1} D_N +O\left(q_6(x) \right) \right]Q_N.
\qquad x\to+\infty.
 \end{equation}
From \eqref{4.15} we conclude that the $n\times n$ matrix 
 $D_N^\dagger\, \mathbf B_N^{-1} D_N$ has the decomposition analogous to
the decomposition in \eqref{3.16}. Let us use
 $[D_N^\dagger\, \mathbf B_N^{-1} D_N]_1$ to denote the restriction of
$D_N^\dagger\, \mathbf B_N^{-1} D_N$ to the subspace  $Q_N \,\mathbb C^n.$
From Theorem~\ref{theorem4.3}(a) we know that
$D_N$ is an isometry from $\text{\rm{Ker}}[J(i\kappa_N)]$ to
$\text{\rm{Ker}}[J(i\kappa_N)^\dagger],$ and hence it is an isometry
from $Q_N\,\mathbb C^n$ to $P_n\,\mathbb C^n.$ Similarly, from
Theorem~\ref{theorem4.3}(b) we know that
$D_N^\dagger$ is an isometry from $\text{\rm{Ker}}[J(i\kappa_N)^\dagger]$ to
$\text{\rm{Ker}}[J(i\kappa_N)]$ 
and hence it is an isometry
from $P_N\,\mathbb C^n$ to $Q_n\,\mathbb C^n.$
We also already know that
$\mathbf B_N^{-1}$ is invertible and commutes with $P_N.$
Thus, $D_N^\dagger\, \mathbf B_N^{-1}  D_N$ is an invertible map 
from $Q_N\,\mathbb C^n$ to $P_n\,\mathbb C^n.$ Consequently, 
$[D_N^\dagger\, \mathbf B_N^{-1}  D_N]_1$
is invertible in $Q_N \,\mathbb C^n.$
Hence, we have
 \begin{equation}\label{6.57}
 \left(D_N^\dagger\, \mathbf B_N^{-1} D_N\right)^+= \left([D_N^\dagger\, \mathbf B_N^{-1} D_N]_1\right)^{-1}\oplus 0.
 \end{equation}
 With the help of \eqref{6.13}, \eqref{6.20}, \eqref{6.56}, and \eqref{6.57}, we get
\begin{equation}\label{6.58}
 W_N(x)^+= 2\,\kappa_N\, e^{2\,\kappa_N x}\,Q_N\left[[D_N^\dagger\,\mathbf B^{-1} D_N]_1^{-1}  +O\left(q_6(x)\right)\right]
Q_N\oplus 0,\qquad x\to+\infty.
 \end{equation}
Using \eqref{6.57} in \eqref{6.58}, after some simplification we obtain
\begin{equation}\label{6.59}
 W_N(x)^+= 2\,\kappa_N\, e^{2\,\kappa_N x}\left[\left(D_N^\dagger\,\mathbf B^{-1} D_N\right)^+  +O\left(q_6(x) \right)\right]
,\qquad x\to+\infty.
\end{equation}
Next, using  \eqref{6.25}, \eqref{6.56}, \eqref{6.59}, and the second equality of \eqref{3.3}, we get
\begin{equation}\label{6.60}
\left[W_N(x)^+\right]'= 4\,\kappa_N^2\,e^{2\,\kappa_N\, x} \left[\left(D_N^\dagger\, \mathbf B_N^{-1} D_N\right)^+  +  O\left(q_6(x)
\right)\right], \qquad x \to +\infty.
\end{equation}
Then, with the help of \eqref{6.50}, \eqref{6.51}, \eqref{6.59}, and \eqref{6.60}, we obtain
\begin{equation}\label{6.61}
2\,\ds\frac{d}{dx}\left[
\Phi_N(x) \,W_N(x)^+\,\Phi_N(x)^\dagger
\right]= O\left(q_6(x)\right), \qquad x \to +\infty.
\end{equation}
Using \eqref{6.61}, from \eqref{6.24} we observe that the proof of \eqref{6.42} is complete.
For any $\epsilon\ge 0,$ we have
\begin{equation}
\label{6.62}
\int_0^\infty dx\, (1+x)^\epsilon\, q_6(x)=\ds \frac{1}{1+\epsilon}\int_0^\infty dx\left( (1+x)^{1+\epsilon}-1\right)|V(x)|.
\end{equation}
Using \eqref{6.24}, \eqref{6.61}, and \eqref{6.62}, we conclude that
$\tilde V$ belongs to $L^1_\epsilon(\mathbb R^+)$ if $V$ belongs to $L^1_{1+\epsilon}(\mathbb R^+).$
This completes the proof of (b).
For the proof of (c) we proceed as follows.
If \eqref{6.44} holds, then \eqref{6.50} and \eqref{6.51} imply the respective asymptotics
  \begin{equation}\label{6.63}
 \Phi_N(x)= e^{- \kappa_N\, x }\left[I+ O\left(e^{-\alpha\,x}\right) \right]\mathbf B_N^{-1/2} D_N\, Q_N ,\qquad x \to +\infty,
 \end{equation}
 \begin{equation}\label{6.64}
 \Phi'_N(x)=-\kappa_N \, e^{-\kappa_N \,x }\left[I+ O\left(e^{-\alpha\, x}\right) \right] \mathbf B_N^{-1/2} D_N\, Q_N,\qquad x \to +\infty.
 \end{equation} 
Comparing \eqref{6.63} and \eqref{6.64} with
\eqref{6.50} and \eqref{6.51}, respectively, we observe that we can repeat the argument in \eqref{6.55}--\eqref{6.61}
by using $O(e^{-\alpha x})$ in place of $O\left(q_6(x)\right).$
Consequently, in analogy with \eqref{6.61} we get
 \begin{equation}\label{6.65}
\left | 2\,\ds\frac{d}{dx}\left[
\Phi_N(x) \,W_N(x)^+\,\Phi_N(x)^\dagger
\right]  \right | \le c\, e^{-\alpha x}, \qquad x \ge x_0,
\end{equation}
where $c$ is a generic constant.
Then, using \eqref{6.24}, \eqref{6.44}, \eqref{6.65}, we obtain \eqref{6.45}.
Thus, the proof of (c) is complete. To prove (d), we proceed as follows. 
If $V(x)=0$ for $x\ge x_0,$ then we have
\begin{equation}\label{6.66}
\Phi_N(x)=R\, e^{-\kappa_N x}, \qquad x \ge x_0,
\end{equation}
where $R$ is a constant $n\times n$ matrix. Using \eqref{6.66} in \eqref{6.3} we get
\begin{equation*}
W_N(x)=\ds\frac{1}{2\kappa_N}\,R^\dagger R\, e^{-2\kappa_N x}, \qquad x \ge x_0,
\end{equation*}
which yields
\begin{equation}\label{6.68}
W_N(x)^+=2\kappa_N(R^\dagger R)^+ e^{2\kappa_N x}, \qquad x \ge x_0.
\end{equation}
From \eqref{6.66} and \eqref{6.68} we have
\begin{equation}\label{6.69}
\Phi_N(x)^\dagger\,W_N(x)^+\,\Phi_N(x)=2\kappa_NR^\dagger(R^\dagger R)^+ R, \qquad x \ge x_0.
\end{equation}
Since the right-hand side in \eqref{6.69} is independent of $x,$ using \eqref{6.69} in \eqref{6.24} we obtain
\begin{equation}\label{6.70}
\tilde V(x)=V(x),\qquad x \ge x_0.
\end{equation}
We see that (d) follows from \eqref{6.70}.
For the proof of (e) we proceed as follows.
 From \eqref{5.15} and \eqref{6.37}, we respectively have
\begin{equation}
\label{6.71}
-\Phi_N(x)^\dagger\,\varphi''(k,x)+\Phi_N(x)^\dagger\,V(x)\,\varphi(k,x)
=k^2\,\Phi_N(x)^\dagger\,\varphi(k,x),
\end{equation}
\begin{equation}
\label{6.72}
-\Phi''_N(x)^\dagger\,\varphi(k,x)+\Phi_N(x)^\dagger\,V(x)\,\varphi(k,x)
=-\kappa_N^2\,\Phi_N(x)^\dagger\,\varphi(k,x).
\end{equation}
 From \eqref{6.71} and \eqref{6.72}, we get
\begin{equation}
\label{6.73}
\ds\frac{d}{dx}\left[
\Phi'_N(x)^\dagger\,\varphi(k,x)-
\Phi_N(x)^\dagger\,\varphi'(k,x)
\right]=\left(k^2+\kappa_N^2\right)
\Phi_N(x)^\dagger\,\varphi(k,x).
 \end{equation}
By integrating \eqref{6.73}, we obtain
\begin{equation}
\label{6.74}
\begin{split}
\Phi'_N(x)^\dagger\,\varphi(k,x)-
\Phi_N(x)^\dagger\,\varphi'(k,x)
-&\Phi'_N(0)^\dagger\,\varphi(k,0)+
\Phi_N(0)^\dagger\,\varphi'(k,0)
\\
& =
\left(k^2+\kappa_N^2\right)\ds\int_0^x
dy\,\Phi_N(y)^\dagger\,\varphi(k,y).
\end{split}
\end{equation}
Using \eqref{2.10} and \eqref{3.28}, with the help of \eqref{2.6} we get
\begin{equation}
\label{6.75}
\Phi'_N(0)^\dagger\,\varphi(k,0)-
\Phi_N(0)^\dagger\,\varphi'(k,0)=0.
\end{equation}
From \eqref{6.74} and \eqref{6.75} we have
\begin{equation}
\label{6.76}
\ds\int_0^x
dy\,\Phi_N(y)^\dagger\,\varphi(k,y)=
\ds\frac{1}{k^2+\kappa_N^2}
\left[\Phi'_N(x)^\dagger\,\varphi(k,x)-
\Phi_N(x)^\dagger\,\varphi'(k,x)
\right].
\end{equation}
Finally, using \eqref{6.76} in \eqref{6.41} we obtain \eqref{6.46}.
Thus, the proof of (e) is complete.
To prove (f), we proceed as follows. Using \eqref{3.29} when $j=N,$ from \eqref{6.3} we get
\begin{equation}
\label{6.77}
W_N(0)=Q_N,
\end{equation}
which implies
\begin{equation}
\label{6.78}
W_N(0)^+=Q_N,
\end{equation}
based on the fact that $Q_N$ is its own Moore--Penrose inverse.
Evaluating both sides of \eqref{6.46} at $x=0,$ with the help of the first equality of \eqref{2.10},
the first equality of \eqref{5.17}, and \eqref{6.75} we establish the first equality of \eqref{6.47}.
Taking the derivative of both sides of \eqref{6.46}, we get
\begin{equation}
\label{6.79}
\begin{split}
\tilde\varphi'(k,x)=&\varphi'(k,x)+\Phi_N(x)\,W_N(x)^+\,\Phi_N(x)^\dagger \,\varphi(k,x)
\\&+\ds\frac{\left[\Phi_N(x)\,W_n(x)^+\right]'}{k^2+\kappa_N^2}\left[\Phi_N'(x)^\dagger\,\varphi(k,x)-\Phi_N(x)^\dagger\,\varphi'(k,x)\right],
\end{split}
\end{equation}
where we have used \eqref{6.73}.
By evaluating both sides of \eqref{6.79} at $x=0,$ with the help of the second equality of \eqref{2.10},
the second equality of \eqref{5.17}, \eqref{6.75}, and \eqref{6.78}, we get
\begin{equation}
\label{6.80}
\tilde B=B+A \,C_N \,Q_N\, C_N^\dagger  A^\dagger A.
\end{equation}
Using the properties of the Gel'fand--Levitan normalization matrix $C_N$ given by
\begin{equation*}
C_N\, Q_N=C_N, \quad C_N^\dagger=C_N,
\end{equation*}
on the right-hand side of \eqref{6.80}, we obtain the second equality of \eqref{6.47}.
Hence, the proof of (f) is complete. 
For the proof of (g), we proceed as follows. Using \eqref{6.47}, after some minor simplification we obtain
\begin{equation}
\label{6.82}
\tilde A^\dagger \tilde B-\tilde B^\dagger\tilde A
= A^\dagger B- B^\dagger A.
\end{equation}
The right-hand side of \eqref{6.82} vanishes as a result of \eqref{2.6}.
Hence, \eqref{2.6} is satisfied if we replace
$A$ and $B$ by $\tilde A$ and $\tilde B,$ respectively.
For any column vector $v$ in $\mathbb C^n,$ we have
\begin{equation}\label{6.83}
v^\dagger\left(\tilde A^\dagger \tilde A+\tilde B^\dagger \tilde B\right)v=|Av|^2+|Bv|^2.
\end{equation}
The right-hand side of \eqref{6.83} shows that the matrix
$\tilde A^\dagger\tilde A+ \tilde B^\dagger \tilde B$ is nonnegative. 
Hence, in order to prove that $\tilde A^\dagger\tilde A+ \tilde B^\dagger \tilde B$ is a positive matrix, it is enough to show
that its kernel contains only the zero vector.
We prove this as follows. From \eqref{6.47} we obtain
\begin{equation}\label{6.84}
\tilde A^\dagger \tilde A+ \tilde B^\dagger \tilde B= A^\dagger A+ B^\dagger B+ (A C_N^2 A^\dagger A)^\dagger  (A C_N^2 A^\dagger A) + B^\dagger  (A C_N^2 A^\dagger A)+ (A C_N^2 A^\dagger A )^\dagger B.
\end{equation}
For any column vector $v$ in $\mathbb C^n,$ from \eqref{6.84} we get
\begin{equation}\label{6.85}
 v^\dagger (\tilde A^\dagger \tilde A+ \tilde B^\dagger \tilde B) v=| Av|^2+| Bv|^2+
| (A C_N^2 A^\dagger A)v|^2+ 2\, \text{\rm{Re}}[\langle Bv,  (A C_N^2 A^\dagger A)v\rangle],
\end{equation}
where the last term on the right-hand side involves the real part of the standard scalar product.
With the help of a Cauchy--Schwarz inequality, we have
\begin{equation}
\label{6.86}
-2\,|Bv| \left|(A C_N^2 A^\dagger A)v\right|\le 
2\, \text{\rm{Re}}[\langle Bv,  (A C_N^2 A^\dagger A)v\rangle] \le 2\,|Bv| \left|(A C_N^2 A^\dagger A)v\right|.
\end{equation}
From \eqref{6.86} we obtain
\begin{equation}
\label{6.87}
-2\, \text{\rm{Re}}[\langle Bv,  (A C_N^2 A^\dagger A)v\rangle] \ge -2\,|Bv| \left|(A C_N^2 A^\dagger A)v\right|.
\end{equation}
If $(\tilde A^\dagger \tilde A+ \tilde B^\dagger \tilde B) v=0,$ then
from \eqref{6.85} and \eqref{6.87} we get
\begin{equation*}
|Av|^2+ \left( |Bv|-|(A C_N^2 A^\dagger A)v|   \right)^2\le 0,
\end{equation*}
which implies that $Av=0$ and $Bv=0.$ From \eqref{2.7}
we know that $A^\dagger A+ B^\dagger B >0,$ and hence
the simultaneous equalities $Av=0$ and $Bv=0$ can hold only when $v$ is the zero column vector.
Thus, the proof of (g) is complete.
\end{proof}

In the next theorem, we describe how the Jost matrix, the scattering 
matrix, and the Jost solution transform when we change the potential, the boundary matrices, and the regular solution
as described in Theorem~\ref{theorem6.3}. The theorem also establishes the fact that such a transformation only removes the bound state 
with the energy $-\kappa_N^2,$  leaves the remaining bound states  and their multiplicities unchanged, and introduces no new bound states.
We note that, in the theorem, we impose the stronger assumption
that the unperturbed potential $V$ belongs to
$L^1_2(\mathbb R^+)$ rather than $L^1_1(\mathbb R^+)$ so that the perturbed potential $\tilde V$  
belongs to $L^1_1(\mathbb R^+),$ which is assured by Theorem~\ref{theorem6.3}(b).
We remark that this is a sufficiency assumption because the asymptotic estimates 
in Theorem~\ref{theorem6.3}(b)
on the potentials
are not necessarily sharp.

\begin{theorem}
\label{theorem6.4} Consider
the unperturbed Schr\"odinger operator with the potential
$V$ satisfying \eqref{2.2} and belonging to $L^1_2(\mathbb R^+),$
the selfadjoint
boundary condition \eqref{2.5} described by the boundary matrices $A$ and $B$ satisfying
\eqref{2.6} and \eqref{2.7}, the regular solution $\varphi(k,x)$ satisfying the initial
conditions \eqref{2.10}, the Jost solution $f(k,x)$ satisfying \eqref{2.9}, and with 
$N$ bound states with  the energies $-\kappa_j^2,$ 
the Gel'fand--Levitan normalization matrices $C_j,$ the orthogonal projections
$Q_j$ onto $\text{\rm{Ker}}[J(i\kappa_j)],$ the orthogonal projections
$P_j$ onto $\text{\rm{Ker}}[J(i\kappa_j)^\dagger],$ 
the Gel'fand--Levitan normalized bound-state solutions $\Phi_j(x),$ the Marchenko normalized
bound-state solutions $\Psi_j(x),$ the dependency matrices $D_j,$ and the multiplicities $m_j$
for the bound states for $1\le j\le N.$ 
Let us use a tilde to identify the quantities
associated with the perturbed Schr\"odinger operator
with the potential
$\tilde V(x)$ expressed as in \eqref{6.24}
and the boundary condition \eqref{2.5} where
the boundary matrices $A$ and $B$ are replaced with
$\tilde A$ and $\tilde B,$ respectively, given in \eqref{6.47}.
Thus, 
we use $\tilde f(k,x)$ for the perturbed Jost solution
satisfying the asymptotics \eqref{2.9},
 $\tilde\varphi(k,x)$ for the perturbed regular solution satisfying the initial
conditions \eqref{5.17}, $\tilde J(k)$ for the perturbed Jost matrix
defined as
\begin{equation}
\label{6.89}
\tilde J(k):=\tilde f(-k^*,0)^\dagger\,\tilde B-\tilde f'(-k^*,0)^\dagger\,\tilde A,\qquad k\in\bCpb,
\end{equation}
and we use $\tilde S(k)$ for the perturbed scattering matrix defined as
\begin{equation}
\label{6.90}
\tilde S(k):=-\tilde J(-k)\,\tilde J(k)^{-1},\qquad k\in\bR.
\end{equation}
We then have the following:

\begin{enumerate}

\item[\text{\rm(a)}] The unperturbed Jost matrix $J(k)$ is transformed into the perturbed Jost matrix $\tilde J(k)$ as
 \begin{equation}
\label{6.91}
\tilde J(k)=\left[I+\ds\frac{2i\kappa_N}{k-i\kappa_N}\,P_N\right] J(k),\qquad
k\in\overline{\mathbb C^+}.\end{equation}

\item[\text{\rm(b)}] The perturbation does not change
the matrix product $J(k)^\dagger J(k)$ for $k\in\mathbb R,$ 
i.e. we have
 \begin{equation}
\label{6.92}
\tilde J(k)^\dagger \tilde J(k)=J(k)^\dagger J(k),\qquad k\in\mathbb R,\end{equation}
and hence the continuous part of the spectral measure $d\rho$ does not change
under the perturbation.

\item[\text{\rm(c)}] Under the perturbation, the determinant of the Jost matrix is transformed as
 \begin{equation}
\label{6.93}
\det[\tilde J(k)]=\left(\ds\frac{k+i\kappa_N}{k-i\kappa_N}\right)^{m_N} \det[J(k)],\qquad
k\in\overline{\mathbb C^+},\end{equation}
where we recall
that $m_N$ is the multiplicity of the bound state at $k=i\kappa_N$ for the unperturbed
Schr\"odinger operator.

\item[\text{\rm(d)}] Under the perturbation,
the bound state with the energy $-\kappa_N^2$ is removed without adding any new bound states in such a way that
the remaining bound states with the energies $-\kappa_j^2$ and multiplicities $m_j$ for $1\le j\le N-1$
are unchanged.

\item[\text{\rm(e)}] The perturbed scattering matrix $\tilde S(k)$ is related to the
unperturbed scattering matrix $S(k)$ as
 \begin{equation}
\label{6.94}
\tilde S(k)=\left[I-\ds\frac{2i\kappa_N}{k+i\kappa_N}\,P_N\right] S(k)\left[I-\ds\frac{2i\kappa_N}{k+i\kappa_N}\,P_N\right],
\qquad k\in\mathbb R.\end{equation}

\item[\text{\rm(f)}] Under the perturbation, the determinant of the scattering matrix  is transformed as 
 \begin{equation}
\label{6.95}
\det[\tilde S(k)]=\left(\ds\frac{k-i\kappa_N}{k+i\kappa_N}\right)^{2 m_N} \det[S(k)],
\qquad k\in\mathbb R.
\end{equation}

\item[\text{\rm(g)}] 
For $k \in \overline{\mathbb C^+},$ the perturbed Jost solution $\tilde f(k,x)$ is related to the
unperturbed Jost solution $f(k,x)$ as
\begin{equation}
\label{6.96}
\tilde f(k,x)=\left[ f(k,x)  +\ds\frac{1}{k^2+\kappa_N^2}
\,\Phi_N(x) \,W_N(x)^+\,q_7(k,x)\right]
\left[ I+ \ds\frac{2i\kappa_N}{k-i\kappa_N} P_N \right],
\end{equation}
where we have defined
\begin{equation}
\label{6.97}
q_7(k,x):=\Phi'_N(x)^\dagger\,f(k,x)-\Phi_N(x)^\dagger
\,f'(k,x).
\end{equation}
The transformation in \eqref{6.96} can equivalently be expressed as
\begin{equation}
\label{6.98}
\tilde f(k,x)=\left[ f(k,x)-\Phi_N(x) \,W_N(x)^+\ds\int_x^\infty dy\,\Phi_N(y)^\dagger\, f(k,y)\right]
\left[ I+ \ds\frac{2i\kappa_N}{k-i\kappa_N} \,P_N \right].
\end{equation}

\end{enumerate}
\end{theorem}

\begin{proof}
From the $x$-derivative of \eqref{2.15} we have
\begin{equation}
\label{6.99}
\varphi'(k,x)=\ds\frac{1}{2ik}\left[f'(k,x)\,J(-k)-f'(-k,x)\,J(k)\right],
\qquad k\in\mathbb R.
\end{equation}
The analog of \eqref{2.15} for the perturbed problem yields
\begin{equation}
\label{6.100}
\tilde\varphi(k,x)=\ds\frac{1}{2ik}\left[\tilde f(k,x)\,\tilde J(-k)-\tilde f(-k,x)\,\tilde J(k)\right],
\qquad k\in\mathbb R.
\end{equation}
Using \eqref{2.15}, \eqref{6.99}, and \eqref{6.100} in \eqref{6.46}, we obtain
\begin{equation}
\label{6.101}
\Delta(k,x)=\Delta(-k,x),
\qquad k\in\mathbb R,
\end{equation}
where we have defined
\begin{equation}
\label{6.102}
\Delta(k,x):=\tilde f(k,x)\,\tilde J(-k)-f(k,x)\,J(-k)
-\ds\frac{\Phi_N(x)\,W_N(x)^+}{k^2+\kappa_N^2}
\,q_7(k,x)\,
J(-k),
\qquad k\in\mathbb R.
\end{equation}
As we will show below, 
from \eqref{6.102} we get
\begin{equation}
\label{6.103}
\Delta(k,x)=e^{ikx}\,q_8(-k)\left[I+o(1)\right],
\qquad  x \to +\infty, \quad k\in\mathbb R,
\end{equation}
where the $n\times n$ matrix $q_8(k)$ is defined as
\begin{equation}
\label{6.104}
q_8(k):=\tilde J(k)-J(k)-\ds\frac{2i\kappa_N}{k-i\kappa_N}\,P_N\,J(k),
\qquad k\in\mathbb R.
\end{equation}
Note that \eqref{6.101} and \eqref{6.103} yield
\begin{equation*}
e^{ikx}\, q_8(-k) [ I+o(1)]= e^{-ikx} \,q_8(k) \left[I+o(1)\right], \qquad  x \to +\infty, \quad k\in\mathbb R,
\end{equation*}
from which we conclude that $q_8(k)\equiv 0.$ Hence, \eqref{6.104} yields \eqref{6.91} for
$k\in\mathbb R.$ Using the fact \cite{AW2021} that
the Jost matrix $J(k)$ has an analytic extension from $k\in\mathbb R$ to the upper-half complex plane
$\mathbb C^+,$ we conclude that \eqref{6.91} holds for $k\in\overline{\mathbb C^+}.$
Thus, the proof of (a) is complete, provided we show that \eqref{6.103} holds with $q_8(k)$ as in \eqref{6.104}.
To establish \eqref{6.103} we proceed as follows.
From \eqref{4.1}, when $j=N$ we get
\begin{equation}
\label{6.106}
\Phi_N(x)=\Psi_N(x)\,D_N,
\end{equation}
and from the first equality of \eqref{3.39} we have
\begin{equation}
\label{6.107}
\Psi_N(x)=f(i\kappa_N,x)\, \mathbf B_N^{-1/2}\,P_N.
\end{equation}
Using \eqref{6.107} in \eqref{6.106} we obtain
\begin{equation}
\label{6.108}
\Phi_N(x)=f(i\kappa_N,x)\, \mathbf B_N^{-1/2}\,P_N\,D_N,
\end{equation}
and using \eqref{2.9} in \eqref{6.108} we get
\begin{equation}
\label{6.109}
\Phi_N(x)=e^{-\kappa_N x}\,\mathbf B_N^{-1/2}\,P_N \left[I+o(1)\right], \qquad x\to+\infty,
\end{equation}
\begin{equation}
\label{6.110}
\Phi'_N(x)=e^{-\kappa_N x} \,\mathbf B_N^{-1/2}\,P_N\,D_N\left[- \kappa_N I+o(1)\right],
\qquad x\to+\infty.
\end{equation}
By taking the matrix adjoints, from \eqref{6.109} and \eqref{6.110} we obtain the respective asymptotics
\begin{equation}
\label{6.111}
\Phi_N(x)^\dagger=  e^{-\kappa_N x}\, D_N^\dagger\,P_N\,\mathbf B_N^{-1/2}\left[I+o(1)\right],\qquad
x\to+\infty,
\end{equation}
\begin{equation}
\label{6.112}
\Phi'_N(x)^\dagger=   e^{-\kappa_N x}\, D_N^\dagger\,P_N\,\mathbf B_N^{-1/2}\left[-\kappa_N I+o(1)\right],
\qquad x\to+\infty,
\end{equation}
where we have used the fact that each of $P_N,$ $Q_N,$ and $\mathbf B_N^{-1/2}$ is an $n\times n$ selfadjoint matrix.
By taking the $x$-derivative of both sides of \eqref{6.3} and then using \eqref{6.109}--\eqref{6.112} in the resulting
equality, we obtain
\begin{equation}
\label{6.113}
W_N'(x)=-   e^{-2\kappa_N x} D_N^\dagger P_N\,\mathbf B_N^{-1} P_N D_N\left[I+o(1)\right] ,\qquad x\to+\infty.
\end{equation}
Using \eqref{4.2} in \eqref{6.113} we get
\begin{equation}
\label{6.114}
W_N'(x)=- e^{-2\kappa_N x}  D_N^\dagger\,\mathbf B_N^{-1} D_N\left[I+o(1)\right],\qquad x\to+\infty.
\end{equation}
From \eqref{6.3} we have
\begin{equation}
\label{6.115}
W_N(x)=  o(1),\qquad x\to+\infty,
\end{equation}
and hence \eqref{6.114} and \eqref{6.115} imply
\begin{equation}
\label{6.116}
W_N(x)=\ds  \frac{1}{2\kappa_N}\,e^{-2\kappa_N x}  D_N^\dagger\,\mathbf B_N^{-1} D_N\left[I+o(1)\right] ,\qquad x\to+\infty.
\end{equation}
From \eqref{6.3} with the help of the second equality in \eqref{3.39}, we get
\begin{equation*}
W_N(x)= Q_N \,W_N(x)\,Q_N.
\end{equation*}
Using  \eqref{4.15} and after premultiplying and postmultiplying both sides of \eqref{6.116}  by $Q_N,$ we obtain
\begin{equation}\label{6.118}
W_N(x)= \ds  \frac{1}{2\kappa_N}\,e^{-2\kappa_N x} D_N^\dagger\,\mathbf B_N^{-1} D_N\left[I+\displaystyle Q_N \,o(1)\, Q_N\right] ,\qquad x\to+\infty.
\end{equation}
From \eqref{6.118} we conclude that
$[W_N(x)]_1$ defined in \eqref{6.13} has the asymptotics
\begin{equation*}
[W_N(x)]_1=\ds  \frac{1}{2\kappa_N}\,e^{-2\kappa_N x} D_N^\dagger\,\mathbf B_N^{-1} D_N\left[I+ \displaystyle Q_N \,o(1)  \,Q_N\right] ,\qquad x\to+\infty,
\end{equation*}
and with the help of \eqref{6.20} we obtain
\begin{equation}
\label{6.120}
W_N(x)^+=2\kappa_N\,e^{2\kappa_N x}\left(D_N^\dagger\,\mathbf B_N^{-1} D_N\right)^+\left[I+Q_N\,o(1)\,Q_N\right]\oplus 0,
\qquad x\to+\infty.
\end{equation}
From \eqref{6.120} we get
\begin{equation}
\label{6.121}
W_N(x)^+=
2\kappa_N\,e^{2\kappa_N x}\left(D_N^\dagger\,\mathbf B_N^{-1} D_N\right)^+\left[I+o(1) \right],
\qquad x\to+\infty,
\end{equation}
where we recall that the superscript $+$ denotes the Moore--Penrose inverse.
With the help of \eqref{2.9}, \eqref{6.111}, and \eqref{6.112}, as $x\to+\infty$ we obtain
\begin{equation}
\label{6.122}
\Phi_N'(x)^\dagger\,f(k,x)-\phi_N(x)^\dagger\,f'(k,x)=-(\kappa_N+ik)\,e^{ikx-\kappa_N x} D_N^\dagger P_N \mathbf B_N^{-1/2}\left[I+o(1)\right].
\end{equation}
Finally, using \eqref{2.9} for $f(k,x),$ the analog of \eqref{2.9} for $\tilde f(k,x),$ \eqref{6.109}, \eqref{6.110}, \eqref{6.121}, and \eqref{6.122}
on the right-hand side of \eqref{6.102}, we establish \eqref{6.103} with $q_8(k)$ defined as in \eqref{6.104}, provided we can show that 
\begin{equation*}
\mathbf B_N^{-1/2} P_N\,D_N\left(D_N^\dagger\,\mathbf B_N^{-1} D_N\right)^+ D_N^\dagger P_N \mathbf B_N^{-1/2}=P_N,
\end{equation*}
or equivalently if we can show that
\begin{equation}
\label{6.124}
P_N\,\mathbf B_N^{-1/2} D_N\left(D_N^\dagger\,\mathbf B_N^{-1} D_N\right)^+ D_N^\dagger \mathbf B_N^{-1/2} P_N=P_N,
\end{equation}
because $P_N$ commutes with $\mathbf B_N^{-1/2}.$
We establish \eqref{6.124} as follows. We know from Section~\ref{section3} that the matrix $\mathbf B_N^{-1/2}$ is selfadjoint, and hence we have
\begin{equation}
\label{6.125}
D_N^\dagger \mathbf B_N^{-1} D_N=
\left(\mathbf B_N^{-1/2} D_N\right)^\dagger \left(\mathbf B_N^{-1/2} D_N\right).
\end{equation}
For the Moore--Penrose inverse of the product $\mathbf B_N^{-1/2} D_N,$ we have
\begin{equation}\label{6.126}
\left(\mathbf B_N^{-1/2} D_N\right)^+ = D_N^+\,\mathbf B_N^{1/2},
\end{equation}
which can be verified by showing that the four equalities in \eqref{3.3} hold if we let 
$\mathbf M=\mathbf B_N^{-1/2} D_N$ and 
$\mathbf M^+=D_N^+\,\mathbf B_N^{1/2}$ there.
The verification of the first equality in \eqref{3.3} is obtained with the help of \eqref{4.18} and the first equality of \eqref{4.20}.
The second equality of \eqref{3.3} is verified with the help of \eqref{4.18} and
the second equality of \eqref{4.20}.
The verification of the third equality of \eqref{3.3} is obtained by using \eqref{4.17}, \eqref{4.18}, the fact that
$P_N$ commutes with $\mathbf B_N^{1/2}$ and $\mathbf B_N^{-1/2},$ and that the matrix product
$\mathbf M \mathbf M^+$ becomes equal to
the selfadjoint matrix $P_N.$ The fourth equality of \eqref{3.3} is verified
by using \eqref{4.16}, \eqref{4.18}, and the fact that $\mathbf M^+ \mathbf M$ becomes equal to the selfadjoint matrix
$Q_N.$ By expressing the right-hand side of \eqref{6.125} and then the left-hand side of \eqref{6.124}
in terms of
$\mathbf M$ and $\mathbf M^\dagger,$ we see that \eqref{6.124} is equivalent to
\begin{equation}
\label{6.127}
P_N\,\mathbf M\left( \mathbf M^\dagger \,\mathbf M\right)^+\mathbf M^\dagger P_N=P_N.
\end{equation} 
Using \eqref{3.5}, we observe that
\eqref{6.127} is equivalent to
\begin{equation*}
P_N\,\mathbf M\,\mathbf M^+ \left(\mathbf M^\dagger\right)^+ \mathbf M^\dagger P_N=P_N,
\end{equation*} 
which, after using \eqref{3.4}, can also be written as
\begin{equation}
\label{6.129}
P_N\,\mathbf M\,\mathbf M^+ \left(\mathbf M\, \mathbf M^+\right)^\dagger P_N=P_N.
\end{equation} 
Using $\mathbf M\,\mathbf M^+=P_N$ and the fact that $P_N$ is an orthogonal projection, we
see that \eqref{6.129} indeed holds. Thus, the proof of (a) is complete.
 For the proof of (b) we proceed as follows.
 By taking the matrix adjoint of \eqref{6.91}, we obtain
 \begin{equation}
\label{6.130}
\tilde J(k)^\dagger=J(k)^\dagger \left[I+\ds\frac{2i\kappa_N}{k-i\kappa_N}\,P_N\right]^\dagger.
\end{equation}
 Since $\kappa_N$ is positive and $P_N$ is an orthogonal projection, we can directly verify that
  \begin{equation}
\label{6.131}
 \left[I+\ds\frac{2i\kappa_N}{k-i\kappa_N}\,P_N\right]^\dagger=\left[I+\ds\frac{2i\kappa_N}{k-i\kappa_N}\,P_N\right]^{-1}=
 \left[I-\ds\frac{2i\kappa_N}{k+i\kappa_N}\,P_N\right].
\end{equation}
Using \eqref{6.91} and \eqref{6.130} on the left-hand side of \eqref{6.92}, with the help of
\eqref{6.131} we establish \eqref{6.92}.
Hence, the proof of (b) is complete.
We now turn to the proof of (c). Since the orthogonal projection matrix $P_N$ has $m_N$ eigenvalues
equal to $+1$ and the remaining eigenvalues all equal to zero, we have 
\begin{equation}
\label{6.132}
 \det\left[I+\ds\frac{2i\kappa_N}{k-i\kappa_N}\,P_N\right]=\left( \ds\frac{k+i\kappa_N}{k-i\kappa_N}\right)^{m_N}.
\end{equation}
Using \eqref{6.91} and \eqref{6.132} we obtain \eqref{6.93}. Thus, the proof of (c) is complete. 
From Theorems~3.11.1 and 3.11.16 of \cite{AW2021} we know that
the bound states for the perturbed and unperturbed Schr\"odinger operators correspond to the zeros of 
$\textrm{det} [\tilde J(k)]$ and $\textrm{det} [{J}(k)],$ respectively, on the positive imaginary axis and the multiplicities of the bound states are determined
 by the order of the zeros of $\textrm{det} [\tilde J(k)]$ and $\textrm{det} [{J}(k)],$
respectively.  Hence, (d) follows from \eqref{6.93}.
We establish (e) 
by using \eqref{6.91} in \eqref{6.90}
and by simplifying the right-hand side of the resulting equation with the help of \eqref{2.12} and the second equality of \eqref{6.131}.
We obtain the proof of (f) by taking the determinant of both sides of \eqref{6.94}
and using the analog of \eqref{6.132} when $\kappa_N$ there is replaced with $-\kappa_N.$
Finally, we prove (g) as follows. 
Using \eqref{2.1} for $f(k,x)$ and the adjoint of \eqref{6.37}, we obtain
\begin{equation*}
\Phi''_N(x)^\dagger \, f(k,x)- \Phi_N(x)^\dagger\, f''(k,x)= (k^2+\kappa_N^2) \,\Phi_N(x)^\dagger \,f(k,x),
\end{equation*}
which is equivalent to
\begin{equation}
\label{6.134}
q'_7(k,x)
= (k^2+\kappa_N^2)\, \Phi_N(x)^\dagger\, f(k,x),
\end{equation}
where $q'_7(k,x)$ is the $x$-derivative of the quantity $q_7(k,x)$ defined in \eqref{6.97}.
Integrating \eqref{6.134} and using the asymptotics listed in \eqref{2.9}, \eqref{6.50}, and \eqref{6.51}, we
get
\begin{equation}
\label{6.135}
\ds\frac{1}{k^2+\kappa_N^2}\,
q_7(k,x)
=-\int_x^\infty dy\,\Phi_N(y)^\dagger\, f(k,y).
\end{equation}
We introduce the quantity $\tilde h(k,x)$ as
\begin{equation}
\label{6.136}
\tilde h(k,x):=f(k,x)  +\ds\frac{1}{k^2+\kappa_N^2}
\,\Phi_N(x) \,W_N(x)^+\,q_7(k,x),
\end{equation}
which differs from the right-hand side of \eqref{6.96} by the second term in the brackets there.
In other words, the quantity $\tilde h(k,x)$ is related to the quantity
$\tilde f(k,x)$ appearing in \eqref{6.96} as 
\begin{equation}\label{6.137}
\tilde f(k,x)=\tilde h(k,x)  \left[ I+ \frac{2i\kappa_N}{k-i\kappa_N} \,P_N \right].
\end{equation}
From \eqref{6.135} and \eqref{6.136} we observe that $\tilde h(k,x)$ does not have a singularity
at $k=i\kappa_N.$
We would like to show that 
$\tilde h(k,x)$ satisfies the Schr\"odinger equation \eqref{2.1} with the potential $\tilde V(x)$ given in \eqref{6.24}. 
For this we proceed as follows.
The $x$-derivative of \eqref{6.136}, after using \eqref{6.25} and \eqref{6.134}, is evaluated as
\begin{equation}
\label{6.138}
\tilde h'(k,x)= f'(k,x)+\ds \frac{\Phi'_N(x) \,W_N(x)^+\, q_7(k,x)}{k^2+\kappa_N^2} 
+\Phi_N(x) \,W_N(x)^+\, \Phi_N(x)^\dagger\, \tilde h(k,x).
\end{equation}
By taking the $x$-derivative of \eqref{6.138}, we obtain
the expression for $\tilde h''(k,x).$ Using that expression for $\tilde h''(k,x),$ the expression for $f''(k,x)$ obtained by using $f(k,x)$ in
\eqref{2.1}, the expression for $\Phi_N''(x)$ obtained from \eqref{6.37}, the expression for $\tilde h'(k,x)$ from
\eqref{6.138}, the expression for $q'_7(k,x)$ from \eqref{6.134},  and the expression for $\tilde V(x)$
from \eqref{6.24}, 
we prove that
\begin{equation}
\label{6.139}
-\tilde h''(x)+\tilde V(x)\,\tilde h(k,x)-k^2\,\tilde h(k,x)=0,
\end{equation}
which confirms that 
$\tilde h(k,x)$ indeed satisfies the Schr\"odinger equation \eqref{2.1} when
$V(x)$ there is replaced with $\tilde V(x)$ in \eqref{6.24}.
Next, using \eqref{2.9} and \eqref{6.109}--\eqref{6.112}, we obtain the asymptotics
\begin{equation}\label{6.140}
\tilde h(k,x)= e^{ikx} \left[ I- \frac{2i\kappa_N}{k+i\kappa_N} P_N \right] \left(I+o(1)\right), \qquad x \to +\infty.
\end{equation}
From \eqref{6.137}, \eqref{6.139}, and \eqref{6.140}, we see that the quantity $\tilde f(k,x)$ appearing in \eqref{6.96}
is indeed the Jost solution to the Schr\"odinger equation \eqref{2.1} when the
potential $V(x)$ there is replaced with the perturbed potential $\tilde V(x).$ Using \eqref{6.132}, \eqref{6.135}, and \eqref{6.136}, we see that
the perturbed Jost solution $\tilde f(k,x)$ has the alternate expression given in \eqref{6.98}.
Thus, the proof of (g) is complete.
\end{proof}

In the following remark, we emphasize that the perturbed Jost solution $\tilde f(k,x)$ appearing in \eqref{6.96} and \eqref{6.98} does not have a singularity
at $k=i\kappa_N$ despite the factor $k-i\kappa_N$ appearing in the denominators in \eqref{6.96} and \eqref{6.98}.

\begin{remark}\label{remark6.5}
{\rm
As stated in the proof of Theorem~\ref{theorem6.4}, the quantity $\tilde h(k,x)$ introduced in \eqref{6.136}
does not have a singularity at $k=i\kappa_N.$ On the other hand, \eqref{6.137} suggests that the perturbed Jost solution
$\tilde f(k,x)$ might have a pole at $k=i\kappa_N$ due to the factor $k-i\kappa_N$ in the denominator on the right-hand side of
\eqref{6.137}. To prove that the factor $k-i\kappa_N$ in the denominator in \eqref{6.137} yields a removable singularity, we proceed as follows.
From \eqref{6.140} we have
\begin{equation*}
\tilde h(i\kappa_N,x)= e^{-\kappa_N x} \left[ I- P_N \right]\left(I+o(1)\right),\qquad x \to +\infty,
\end{equation*}
which yields
\begin{equation}\label{6.142}
\tilde h(i\kappa_N,x)\,P_N= e^{-\kappa_N x} \,o(1),\qquad x \to +\infty.
\end{equation}
From Proposition~3.2.2 of \cite{AW2021} it follows that that \eqref{2.1} with the potential
$V(x)$ there replaced with $\tilde V(x),$ for each fixed $k\in\overline{\mathbb C^+}\setminus\{0\},$ has an $n\times n$ matrix solution $\tilde g(k,x)$ satisfying
\begin{equation}
\label{6.143}
\tilde g(k,x)=e^{-ikx}\left[I+o(1)\right],\quad \tilde g'(k,x)=-e^{-ikx}\left[ik\,I+o(1)\right],
\qquad x\to+\infty.
\end{equation}
From \eqref{6.143} we see that $\tilde g(i\kappa_N,x)$ increases exponentially as $x\to+\infty$ because it has the asymptotics
\begin{equation}\label{6.144}
\tilde g(i\kappa_N,x)= e^{\kappa_Nx}\left(I+o(1)\right),\qquad x \to+\infty.
\end{equation} 
In contrast to $\tilde g(i\kappa_N,x),$ the perturbed Jost solution $\tilde f(i\kappa_N,x)$
decreases exponentially as $x\to+\infty$ because it has the asymptotics
\begin{equation}\label{6.145}
\tilde f(i\kappa_N,x)= e^{-\kappa_Nx}\left(I+o(1)\right),\qquad x \to+\infty.
\end{equation} 
From Proposition~3.2.2 of \cite{AW2021} we know that $\tilde f(i\kappa_N,x)$
and $\tilde g(i\kappa_N,x)$ form a fundamental set of solutions to \eqref{2.1} at $k=i\kappa_N$ with the potential $\tilde V(x).$
By expressing $\tilde h(i\kappa_N,x)\,P_N$ as a linear combination
of $\tilde f(i\kappa_N,x)$ and $\tilde g(i\kappa_N,x),$ from \eqref{6.144} and \eqref{6.145} we conclude that
there exists a constant $n\times n$ matrix $M$
such that
\begin{equation}\label{6.146}
\tilde h(i\kappa_N,x)\,P_N= \tilde f(i\kappa_N,x) \,M.
\end{equation}
From \eqref{6.142}, \eqref{6.145}, and \eqref{6.146}, we see that we must have $M=0$ in \eqref{6.146}.
Hence, \eqref{6.146} yields
\begin{equation}\label{6.147}
\tilde h(i\kappa_N,x)\,P_N=0.
\end{equation}
Using \eqref{6.147} we can write \eqref{6.137} in the equivalent form as
\begin{equation}\label{6.148}
\tilde f(k,x)= \tilde h(k,x) + \ds \frac{2i\kappa_N}{k-i\kappa_N} \left (\tilde h(k,x)- \tilde h(i\kappa_N,x)\right)P_N.
\end{equation}
Since we have
\begin{equation*}
\tilde h(k,x)- \tilde h(i\kappa_N,x)=O(k-i\kappa_N),\qquad k\to i\kappa_N,
\end{equation*}
it follows that the right-hand side of \eqref{6.148} remains bounded as $ k \to i\kappa_N.$ 
Consequently, the singularity of \eqref{6.137} at $k=i\kappa_N$ is a removable
singularity.

}
\end{remark}

In the next  theorem, when we remove a bound state from the Schr\"odinger operator, we show that the projection matrices and
the Gel'fand--Levitan normalization matrices associated with the remaining bound states stay unaffected. We also show how the 
spectral measure changes when a bound state is removed. We recall that the bound-state
energies are not required to be ordered in an increasing or decreasing order, and hence without any loss of generality
we assume that we remove the bound state at $k=i\kappa_N$ with the Gel'fand--Levitan normalization matrix
$C_N$ and that the remaining bound states occur at $k=i\kappa_j$ with
the associated Gel'fand--Levitan normalization matrices $C_j$ for $1\le j\le N-1.$
As in Theorem~\ref{theorem6.4}, we impose the stronger assumption
that the unperturbed potential $V$ belongs to
$L^1_2(\mathbb R^+)$ rather than $L^1_1(\mathbb R^+)$ so that the perturbed potential $\tilde V$  
belongs to $L^1_1(\mathbb R^+),$ which is assured by Theorem~\ref{theorem6.3}(b).
We recall that this is a sufficiency assumption because the asymptotic estimates used
in Theorem~\ref{theorem6.3}(b)
on the potentials
are not necessarily sharp.

\begin{theorem}
\label{theorem6.6} Consider
the unperturbed Schr\"odinger operator associated with
\eqref{2.1} and \eqref{2.5}, where the potential $V$ satisfies \eqref{2.2} and belongs to $L^1_2(\mathbb R^+)$ and the boundary matrices $A$ and $B$ appearing in \eqref{2.5} satisfy
\eqref{2.6} and \eqref{2.7}.
Let $\varphi(k,x)$ be the corresponding regular solution satisfying
the initial conditions \eqref{2.10}, $J(k)$ be the Jost matrix defined in \eqref{2.11}, $d\rho$ be the spectral
measure given in \eqref{4.1}, 
and each of $k=i\kappa_j$ for $1\le j\le N$ correspond to the bound state with the energy $-\kappa_j^2,$ 
the Gel'fand--Levitan normalization matrix $C_j,$ 
the orthogonal projection
$Q_j$ onto $\text{\rm{Ker}}[J(i\kappa_j)],$ 
and
the Gel'fand--Levitan normalized bound-state solution $\Phi_j(x).$
Let us use a tilde to identify the quantities
associated with the perturbed Schr\"odinger operator
with the potential
$\tilde V(x)$ expressed as in \eqref{6.24}
and the boundary condition \eqref{2.5} where
the boundary matrices $A$ and $B$ are replaced with
$\tilde A$ and $\tilde B,$ respectively, given in \eqref{6.47}.
We use 
$\tilde\varphi(k,x)$ for the perturbed regular solution satisfying the initial
conditions \eqref{5.17}, the quantity $\tilde J(k)$ for the perturbed Jost matrix
defined as in \eqref{6.89}, the matrix
$\tilde Q_j$ for the orthogonal projection onto $\text{\rm{Ker}}[\tilde J(i\kappa_j)],$
the quantity
$\tilde C_j$ for the Gel'fand--Levitan normalization matrix defined
as in \eqref{3.33}, \eqref{3.34}, and \eqref{3.36} but
by using $\tilde Q_j$ instead of $Q_j$ and
by using $\tilde\varphi(i\kappa_j,x)$ instead of $\varphi(i\kappa_j,x)$ there for $1\le j\le N-1,$ and we define
the
perturbed spectral measure $d\tilde\rho$ 
as
\begin{equation}
\label{6.150}
d\tilde\rho=\begin{cases}
\ds\frac{\sqrt{\lambda}}{\pi}\,\left(\tilde J(k)^\dagger\,\tilde J(k)\right)^{-1}\,d\lambda,\qquad \lambda\ge 0,
\\
\noalign{\medskip}
\ds\sum_{j=1}^{N-1} \tilde C_j^2\,\delta(\lambda-\lambda_j)\,d\lambda,
\qquad \lambda<0.\end{cases}
\end{equation}
We then have the following:

\begin{enumerate}

\item[\text{\rm(a)}] Under the perturbation,
the rest of the projection matrices $Q_j$ are unchanged, i.e. we have
 \begin{equation}
\label{6.151}
\tilde Q_j=Q_j,\qquad 1\le j\le N-1.\end{equation}

\item[\text{\rm(b)}] Under the perturbation, the rest of the Gel'fand--Levitan normalization matrices are
 unchanged, i.e. we have
 \begin{equation}
\label{6.152}
\tilde C_j=C_j,\qquad 1\le j\le N-1.\end{equation}

\item[\text{\rm(c)}] The perturbed spectral measure $d\tilde\rho$ is related to the unperturbed spectral measure $d\rho$ as
\begin{equation}
\label{6.153}
d\tilde\rho=d\rho-C_N^2\,\delta(\lambda-\lambda_N)\,d\lambda.
\end{equation}
\end{enumerate}

\end{theorem}

\begin{proof}
The proof of (a) directly follows from \eqref{6.91}
and the fact that for
$1\le j\le N-1$  the matrices $Q_j$ and $\tilde Q_j$ are the orthogonal
projections onto $\text{\rm{Ker}}[J(i\kappa_j)]$
and $\text{\rm{Ker}}[\tilde J(i\kappa_j)],$
respectively, 
which indicate that any nonzero vector
in the kernel of $J(i\kappa_j)$ is also in the kernel of $\tilde J(i\kappa_j)$ and vice versa.
For the proof of (b), we proceed as follows. From \eqref{3.36} and \eqref{6.151} 
we see that \eqref{6.152} holds if and only if $\mathbf H_j$ is unchanged. With the help of
\eqref{3.34} we see that
$\mathbf H_j$ is unchanged if and only if $\mathbf G_j$ is unchanged.
Thus, the proof of (b) can be obtained by showing that 
the integral on the right-hand side of \eqref{3.33} is unchanged for $1\le j\le N-1.$   
Let us define the $n\times n$ matrices $\mathbf S_j(x)$ and $\tilde{\mathbf S}_j(x)$ as
\begin{equation}
\label{6.154}
\mathbf S_j(x):=\varphi(i\kappa_j,x)\,Q_j,\quad
\tilde{\mathbf S}_j(x):=\tilde\varphi(i\kappa_j,x)\,Q_j,
\qquad 1\le j\le N-1.
\end{equation}
From the second equality of \eqref{3.39}, by using the fact that
$\mathbf H_j^{-1/2}$ and $Q_j$ commute, we obtain
\begin{equation}
\label{6.155}
\Phi_j(x):=\varphi(i\kappa_j,x)\,Q_j\,\mathbf H_j^{-1/2},
\qquad 1\le j\le N-1.
\end{equation}
From \eqref{6.154} and \eqref{6.155} it follows that
\begin{equation*}
\mathbf S_j(x)=\Phi_j(x)\,\mathbf H_j^{1/2},
\qquad 1\le j\le N-1.
\end{equation*}
Evaluating \eqref{6.46} at $k=i\kappa_j$ and postmultiplying 
the resulting equation with $Q_j,$ we obtain
\begin{equation}
\label{6.157}
\tilde{\mathbf S}_j(x)=\mathbf S_j(x)
+\ds\frac{1}{\kappa_N^2-\kappa_j^2}
\,\Phi_N(x) \,W_N(x)^+\,q_9(x)\,\mathbf H_j^{1/2},\qquad 1\le j\le N-1,
\end{equation}
where 
we have defined
\begin{equation*}
q_9(x):=\Phi'_N(x)^\dagger\,\Phi_j(x)-\Phi_N(x)^\dagger\,\Phi_j'(x).
\end{equation*}
By taking the matrix adjoint of \eqref{6.157}, we get
\begin{equation}
\label{6.159}
\tilde{\mathbf S}_j(x)^\dagger=\mathbf S_j(x)^\dagger
+\ds\frac{1}{\kappa_N^2-\kappa_j^2}
\, \mathbf H_j^{1/2}\,q_9(x)^\dagger\,W_N(x)^+\,\Phi_N(x)^\dagger,\qquad 1\le j\le N-1,
\end{equation}
where we have used the fact that
$\mathbf H_j^{1/2}$ and $W_N(x)^+$ are both selfadjoint.
The selfadjointness of $W_N(x)$ is seen from \eqref{6.3}.  
From \eqref{3.4} we know that the operations for the adjoint and the Moore--Penrose inverse 
for any square matrix commute. Consequently, from \eqref{6.3} we
conclude that $W_N(x)^+$
is also selfadjoint.
From \eqref{6.157} and \eqref{6.159} we get
\begin{equation}
\label{6.160}
\mathbf H_j^{-1/2}\left[\tilde{\mathbf S}_j(x)^\dagger\,
\tilde{\mathbf S}_j(x)-\mathbf S_j(x)^\dagger\,\mathbf S_j(x)\right]\mathbf H_j^{-1/2}=\ds\frac{1}{\kappa_N^2-\kappa_j^2}
\left[q_{10}(x)+q_{11}(x)+q_{12}(x)\right],
\end{equation}
where we have defined
\begin{equation}
\label{6.161}
q_{10}(x):=\Phi_j(x)^\dagger\,\Phi_N(x) \,W_N(x)^+\,
q_9(x),
\end{equation}
\begin{equation}
\label{6.162}
q_{11}(x):=
q_9(x)^\dagger\,W_N(x)^+\,\Phi_N(x)^\dagger\,\Phi_j(x),
\end{equation}
\begin{equation}
\label{6.163}
q_{12}(x):=q_9(x)^\dagger\,
\ds\frac{W_N(x)^+\,\Phi_N(x)^\dagger\,\Phi_N(x)\, W_N(x)^+}{\kappa_N^2-\kappa_j^2}\,
q_9(x).
\end{equation}
We have
\begin{equation}
\label{6.164}
\Phi_j(x)^\dagger\,\Phi_N(x) =
\ds\frac{1}{\kappa_N^2-\kappa_j^2}\,\ds\frac{d}{dx}\left[q_9(x)^\dagger
\right],
\end{equation}
\begin{equation}
\label{6.165}
W_N(x)^+\,\Phi_N(x)^\dagger\,\Phi_N(x)\, W_N(x)^+=\ds\frac{d W_N(x)^+}{dx},
\end{equation}
\begin{equation}
\label{6.166}
\Phi_N(x)^\dagger\,\Phi_j(x) =
\ds\frac{1}{\kappa_N^2-\kappa_j^2}\,\ds\frac{d}{dx}\left[q_9(x)
\right].
\end{equation}
Note that \eqref{6.165} follows from \eqref{6.25}. We can verify \eqref{6.164} and \eqref{6.166} 
with the help of \eqref{6.37} and the corresponding Schr\"odinger equation
\begin{equation*}
-\Phi_j''(x)+V(x)=-\kappa_j^2\,\Phi_j(x),\qquad 1\le j\le N-1.
\end{equation*}
Using \eqref{6.164}, \eqref{6.165}, and \eqref{6.166} in
\eqref{6.161}, \eqref{6.162}, and \eqref{6.163}, respectively, we observe that the right-hand side of
\eqref{6.160} is proportional to the derivative of the product of three matrix-valued functions.
Thus, we can write \eqref{6.160} as
\begin{equation}
\label{6.168}
\mathbf H_j^{-1/2}\left[\tilde{\mathbf S}_j(x)^\dagger\,
\tilde{\mathbf S}_j(x)-\mathbf S_j(x)^\dagger\,\mathbf S_j(x)\right]\mathbf H_j^{1/2}=\ds\frac{1}{\left(\kappa_N^2-\kappa_j^2\right)^2}\,\ds\frac{d\Upsilon(x)}{dx},
\end{equation}
where we have defined
\begin{equation}
\label{6.169}
\Upsilon(x):=q_9(x)^\dagger\,
W_N(x)^+\,q_9(x).
\end{equation}
Because $\mathbf H_j^{-1/2}$ is an invertible matrix and $\kappa_j\ne \kappa_N$
for $1\le j\le N-1,$ from \eqref{6.168} we conclude that
\begin{equation}
\label{6.170}
\ds\int_0^\infty dx\,\left[\tilde{\mathbf S}_j(x)^\dagger\,
\tilde{\mathbf S}_j(x)-\mathbf S_j(x)^\dagger\,\mathbf S_j(x)\right]=0,\end{equation}
if and only if we have
\begin{equation}
\label{6.171}\Upsilon(+\infty)=\Upsilon(0).
\end{equation}
With the help of \eqref{6.50},  \eqref{6.51}, their adjoints,  and their analogs for $\Phi_j(x),$ we get
\begin{equation}
\label{6.172}
q_9(x)=
O\left(e^{-(\kappa_N+\kappa_j)x}\right),\qquad x\to+\infty,
\end{equation}
\begin{equation}
\label{6.173}
q_9(x)^\dagger=
O\left(e^{-(\kappa_j+\kappa_N)x}\right),\qquad x\to+\infty.
\end{equation}
On the other hand, from \eqref{6.59} we have
\begin{equation}
\label{6.174}
W_N(x)^+=O\left(e^{2\kappa_N x}\right),\qquad x\to+\infty.
\end{equation}
Using \eqref{6.172}, \eqref{6.173}, \eqref{6.174}, we obtain
\begin{equation*}
\Upsilon(x)=O\left(e^{-2\kappa_j x}\right),\qquad x\to+\infty,
\end{equation*}
which yields
\begin{equation}
\label{6.176}
\Upsilon(+\infty)=0.
\end{equation}
Note that we have the analogs of
\eqref{6.38} given by
\begin{equation}
\label{6.177}
\Phi_j(0)=A C_j,\quad \Phi_j'(0)=B C_j,\qquad 1\le j\le N-1.
\end{equation}
Using \eqref{6.38}, \eqref{6.77}, \eqref{6.177} on the right-hand side of \eqref{6.169}, we obtain
\begin{equation}
\label{6.178}
\Upsilon(0)=C_j\left(A^\dagger B-B^\dagger A\right) C_N\,Q_N\,C_N \left(B^\dagger A-A^\dagger B\right) C_j,\qquad
1\le j\le N-1,
\end{equation}
where we have used the fact that $C_j$ and $C_N$ are selfadjoint matrices.
Using \eqref{2.6} on the right-hand side of \eqref{6.178}, we see that
\begin{equation}
\label{6.179}
\Upsilon(0)=0,
\end{equation}
and hence \eqref{6.176} and \eqref{6.179} yield \eqref{6.171}, which in turn implies that \eqref{6.170} holds.
Comparing \eqref{6.170} with \eqref{3.33}, we see that the matrix $\mathbf G_j$ is unchanged and we have
\begin{equation*}
\tilde{\mathbf G}_j=\mathbf G_j,\qquad 1\le j\le N-1.
\end{equation*}
Thus, \eqref{6.152} holds and the proof of (b) is complete. Finally, we establish \eqref{6.153}
by using \eqref{5.1}, \eqref{6.92},  (d) of Theorem~\ref{theorem6.4}, \eqref{6.150}, and \eqref{6.152}.
Hence, the proof of (c) is complete.
\end{proof}

\section{The  transformation to decrease the multiplicity of a bound state}
\label{section7}

In Section~\ref{section6} we have presented a method to completely remove an eigenvalue from the discrete spectrum
of the matrix-valued Schr\"odinger operator with the general selfadjoint boundary condition, without changing
the continuous spectrum. We have provided
the transformations for all relevant quantities resulting from the removal. Without loss of generality, we have removed
the bound state at $k=i\kappa_N$ with the multiplicity $m_N$ and the Gel'fand--Levitan normalization matrix $C_N$
so that after the removal we have left with the remaining bound states at $k=i\kappa_j$ for $1\le j\le N-1.$ We have seen that
the remaining bound-state energies $-\kappa_j^2,$ multiplicities $m_j,$ and Gel'fand--Levitan normalization matrices $C_j$
are not affected by the removal. 
In the scalar case, i.e. when $n=1,$ the multiplicity $m_N$ is equal to $1,$ and hence the reduction of the multiplicity
$m_N$ is always equivalent to the complete removal of the bound state at $k=i\kappa_N.$
On the other hand, when $n\ge 2$
the multiplicity $m_N$ is a fixed positive integer satisfying $1\le m_N\le n$ and hence the reduction in the 
multiplicity $m_N$ is not necessarily the same as the complete removal of the bound state.

This section is complementary to Section~\ref{section6}. Instead of completely removing the bound state
at $k=i\kappa_N,$ we reduce its multiplicity from $m_N$ to $\tilde m_N$ so that
we have $1\le \tilde m_N<m_N.$ Thus, we assume that the number $N$ of bound states is at least $1$ and
that the  multiplicity $m_N$ is at least $2.$
As in the previous sections, we use a tilde to identify the perturbed
quantities obtained after the multiplicity of the bound state is reduced. We assume that our
unperturbed potential $V(x)$
satisfies \eqref{2.2} and belongs
to $L^1_1(\mathbb R^+).$ 
Our unperturbed regular solution is
$\varphi(k,x),$ unperturbed Jost solution is $f(k,x),$ 
unperturbed Jost matrix is $J(k),$ unperturbed scattering matrix is $S(k),$ unperturbed 
matrix $Q_N$ corresponds to the
orthogonal projection onto the kernel of
$J(i\kappa_N),$ and 
unperturbed
boundary matrices are given by $A$ and $B.$
Our perturbed potential is $\tilde V(x),$ perturbed regular solution is
$\tilde\varphi(k,x),$ perturbed Jost solution is $\tilde f(k,x),$ 
perturbed Jost matrix is $\tilde J(k),$ perturbed scattering matrix is $\tilde S(k),$ 
and perturbed
boundary matrices are given by $\tilde A$ and $\tilde B.$

We assume that the multiplicity of the bound state at $k=i\kappa_N$ is reduced by $m_{N{\text{\rm{r}}}},$ and hence we define the positive
integer representing the reduction in the multiplicity of the bound state at $k=i\kappa_N$ as
\begin{equation}
\label{7.1}
m_{N{\text{\rm{r}}}}:=m_N-\tilde m_N.
\end{equation}
The subscript $N{\text{\rm{r}}}$ indicates that we refer to the $N$th bound state at $k=i\kappa_N$ and that we reduce
its multiplicity.
From \eqref{7.1} we see that $m_{N{\text{\rm{r}}}}$ satisfies the inequality
$1\le m_{N{\text{\rm{r}}}}\le n-1.$
The reduction of the multiplicity of the bound state at $k=i\kappa_N$
by $m_{N{\text{\rm{r}}}}$ is achieved 
by introducing the matrix $Q_{N{\text{\rm{r}}}}$ in such a way that
$Q_{N{\text{\rm{r}}}}$ is an orthogonal projection onto
a proper subspace of 
$Q_N\, \mathbb C^n,$ it has its rank equal to $m_{N{\text{\rm{r}}}},$
and it satisfies $Q_{N{\text{\rm{r}}}}\le Q_N.$
This last matrix inequality is equivalent to 
$Q_N-Q_{N{\text{\rm{r}}}}\ge 0,$ which indicates that
the matrix $Q_N-Q_{N{\text{\rm{r}}}}$ is nonnegative. In other words,
the eigenvalues of the matrix
 $Q_N-Q_{N{\text{\rm{r}}}}$ are all real and nonnegative.

The transformations for the relevant quantities are obtained as follows. We omit the proofs in this section
because they are similar to those presented  in Sections~\ref{section3} and \ref{section6}.
We note that we have
\begin{equation*}
Q_{N{\text{\rm{r}}}}\, Q_N=
 Q_N\, Q_{N{\text{\rm{r}}}}= Q_{N{\text{\rm{r}}}}.
 \end{equation*}
Consequently, the columns of the matrix $\varphi(i\kappa_N,x)\,Q_{N{\text{\rm{r}}}}$ are square integrable because
the columns of the matrix $\varphi(i\kappa_N,x)\,Q_N$ are square integrable.
Analogous to \eqref{3.33} and \eqref{3.34}, we let
\begin{equation}\label{7.3}
 \mathbf G_{N{\text{\rm{r}}}}:= \int_0^\infty dx\,Q_{N{\text{\rm{r}}}}\, \varphi(i\kappa_N,x)^\dagger \,\varphi(i\kappa_N,x) \,Q_{N{\text{\rm{r}}}},
 \end{equation}
 \begin{equation}\label{7.4}
 \mathbf H_{N{\text{\rm{r}}}}:= I-Q_{N{\text{\rm{r}}}}+\mathbf  G_{N{\text{\rm{r}}}}.
 \end{equation}
As in Section~\ref{section3}, we prove that the matrix
$\mathbf H_{N{\text{\rm{r}}}}$ is positive and hence invertible. Thus, the positive matrix
$\mathbf H_{N{\text{\rm{r}}}}^{1/2}$
and its inverse
$\mathbf H_{N{\text{\rm{r}}}}^{-1/2}$
are uniquely defined.
Furthermore, we have
\begin{equation*}
\mathbf H_{N{\text{\rm{r}}}}\,Q_{N{\text{\rm{r}}}}= Q_{N{\text{\rm{r}}}}\, \mathbf H_{N{\text{\rm{r}}}}.
 \end{equation*}
Analogous to \eqref{3.36}, with the help of \eqref{7.3} and \eqref{7.4} we introduce the $n\times n$ matrix $C_{N{\text{\rm{r}}}}$ as
\begin{equation*}
C_{N{\text{\rm{r}}}}:= \mathbf H_{N{\text{\rm{r}}}}^{-1/2} \,Q_{N{\text{\rm{r}}}}.
\end{equation*}
In a manner similar to \eqref{3.28} we introduce
the $n\times n$ matrix solution $\Phi_{N{\text{\rm{r}}}}(x)$ to the Schr\"odinger equation \eqref{3.19} as
\begin{equation}\label{7.7} 
\Phi_{N{\text{\rm{r}}}}(x):= \varphi(i\kappa_N, x) \,C_{N{\text{\rm{r}}}}.
\end{equation}
Analogous to \eqref{3.29} we obtain the normalization property
\begin{equation*}
\int_0^\infty dx\, \Phi_{N{\text{\rm{r}}}}(x)^\dagger\, \Phi_{N{\text{\rm{r}}}}(x)= Q_{N{\text{\rm{r}}}}.
\end{equation*}
Let $\{w^{(l)}_N\}_{l=1}^{m_{N{\text{\rm{r}}}}}$ be an orthonormal basis for $Q_{N{\text{\rm{r}}}}\, \mathbb C^n.$ Analogous to the
second equality in \eqref{3.14}, we have
\begin{equation}\label{7.9}
Q_{N{\text{\rm{r}}}}= \ds\sum_{l=1}^{m_{N{\text{\rm{r}}}}} w^{(l)}_N (w^{(l)}_N)^\dagger.
\end{equation} 
Each column of $\varphi(i\kappa_N,x) \,Q_{N{\text{\rm{r}}}}$ is square integrable and
satisfies the boundary condition \eqref{2.1}. Consequently, from Theorem~3.11.1 of \cite{AW2021} it follows that
for each column vector $w^{(l)}_N$ there exists a unique column vector $\beta^{(l)}_N$
in $\text{\rm{Ker}}[J(i\kappa_N)^\dagger]$ such that
\begin{equation}\label{7.10}
 \varphi(i\kappa_N,x)\, w^{(l)}_N=f(i\kappa_N,x)\,\beta^{(l)}_N, \qquad 1\le l\le m_{N{\text{\rm{r}}}}. 
\end{equation}
From \eqref{7.9} we have
\begin{equation*}
Q_{N{\text{\rm{r}}}}\,w^{(l)}_N=w^{(l)}_N,
\end{equation*}
and hence \eqref{7.10} is equivalent to
\begin{equation}\label{7.12}
\varphi(i\kappa_N,x)\,  Q_{N{\text{\rm{r}}}} \,w^{(l)}_N=f(i\kappa_N,x)\,\beta^{(l)}_N,\qquad 1\le l\le m_{N{\text{\rm{r}}}}.  
\end{equation}

Let us use $P_{N{\text{\rm{r}}}}$ to denote the orthogonal projection onto the subspace of $\text{\rm{Ker}}[J(i\kappa_N)^\dagger]$ generated by
 the orthonormal set $\{\beta^{(l)}_N\}_{l=1}^{m_{N{\text{\rm{r}}}}}.$ From Proposition~3.11.1
 of \cite{AW2021} it follows that the map $w^{(l)}_N \mapsto \beta^{(l)}_N$ is one-to-one.
 Consequently, the rank of  $P_{N{\text{\rm{r}}}}$ is equal to $m_{N{\text{\rm{r}}}}.$  Consider
any column vector
$w$ in $Q_{N{\text{\rm{r}}}}\, \mathbb C^n,$ which can be expressed as
 \begin{equation}
 \label{7.13}
 w= \sum_{l=1}^{m_{N{\text{\rm{r}}}}} a_l \,w^{(l)}_N,
\end{equation}
 for some appropriate coefficients $a_l.$
 By letting
 \begin{equation}\label{7.14}
v:=\sum_{l=1}^{m_{N{\text{\rm{r}}}}} a_l  \,\beta^{(l)}_N ,
\end{equation}
 we see that $v\in P_{N{\text{\rm{r}}}}\, \mathbb C^n$ and we also have $P_{N{\text{\rm{r}}}}\,v=v.$ Consequently,
 using \eqref{7.10}, \eqref{7.12}, \eqref{7.13}, and
 \eqref{7.14}, we get
  \begin{equation}\label{7.15}
\varphi(i\kappa_N,x)\, Q_{N{\text{\rm{r}}}}\, w
 = f(i\kappa_N, x)\, P_{N{\text{\rm{r}}}}\, v.
 \end{equation}
 As already indicated, each column of the matrix
$ \varphi(i\kappa_N,x)\, Q_{N{\text{\rm{r}}}}$ is a square-integrable solution to \eqref{3.19} and satisfies the boundary
 condition \eqref{2.5}.
 From \eqref{7.15} it follows that
 each column of the matrix $f(i\kappa_N, x)\, P_{N{\text{\rm{r}}}}$
 is a square-integrable solution to \eqref{3.19} and satisfies the boundary
 condition \eqref{2.5}. Hence, we can use
 the matrix $f(i\kappa_N, x)\, P_{N{\text{\rm{r}}}}$ to construct an $n\times n$ matrix
 $M_{N{\text{\rm{r}}}}$ as in \eqref{3.27}. For this, we use the analogs of \eqref{3.24} and \eqref{3.25} by letting
\begin{equation*}
\mathbf A_{N{\text{\rm{r}}}}:= \int_0^\infty dx\, P_{N{\text{\rm{r}}}}\, f(i\kappa_N,x)^\dagger\, f(i\kappa_N, x) \,P_{N{\text{\rm{r}}}},
\end{equation*}
\begin{equation*}
\mathbf B_{N{\text{\rm{r}}}}:= I-P_{N{\text{\rm{r}}}}+\mathbf A_{N{\text{\rm{r}}}}.
\end{equation*}
The matrix $\mathbf B_{N{\text{\rm{r}}}}$ is positive and hence invertible, and it commutes with
$P_{N{\text{\rm{r}}}}.$ Withe the help of the analogs of \eqref{3.26} and \eqref{3.27}, we define
the $n\times n$ matrix $M_{N{\text{\rm{r}}}}$ as
\begin{equation*}
M_{N{\text{\rm{r}}}}:=\mathbf B_{N{\text{\rm{r}}}}^{-1/2} \,P_{N{\text{\rm{r}}}},
\end{equation*}
where $\mathbf B_{N{\text{\rm{r}}}}^{-1/2}$ is the positive matrix corresponding to the
inverse of the positive matrix $\mathbf B_{N{\text{\rm{r}}}}^{1/2}.$ Analogous to \eqref{3.21}, we define
the $n\times n$ matrix solution $\Psi_{N{\text{\rm{r}}}}$ to \eqref{3.19} as
\begin{equation}\label{7.19}
\Psi_{N{\text{\rm{r}}}}:= f(i\kappa_N, x) \,M_{N{\text{\rm{r}}}}.
\end{equation}
In a manner similar to \eqref{4.1} and \eqref{4.2}, using \eqref{7.7} and \eqref{7.19} we introduce the
dependency matrix $D_{N{\text{\rm{r}}}}$ so that
\begin{equation*}
\Phi_{N{\text{\rm{r}}}} = \Psi_{N{\text{\rm{r}}}}\, D_{N{\text{\rm{r}}}},
\end{equation*}
\begin{equation*}
D_{N{\text{\rm{r}}}}=   P_{N{\text{\rm{r}}}}\,   D_{N{\text{\rm{r}}}} .
\end{equation*}
Analogous to \eqref{4.14}, we prove that the dependency matrix $D_{N{\text{\rm{r}}}}$ satisfies
\begin{equation*}
D_{N{\text{\rm{r}}}}= P_{N{\text{\rm{r}}}}\, \mathbf B_{N{\text{\rm{r}}}}^{1/2}\ds \sum_{l=1}^{m_{N{\text{\rm{r}}}}} \beta^{(l)}_N  \,[w^{(l)}_N]^\dagger \,
\mathbf H_{N{\text{\rm{r}}}}^{-1/2} \,Q_{N{\text{\rm{r}}}}.
\end{equation*}

We recall that we assume that the unperturbed potential $V$ satisfies \eqref{2.2} and belongs to $L^1_1(\mathbb R^+).$
As in Section~\ref{section6} we prove the following:

\begin{enumerate}

\item[\text{\rm(1)}]  We establish the appropriate analogs of
(a)--(h) of Theorem~\ref{theorem4.1} for the quantities $D_{N{\text{\rm{r}}}},$
$Q_{N{\text{\rm{r}}}},$ $P_{N{\text{\rm{r}}}},$ $\Psi_{N{\text{\rm{r}}}},$ and $\Phi_{N{\text{\rm{r}}}},$
where $m_{N{\text{\rm{r}}}}$ corresponds to the common rank of the orthogonal projections
$Q_{N{\text{\rm{r}}}}$ and $P_{N{\text{\rm{r}}}}.$

\item[\text{\rm(2)}]  We show that the analogs of  \eqref{4.42} and \eqref{4.50} hold.
This is done by proving that 
\begin{equation*}
D_{N{\text{\rm{r}}}}=M_{N{\text{\rm{r}}}}^+\ f(i\kappa_N,x)^{-1} \,  \Phi_{N{\text{\rm{r}}}}(x),
\end{equation*}
where the right-hand side is evaluated at any $x$-value at which the matrix  $f(i\kappa_N,x)$ is invertible.
We also show that
\begin{equation*}
D_{N{\text{\rm{r}}}}= M_{N{\text{\rm{r}}}}^+\, f' (i\kappa_N,x)^{-1}  \,\Phi'_{N{\text{\rm{r}}}}(x),
\end{equation*} 
where the right-hand side is evaluated at any $x$-value at which the matrix  $f'(i\kappa_N,x)$ is invertible.

\item[\text{\rm(3)}] 
We show that the analog of Theorem~\ref{theorem4.3} holds. In other words, we establish that

\begin{enumerate}
\item[\text{\rm(a)}]
 The matrix $D_{N{\text{\rm{r}}}}$ is a partial isometry
with the initial subspace $Q_{N{\text{\rm{r}}}}\,\mathbb C^n$ and the
final subspace  $P_{N{\text{\rm{r}}}}\, \mathbb C^n.$

\item[\text{\rm(b)}] 
 The adjoint matrix $D_{N{\text{\rm{r}}}}^\dagger$ is a partial isometry
with the initial subspace $P_{N{\text{\rm{r}}}}\, \mathbb C^n$ and the
final subspace  $Q_{N{\text{\rm{r}}}}\, \mathbb C^n.$ 

\end{enumerate}

\item[\text{\rm(4)}]
We solve the  Gel'fand--Levitan system of integral equations \eqref{5.49} with the kernel $G(x,y)$ defined as
\begin{equation*}
G(x,y):= - \varphi(k,x)\, C_{N{\text{\rm{r}}}}^2\, \varphi(k,y)^\dagger,
\end{equation*}
which is the analog of \eqref{6.1}.
Then, by proceeding as in the proof of Theorem~\ref{theorem6.1}, we show that the solution $\mathcal A(x,y)$ to \eqref{5.49} is given by
\begin{equation}\label{7.26}
\mathcal A(x,y)= \Phi_{N{\text{\rm{r}}}}(x)^\dagger \,W_{N{\text{\rm{r}}}}(x)^+ \,\Phi_{N{\text{\rm{r}}}}(y), \qquad  0 \le y <x,
\end{equation}
which is the analog of \eqref{6.2}. Here, we have defined 
\begin{equation*}
W_{N{\text{\rm{r}}}}(x):=\int_x^\infty dy\,\Phi_{N{\text{\rm{r}}}}(y)^\dagger\, \Phi_{N{\text{\rm{r}}}}(y),
\end{equation*}
which is the analog of \eqref{6.3}.
We recall that $W_{N{\text{\rm{r}}}}(x)^+$ denotes the Moore--Penrose inverse of  $W_{N{\text{\rm{r}}}}(x).$

\item[\text{\rm(5)}]
We define the perturbed potential $\tilde V(x)$ as
\begin{equation}
\label{7.28}
\tilde V(x):=V(x)+2\,\ds\frac{d}{dx}\left[
\Phi_{N{\text{\rm{r}}}}(x) \,W_{N{\text{\rm{r}}}}(x)^+\,\Phi_{N{\text{\rm{r}}}}(x)^\dagger
\right],
\end{equation} 
which is the analog of \eqref{6.24}.
We then show that
$\mathcal A(x,y)$ and $\tilde V(x)$ satisfy (a)--(d) of Theorem~\ref{theorem6.2}, and this is done
by proceeding as in the proof of Theorem~\ref{theorem6.2}.

\item[\text{\rm(6)}]
By proceeding as in the proof of Theorem~\ref{theorem6.3}, we establish the following:
\begin{enumerate}

\item[\text{\rm(a)}] We construct the quantity $\tilde\varphi(k,x)$ as in \eqref{5.18} but by using \eqref{7.26}, and we obtain
\begin{equation}
\label{7.29}
\tilde\varphi(k,x)=\varphi(k,x)+\int_0^x dy\,\mathcal A(x,y)\,\varphi(k,y).
\end{equation}
We show that $\tilde\varphi(k,x)$ given in \eqref{7.29}  is a solution to the perturbed Schr\"odinger equation \eqref{2.1}
 with the potential $\tilde V(x)$ in \eqref{7.28}. Furthermore, we show that $\tilde\varphi(k,x)$ given in \eqref{7.29} 
 satisfies the analog of \eqref{6.41}, and hence we have
 \begin{equation*}
\tilde\varphi(k,x)=\varphi(k,x)+ \Phi_{N{\text{\rm{r}}}}(x) \,W_{N{\text{\rm{r}}}}(x)^+ \int_0^x dy\,\Phi_{N{\text{\rm{r}}}}(y)^\dagger\,\varphi(k,y).
\end{equation*}

\item[\text{\rm(b)}] The perturbed potential $\tilde V(x)$ appearing in \eqref{7.28} 
satisfies \eqref{2.2}, and we have the asymptotics for the potential increment $\tilde V(x)-V(x)$ given by
\begin{equation*}
\tilde V(x)-V(x)= O\left( q_6(x) \right), \qquad x \to +\infty,
\end{equation*}
which is the analog of \eqref{6.42}
and we recall that $q_6(x)$ is the quantity defined in
\eqref{6.43}.
We further show that, if the unperturbed potential
$V$ belongs  to $L^1_{1+\epsilon}(\mathbb R^+)$ for some fixed $\epsilon\ge 0,$ then the perturbed
potential $\tilde V$ belongs to
 $L^1_{\epsilon}(\mathbb R^+).$

\item[\text{\rm(c)}] If the unperturbed potential $V(x)$ is further restricted to satisfy
\begin{equation}\label{7.32}
|V(x)| \le c\,e^{-\alpha x}, \qquad x \ge x_0,
\end{equation}
for some positive constants $\alpha$ and $x_0,$ then the potential increment $\tilde V(x)-V(x)$ satisfies
\begin{equation}\label{7.33}
|\tilde V(x)-V(x)| \le c\,e^{-\alpha x}, \qquad x \ge x_0.
\end{equation}
Here, $c$ denotes a generic constant not necessarily taking the same value in different appearances.
We remark that \eqref{7.32} and \eqref{7.33} are the analogs of \eqref{6.44} and \eqref{6.45}, respectively.

\item[\text{\rm(d)}]
If the support of $V$ is contained in the interval $[0,x_0],$ then the support of $\tilde V$ is also contained in $[0,x_0].$ 

\item[\text{\rm(e)}] For $ k \ne i \kappa_N,$  the perturbed quantity $\tilde\varphi(k,x)$ can be expressed as
\begin{equation*}
\begin{split}
\tilde\varphi(k,x)=&\varphi(k,x)\\
&+\ds\frac{1}{k^2+\kappa_N^2}
\,\Phi_{N{\text{\rm{r}}}}(x) \,W_{N{\text{\rm{r}}}}(x)^+\left[\Phi'_{N{\text{\rm{r}}}}(x)^\dagger\,\varphi(k,x)-\Phi_{N{\text{\rm{r}}}}(x)^\dagger
\,\varphi'(k,x)\right],
\end{split}
\end{equation*}
which is the analog of \eqref{6.46}.

\item[\text{\rm(f)}] The perturbed quantity $\tilde\varphi(k,x)$ satisfies the initial conditions \eqref{5.17}
with the matrices $\tilde A$ and $\tilde B$
expressed in terms of the unperturbed
boundary matrices $A$ and $B$ and
the matrix $C_{N{\text{\rm{r}}}}$ as
\begin{equation}\label{7.35}
\tilde A=A, \quad \tilde B=B+A\,C_{N{\text{\rm{r}}}}^2 A^\dagger A,
\end{equation}
which is the analog of \eqref{6.47}.

\item[\text{\rm(g)}] 
The matrices $\tilde A$ and $\tilde B$ appearing in \eqref{7.35} satisfy \eqref{2.6} and \eqref{2.7}. Hence, 
as a consequence of (a) and (f), it follows that 
the quantity $\tilde\varphi(k,x)$ is the regular solution 
to the  Schr\"odinger equation with the potential $\tilde V(x)$ in \eqref{7.28} 
and with the selfadjoint
boundary condition \eqref{2.5} with $A$ and $B$ there replaced with
$\tilde A$ and $\tilde B,$ respectively.
In other words, $\tilde \varphi(k,x)$ satisfies \eqref{5.16} and \eqref{5.17}.

\end{enumerate}

\item[\text{\rm(7)}]  If we further assume  that $V$ belongs to $L^1_{2}(\mathbb R^+),$  then we proceed
as in the proof of Theorem~\ref{theorem6.4} and establish the following:

\begin{enumerate}

\item[\text{\rm(a)}] The unperturbed Jost matrix $J(k)$ is transformed into the perturbed Jost matrix $\tilde J(k)$ as
 \begin{equation*}
\tilde J(k)=\left[I+\ds\frac{2i\kappa_N}{k-i\kappa_N}\,P_{N{\text{\rm{r}}}}\right] J(k),\qquad
k\in\overline{\mathbb C^+},\end{equation*}
which is the analog of \eqref{6.91}.

\item[\text{\rm(b)}] The matrix product $J(k)^\dagger \,J(k)$ for $k\in\mathbb R$ does not change under
the perturbation,
i.e. we have
 \begin{equation*}
\tilde J(k)^\dagger\,\tilde J(k)=J(k)^\dagger\,J(k),\qquad k\in\mathbb R,\end{equation*}
which is the analog of \eqref{6.92}.
Consequently, the continuous part of the spectral measure $d\rho$ does not change
under the perturbation.

\item[\text{\rm(c)}] Under the perturbation, the determinant of the Jost matrix is transformed as
 \begin{equation}
\label{7.38}
\det[\tilde J(k)]=\left(\ds\frac{k+i\kappa_N}{k-i\kappa_N}\right)^{m_{N{\text{\rm{r}}}}} \det[J(k)],\qquad
k\in\overline{\mathbb C^+},\end{equation}
which is the analog of \eqref{6.93}.
We recall
that $m_{N{\text{\rm{r}}}}$ corresponds to the rank of the orthogonal projection  $Q_{N{\text{\rm{r}}}}.$

\item[\text{\rm(d)}] Under the perturbation,
the bound state with the energy $-\kappa_N^2$ remains but its multiplicity is reduced to $m_N-m_{N{\text{\rm{r}}}}.$ 
No new bound states are added, and the  bound states with the energies $-\kappa_j^2$ and multiplicities $m_j$ for $1\le j\le N-1$
are unchanged. We confirm that the multiplicity of the bound state with the energy $-\kappa_N^2$ is reduced to $m_N-m_{N{\text{\rm{r}}}}$ 
by the following observation. 
The unperturbed Schr\"odinger equation has a bound state with the energy $-\kappa_N^2$ and multiplicity $m_N.$ 
From Theorem~3.11.6 of \cite{AW2021} we know that
$\det[J(k)]$ has a zero of order $m_N$ at $k=i\kappa_N.$ Moreover, from \eqref{7.38} we see that $\det[\tilde J(k)]$
has a zero of order $m_N-m_{N{\text{\rm{r}}}}$ at $k= i\kappa_N.$
\item[\text{\rm(e)}] Under the perturbation, the scattering matrix $S(k)$ undergoes the transformation
 \begin{equation*}
\tilde S(k)=\left[I-\ds\frac{2i\kappa_N}{k+i\kappa_N}\,P_{N{\text{\rm{r}}}}\right] S(k)\left[I-\ds\frac{2i\kappa_N}{k+i\kappa_N}\,P_{N{\text{\rm{r}}}}\right],
\qquad k\in\mathbb R,\end{equation*}
which is the analog of \eqref{6.94}.

\item[\text{\rm(f)}] Under the perturbation, the determinant of the scattering matrix  is transformed as 
 \begin{equation*}
\det[\tilde S(k)]=\left(\ds\frac{k-i\kappa_N}{k+i\kappa_N}\right)^{2\, m_{N{\text{\rm{r}}}}} \det[S(k)],
\qquad k\in\mathbb R,\end{equation*}
which is the analog of \eqref{6.95}.

\item[\text{\rm(g)}] Under the perturbation, for
$k\in\overline{\mathbb C^+}$
the unperturbed Jost solution $f(k,x)$ is transformed into the perturbed Jost solution $\tilde f(k,x)$ as
\begin{equation}\label{7.41}
\tilde f(k,x)=\left[f(k,x)+\ds\frac{1}{k^2+\kappa_N^2}
\,\Phi_{N{\text{\rm{r}}}}(x) \,W_{N{\text{\rm{r}}}}(x)^+ \,q_{13}(x) \right]  \left[ I+ \ds\frac{2i\kappa_N}{k-i\kappa_N} \,P_{N{\text{\rm{r}}}} \right],
\end{equation}
where we have defined
\begin{equation*}
q_{13}(x):=\Phi'_{N{\text{\rm{r}}}}(x)^\dagger\,f(k,x)-\Phi_{N{\text{\rm{r}}}}(x)^\dagger
\,f'(k,x),
\end{equation*}
which are the analogs of \eqref{6.96} and
\eqref{6.97}, respectively.
We can express \eqref{7.41} equivalently as
\begin{equation*}
\tilde f(k,x)=\left[ f(k,x)-
\Phi_{N{\text{\rm{r}}}}(x) \,W_{N{\text{\rm{r}}}}(x)^+\int_x^\infty dy\,
\Phi_{N{\text{\rm{r}}}}(y)^\dagger \,f(k,y)  \right] \left[ I+ \frac{2i\kappa_N}{k-i\kappa_N} P_{N{\text{\rm{r}}}} \right],
\end{equation*}
which is the analog of \eqref{6.98}.

\end{enumerate}

\item[\text{\rm(8)}] 
Further assuming that the unperturbed potential $V$ belongs to $L^1_{2}(\mathbb R^+),$ we proceed as in the proof of
Theorem~\ref{theorem6.6}, and we establish the following:

\begin{enumerate}

\item[\text{\rm(a)}] Under the perturbation,
the projection matrices $Q_j$ for $1\le j\le N-1$ remain unchanged, i.e. we have
 \begin{equation*}
\tilde Q_j=Q_j,\qquad 1\le j\le N-1.\end{equation*}

\item[\text{\rm(b)}] Under the perturbation, the Gel'fand--Levitan normalization matrices $C_j$ remain
 unchanged for $1\le j\le N-1,$ i.e. we have
 \begin{equation*}
\tilde C_j=C_j,\qquad 1\le j\le N-1.\end{equation*}

\end{enumerate}
\end{enumerate}

 \section{The transformation  to add a bound state}
\label{section8}

In Section~\ref{section6} the spectrum of the matrix Schr\"odinger operator has been changed by completely removing
a bound state from the discrete spectrum without changing the continuous spectrum. In this section, we add a new bound state
to the spectrum of the matrix Schr\"odinger operator without changing the existing bound states and without changing
 the continuous spectrum. We determine the transformations of all relevant quantities when the new bound state is added. Since several bound
states can be added in succession, the technique presented in this section can be used to add any number of new bound states to the discrete spectrum.

We start with the unperturbed Schr\"odinger operator with the potential $V$ satisfying \eqref{2.2}
and belonging to 
$L^1_1(\mathbb R^+),$  
the boundary matrices
$A$ and $B$ describing the boundary condition as in \eqref{2.5}--\eqref{2.7}, the regular solution
$\varphi(k,x)$ satisfying the initial conditions \eqref{2.10}, the Jost solution $f(k,x)$ satisfying \eqref{2.9},
the spectral measure $d\rho$ described in \eqref{5.1}, the Jost matrix $J(k)$ defined in \eqref{2.11}, the scattering matrix $S(k)$ in \eqref{2.12}, and $N$ bound states with the 
energies $-\kappa_j^2,$ the Gel'fand--Levitan normalization matrices 
$C_j,$ the orthogonal projections
$Q_j$ onto $\text{\rm{Ker}}[J(i\kappa_j)],$ the orthogonal projections
$P_j$ onto $\text{\rm{Ker}}[J(i\kappa_j)^\dagger],$ 
the Gel'fand--Levitan normalized bound-state solutions $\Phi_j(x),$ the Marchenko normalized
bound-state solutions $\Psi_j(x),$ the dependency matrices $D_j,$ and the bound-state multiplicities $m_j$
for $1\le j\le N.$
We then add a new bound state with the energy $-\tilde\kappa_{N+1}^2$
with the Gel'fand--Levitan normalization matrix $\tilde C_{N+1},$ where
the $n\times n$ matrix $\tilde C_{N+1}$ is hermitian and nonnegative and has rank $\tilde m_{N+1}.$
We use a tilde to identify the quantities obtained after the addition of the bound state at $k=i\tilde\kappa_{N+1}$
with the Gel'fand--Levitan normalization matrix $\tilde C_{N+1}.$
Thus, the perturbed Schr\"odinger operator involves the potential $\tilde V(x),$ 
the boundary matrices
$\tilde A$ and $\tilde B,$ the regular solution
$\tilde\varphi(k,x),$ the Jost solution $\tilde f(k,x),$
the spectral measure $d\tilde\rho,$ the Jost matrix $\tilde J(k),$ the scattering matrix $\tilde S(k),$ and $N+1$ bound states with the energies $-\tilde\kappa_j^2,$ 
the Gel'fand--Levitan normalization matrices 
$\tilde C_j,$ the orthogonal projections
$\tilde Q_j$ onto $\text{\rm{Ker}}[\tilde J(i\tilde\kappa_j)],$ the orthogonal projections
$\tilde P_j$ onto $\text{\rm{Ker}}[\tilde J(i\tilde\kappa_j)^\dagger],$ 
the Gel'fand--Levitan normalized bound-state solutions $\tilde\Phi_j(x),$ the Marchenko normalized
bound-state solutions $\tilde\Psi_j(x),$ the dependency matrices $\tilde D_j,$ and the bound-state multiplicities $\tilde m_j$
for $1\le j \le N+1.$ Since the first $N$ bound states are unchanged, we know that
$\tilde\kappa_j=\kappa_j$ and $\tilde m_j=m_j$ for $1\le j\le N.$
Initially, it may be unclear if we have $\tilde Q_j=Q_j$ and $\tilde C_j=C_j$ for $1\le j\le N.$
In Theorem~\ref{theorem8.8} we confirm that we indeed have
 $\tilde Q_j=Q_j$ and $\tilde C_j=C_j$ for $1\le j\le N.$

To determine the transformations of all relevant quantities resulting from the
addition of the new bound state, we use the Gel'fand--Levitan method described in Section~\ref{section5} with
the input consisting the unperturbed quantities and the perturbation specified by $\tilde\kappa_{N+1}$ and $\tilde C_{N+1}.$
The Gel'fand--Levitan method enables us to explicitly determine \eqref{5.18},  \eqref{5.21}, and \eqref{5.23} 
when the spectral measures $d\tilde\rho$ and $d\rho$ differ from each other only in the bound states.
By specifying $\tilde C_{N+1}$ we also specify the positive integer $\tilde m_{N+1}$
satisfying $1\le \tilde m_{N+1}\le n.$ The specification of
$\tilde C_{N+1}$ also yields an orthogonal projection matrix $\tilde Q_{N+1}$ with rank
$\tilde m_{N+1}$ in such a way that 
\begin{equation}\label{8.1} 
 \tilde C_{N+1}\,\tilde Q_{N+1}=\tilde Q_{N+1}\,\tilde C_{N+1}=
\tilde C_{N+1}.
\end{equation}
The matrix $\tilde C_{N+1}$ must be related to 
the orthogonal projection $\tilde Q_{N+1}$ and 
a hermitian positive matrix $\tilde{\mathbf H}_{N+1}^{-1/2}$ as
\begin{equation}\label{8.2} 
\tilde C_{N+1}=
\tilde{\mathbf H}_{N+1}^{-1/2}\,\tilde Q_{N+1},
\end{equation}
which is the analog of \eqref{3.36}.
For the construction of the $n\times n$ matrix $\tilde{\mathbf H}_{N+1}^{-1/2},$
we proceed as follows.
We use an $n\times n$ hermitian nonnegative matrix $\tilde{\mathbf G}_{N+1}$ satisfying the equalities
\begin{equation}\label{8.3} 
\tilde{\mathbf G}_{N+1} \,\tilde Q_{N+1}=\tilde Q_{N+1}\, \tilde{\mathbf G}_{N+1} = \tilde{\mathbf G}_{N+1},
\end{equation}
in such a way that the  restriction of $\tilde{\mathbf G}_{N+1}$ to $\tilde Q_{N+1}\, \mathbb C^n$  is invertible.  We then define
the matrix $\tilde{\mathbf H}_{N+1}$ as
\begin{equation}\label{8.4}
\tilde{\mathbf H}_{N+1}:= I- \tilde Q_{N+1}+\tilde{\mathbf G}_{N+1}.
\end{equation}
From \eqref{8.4} it follows that $\tilde{\mathbf H}_{N+1}$ is indeed hermitian and nonnegative.
In fact, $\tilde{\mathbf H}_{N+1}$ is positive and hence invertible. The positivity can be established by
showing that the kernel of $\tilde{\mathbf H}_{N+1}$ contains only the zero vector in $\mathbb C^n.$
For the proof we proceed as follows. Since $\tilde{\mathbf G}_{N+1}$ is nonnegative, there exists a nonnegative matrix
$\tilde{\mathbf G}_{N+1}^{1/2}$ so that
$\tilde{\mathbf G}_{N+1}^{1/2}\tilde{\mathbf G}_{N+1}^{1/2}=\tilde{\mathbf G}_{N+1}.$ Thus, we can write
\eqref{8.4} in the equivalent form
\begin{equation}\label{8.5}
\tilde{\mathbf H}_{N+1}= I- \tilde Q_{N+1}+\tilde{\mathbf G}_{N+1}^{1/2}\,\tilde{\mathbf G}_{N+1}^{1/2}.
\end{equation}
For any column vector $v$ in $\mathbb C^N,$ from \eqref{8.5} we get
\begin{equation}
\label{8.6}
v^\dagger\, \tilde{\mathbf H}_{N+1} \,v= \left|(I- \tilde Q_{N+1})\, v\right|^2 +\left|\tilde{\mathbf G}_{N+1}^{1/2}\,v\right|^2.
\end{equation}
If $v$ belongs to the kernel of $\tilde{\mathbf H}_{N+1},$ then the left-hand side in \eqref{8.6} is zero, and hence the right-hand side yields
\begin{equation}\label{8.7}
 v=\tilde Q_{N+1} v,\quad \tilde{\mathbf G}_{N+1}^{1/2}\,v=0.
\end{equation}
After premultiplying the second equality in \eqref{8.7} by $\tilde{\mathbf G}_{N+1}^{1/2},$ we conclude that $\tilde{\mathbf G}_{N+1}\,v=0$ when $v$ belongs to
$\tilde Q_{N+1}\, \mathbb C^n.$ On the other hand, we know that the restriction of
$\tilde{\mathbf G}_{N+1}$ to $\tilde Q_{N+1}\, \mathbb C^n$  is invertible. Thus, the column vector
$v$ must be the zero vector. This establishes the fact that $\tilde{\mathbf H}_{N+1}$ is invertible.
We remark that, with the help of \eqref{8.4}, we prove that $\tilde{\mathbf H}_{N+1}$ satisfies
\begin{equation}\label{8.8}
\tilde{\mathbf H}_{N+1} \,\tilde Q_{N+1}=\tilde Q_{N+1}\, \tilde{\mathbf H}_{N+1}.
\end{equation}

In the next proposition we show that the orthogonal projection $\tilde Q_{N+1},$ the matrix $\tilde{\mathbf G}_{N+1},$
and the matrix $\tilde{\mathbf H}_{N+1}$ are uniquely determined by $\tilde C_{N+1}.$

\begin{proposition}
\label{proposition8.1}
Assume that a nonnegative matrix $\tilde C_{N+1}$ is related to
a positive invertible matrix $\tilde{\mathbf H}_{N+1}$ and
an orthogonal projection $\tilde Q_{N+1}$ as in \eqref{8.2}
in such a way that \eqref{8.1}, \eqref{8.3}, and \eqref{8.4} hold, with 
$\tilde{\mathbf G}_{N+1}$ being the matrix appearing in \eqref{8.3} and \eqref{8.4}.
Then, $\tilde C_{N+1}$ uniquely determines the
matrices  $\tilde Q_{N+1},$ $\tilde{\mathbf G}_{N+1},$
and $\tilde{\mathbf H}_{N+1}.$
\end{proposition}

\begin{proof}
Let us decompose $\mathbb C^n$ as a direct sum given by
\begin{equation}\label{8.9}
\mathbb C^n=[\tilde Q_{N+1} \mathbb C^n]  \oplus[(I-\tilde Q_{N+1}) \,\mathbb C^n].
\end{equation}
From \eqref{8.3} we get
\begin{equation}\label{8.10}
\tilde{\mathbf G}_{N+1} = [\tilde{\mathbf G}_{N+1}]_1\oplus 0,
\end{equation}
where we use $[\tilde{\mathbf G}_{N+1}]_1$ to denote the restriction of $\tilde{\mathbf G}_{N+1}$ to $\tilde Q_{N+1}\, \mathbb C^n$ 
and use $0$ to denote the zero operator on $[(I-\tilde Q_{N+1})\, \mathbb C^n].$
From  \eqref{8.4} and\eqref{8.10} we see that we have
\begin{equation}\label{8.11}
\tilde{\mathbf H}_{N+1} = [\tilde{\mathbf G}_{N+1}]_1\oplus I,
\end{equation}
where $I$ on the right-hand side denotes the identity operator on $[(I-\tilde Q_{N+1})\, \mathbb C^n].$
Since $\tilde{\mathbf H}_{N+1}$ is positive and invertible in $\mathbb C^n,$ it follows that
the positive matrix $\tilde{\mathbf H}_{N+1}^{1/2}$ is well defined by
using the analog of \eqref{3.35}. Hence, the positive matrix
$\tilde{\mathbf H}_{N+1}^{-1/2}$ is also well defined as the inverse of $\tilde{\mathbf H}_{N+1}^{1/2}.$
From \eqref{8.11} we obtain
\begin{equation}\label{8.12}
\tilde{\mathbf H}_{N+1}^{-1/2}=([\tilde{\mathbf G}_{N+1}]_1)^{-1/2} \oplus I,
\end{equation}
where we used the fact that $[\tilde{\mathbf G}_{N+1}]_1$ is invertible on $\tilde Q_{N+1}\, \mathbb C^n.$
From \eqref{8.2} and \eqref{8.12} we observe the decomposition
\begin{equation}\label{8.13}
\tilde C_{N+1}= ([\tilde{\mathbf G}_{N+1}]_1)^{-1/2} \oplus 0.
\end{equation}
Since $([\tilde{\mathbf G}_{N+1}]_1)^{-1/2}$ is bijective on $\tilde Q_{N+1}\,\mathbb C^n,$ it follows from \eqref{8.13} that  
$\tilde C_{N+1}$ is onto $\tilde Q_{N+1}\,\mathbb C^n.$ Consequently, we conclude that the range of $\tilde C_{N+1}$ is 
$\tilde Q_{N+1} \mathbb C^n$  and that
$\tilde Q_{N+1}$ is the orthogonal projection onto the range of $\tilde C_{N+1}.$
Thus, $\tilde Q_{N+1}$ is uniquely determined by $\tilde C_{N+1}.$
From \eqref{8.13} it follows that $([\tilde{\mathbf G}_{N+1}]_1)^{-1/2}$ is uniquely determined by
 $\tilde C_{N+1},$ and hence $[\tilde{\mathbf G}_{N+1}]_1$ is also uniquely determined by $\tilde C_{N+1}.$ Moreover,
 from \eqref{8.10} we see that $\tilde{\mathbf G}_{N+1}$ is also uniquely determined by $\tilde C_{N+1}.$
Finally, with the help of \eqref{8.4} we see that 
$\tilde{\mathbf H}_{N+1}$ is also uniquely determined by $\tilde C_{N+1}.$
\end{proof}

Having related $\tilde C_{N+1}$ to $\tilde{\mathbf H}_{N+1}$ and $\tilde Q_{N+1}$ as in \eqref{8.2}, we form the Gel'fand--Levitan kernel $G(x,y)$ as
\begin{equation}\label{8.14}
G(x,y)= \varphi(i\tilde\kappa_{N+1},x) \,\tilde C_{N+1}^2\, \varphi(i \tilde\kappa_{N+1},y)^\dagger,
\end{equation}
where we recall that $\varphi(k,x)$ is the unperturbed regular solution.
The next theorem shows how we obtain the solution to
the  Gel'fand--Levitan system of integral equations \eqref{5.49} with the kernel $G(x,y)$ given in \eqref{8.14}.

\begin{theorem}\label{theorem8.2}
Consider
the unperturbed Schr\"odinger operator with the potential
$V$ satisfying \eqref{2.2} and belonging to $L^1_1(\mathbb R^+),$  with the selfadjoint
boundary condition \eqref{2.5} described by the boundary matrices $A$ and $B$ satisfying
\eqref{2.6} and \eqref{2.7}, 
with the regular solution $\varphi(k,x)$ satisfying the initial conditions
 \eqref{2.10}, and
 containing
$N$ bound states with the energies $-\kappa_j^2$ and
the Gel'fand--Levitan normalization matrices $C_j$ 
for $1\le j\le N.$ 
Choose the kernel $G(x,y)$ of the  Gel'fand--Levitan system of integral equations \eqref{5.49} as in \eqref{8.14},
where $\tilde\kappa_{N+1}$ is a positive constant
distinct from $\kappa_j$ for $1\le j\le N$ and
$\tilde C_{N+1}$ is the  $n\times n$   selfadjoint nonnegative matrix  in \eqref{8.2}.
Then, the corresponding solution $\mathcal A(x,y)$ to \eqref{5.49} is given by 
\begin{equation}
\label{8.15}
\mathcal A(x,y)=-\xi_{N+1}(x)\,\Omega_{N+1}(x)^+\,\xi_{N+1}(y)^\dagger,\qquad 0\le y < x,\end{equation}
with the $n\times n$ matrices $\xi_{N+1}(x)$  and $\Omega_{N+1}(x)$  defined as
\begin{equation}
\label{8.16}
\xi_{N+1}(x):=\varphi(i\tilde\kappa_{N+1},x)\,\tilde C_{N+1},\end{equation}
\begin{equation}
\label{8.17}
\Omega_{N+1}(x):=
\tilde Q_{N+1}+\int_0^x dy\,\xi_{N+1}(y)^\dagger\,
\xi_{N+1}(y),
\end{equation}
where we recall that $\Omega_{N+1}(x)^+$ denotes the Moore--Penrose inverse of $\Omega_{N+1}(x).$

\end{theorem}

\begin{proof} 
The proof is similar to the proof of
Theorem~\ref{theorem6.1}. We decompose $\mathbb C^n$ into the direct sum as in \eqref{8.9}.
From \eqref{8.17} we observe that $\Omega_{N+1}(x)$ maps $\tilde Q_{N+1}\, \mathbb C^n$ into $\tilde Q_{N+1}\, \mathbb C^n $ and that
 it acts as the zero mapping on  $(I-\tilde Q_{N+1}) \,\mathbb C^n.$ Hence, we can decompose $\Omega_{N+1}(x)$ as
\begin{equation}\label{8.18}
\Omega_{N+1}(x)= [\Omega_{N+1}(x)]_1\oplus 0,
\end{equation}
where we use $[\Omega_{N+1}(x)]_1$ to denote the restriction of $\Omega_{N+1}(x)$ to $\tilde Q_{N+1}\, \mathbb C^n.$ 
Next, we prove that $[\Omega_{N+1}(x)]_1$ is invertible on $\tilde Q_{N+1}\, \mathbb C^n$ 
by proceeding as follows. If $v$ is a column vector in $\tilde Q_{N+1}\, \mathbb C^n,$ then as seen from \eqref{8.18} we have
$\Omega_{N+1}(x)\, v=[\Omega_{N+1}(x)]_1\, v.$
Hence, using $v$ in \eqref{8.17} we obtain
\begin{equation}
\label{8.19}
v^\dagger [\Omega_{N+1}(x)]_1\, v= | v|^2+ \int_0^x dy \,\left| \xi_{N+1}(y)\,v\right|^2.
\end{equation}
If $v$ belongs to the kernel of $[\Omega_{N+1}(x)]_1,$ then the left-hand side in \eqref{8.19} must
be zero. Each of the two terms on the right-hand side is nonnegative and hence we must also have $|v|^2=0,$ which
implies that $v$ is the zero vector in  $\tilde Q_{N+1}\, \mathbb C^n.$ 
This establishes the fact that $[\Omega_{N+1}(x)]_1$ is invertible on $\tilde Q_{N+1}\, \mathbb C^n.$ 
We then use the analog of \eqref{3.9} and conclude that
\begin{equation}\label{8.20}
\Omega_{N+1}(x)^+=[\Omega_{N+1}(x)]_1^{-1}\oplus 0.
\end{equation}
From \eqref{8.18} and \eqref{8.20} we observe that
\begin{equation}\label{8.21}
\Omega_{N+1}(x)\, \Omega_{N+1}(x)^+=   \Omega_{N+1}(x)^+ \,\Omega_{N+1}(x)=I\oplus 0= \tilde Q_{N+1},
\end{equation}  
which yields the analog of \eqref{6.21} and \eqref{6.22}.
The proof of the theorem is completed by exploiting the analogy
between $\Omega_{N+1}(x)$ defined in \eqref{8.17}
and $W_N(x)$ defined in \eqref{6.10} and by using the steps
given in the proof of Theorem~\ref{theorem6.1}.
\end{proof}

In the next theorem we present some relevant properties of the quantity $\mathcal A(x,y)$ 
appearing in \eqref{8.15}.
This theorem is the counterpart of Theorem~\ref{theorem6.2}, and hence their proofs are similar.

\begin{theorem} \label{theorem8.3}
Consider
the unperturbed Schr\"odinger operator with the potential
$V$ satisfying \eqref{2.2} and belonging to $L^1_1(\mathbb R^+),$  with the selfadjoint
boundary condition \eqref{2.5} described by the boundary matrices $A$ and $B$ satisfying
\eqref{2.6} and \eqref{2.7}, and
 containing $N$ bound states with the energies $-\kappa_j^2$ and
the Gel'fand--Levitan normalization matrices $C_j$ for $1\le j\le N.$ 
Let $\mathcal A(x,y)$ be the matrix-valued quantity given in \eqref{8.15}, and define the
 perturbed potential $\tilde V(x)$ as
\begin{equation}
\label{8.22}
\tilde V(x):=V(x)-2\,\ds\frac{d}{dx}\left[ \xi_{N+1}(x)\,\Omega_{N+1}(x)^+\,\xi_{N+1}(x)^\dagger\right],\end{equation}
where $\xi_{N+1}(x)$ and $\Omega_{N+1}(x)$ are the $n\times n$ matrices defined in \eqref{8.16} and \eqref{8.17}, respectively.
Then, we have the following:

\begin{enumerate}

\item[\text{\rm(a)}] The quantity  $\mathcal A(x,y)$
is continuously differentiable in the set $\mathbb D,$
where $\mathbb D$ is the subset of $\mathbb R^2$ defined in \eqref{5.19}.

\item[\text{\rm(b)}] The second partial derivatives $\mathcal A_{xx}(x,y)$
and $\mathcal A_{yy}(x,y)$ are locally integrable in $\mathbb D.$

\item[\text{\rm(c)}] 
The quantity $\mathcal A(x,y)$ satisfies the second-order matrix-valued partial differential equation \eqref{5.20}.

\item[\text{\rm(d)}] The perturbed potential $\tilde V(x)$ defined in \eqref{8.22} is related to
the unperturbed potential $V(x)$ and the quantity $\mathcal A(x,y)$ as
in \eqref{5.21}.

\end{enumerate}
\end{theorem}

\begin{proof}
By arguing as in the proof of \eqref{6.25}, we prove that the $x$-derivative of $\Omega_{N+1}(x)^+$ is given by
\begin{equation}\label{8.23} 
\left[\Omega_{N+1}(x)^+\right]'=-\Omega_{N+1}(x)^+\,\xi_{N+1}(x)^\dagger\,\xi_{N+1}(x)\,\Omega_{N+1}(x)^+.
\end{equation}
We then use the steps analogous to the steps in the proof of
Theorem~\ref{theorem6.2}
and establish the results stated in (a)--(d).
\end{proof}

The results presented in the next proposition are needed later on. Those results are related to the large spacial asymptotics
of the unperturbed regular solution $\varphi(k,x)$ and an exponentially growing matrix solution to
\eqref{2.1} when $k=i\tilde\kappa_{N+1}$ with the unperturbed potential $V(x).$

\begin{proposition}
\label{proposition8.4}
Consider
the unperturbed Schr\"odinger operator with the potential
$V$ satisfying \eqref{2.2} and belonging to $L^1_1(\mathbb R^+),$  with the selfadjoint
boundary condition \eqref{2.5} described by the boundary matrices $A$ and $B$ satisfying
\eqref{2.6} and \eqref{2.7}, 
with the regular solution $\varphi(k,x)$ satisfying the initial conditions
 \eqref{2.10}, and
 containing
$N$ bound states with the eigenvalues $-\kappa_j^2$ and
the Gel'fand--Levitan normalization matrices $C_j$ 
for $1\le j\le N.$ 
Let $\tilde \kappa_{N+1}$ be a positive constant distinct from $\kappa_j$ for $1\le j\le N.$
We have the following:

\begin{enumerate}

\item[\text{\rm(a)}] For each nonnegative constant $a$ and each fixed $k\in\overline{\mathbb C^+}\setminus\{0\},$ the Schr\"odinger equation
\eqref{2.1} has an $n\times n$ matrix-valued solution $g(k,x)$ 
satisfying
\begin{equation*}
g(k,x)=e^{-ikx}\left[I+o(1)\right],\quad g'(k,x)=-e^{-ikx}\left[ik\,I+o(1)\right],
\qquad x\to+\infty.
\end{equation*}
Furthermore,
for each fixed $k\in\overline{\mathbb C^+},$ the Jost solution $f(k,x)$
and the solution $g(k,x)$ form a fundamental set of solutions to \eqref{2.1}.

\item[\text{\rm(b)}] For each fixed parameter
$\varepsilon$ with $0<\varepsilon<1,$ the quantity $g(i\tilde\kappa_{N+1},x)$ has the asymptotics 
\begin{equation}\label{8.25}
g(i\tilde\kappa_{N+1},x)=  e^{\tilde\kappa_{N+1} x}\left[I+  O\left(q_6(x)+q_{14}(a,\varepsilon,x)+q_{15}(\varepsilon,x)
\right)\right], \qquad x \to +\infty,
\end{equation}
where $q_6(x)$ is the quantity in \eqref{6.43}
 and we have defined
\begin{equation}
\label{8.26}
q_{14}(a,\varepsilon,x):= e^{-2 \varepsilon \tilde\kappa_{N+1} x}  \int_a^{(1-\varepsilon)x} dy\, |V(y)|,
\end{equation}
\begin{equation}
\label{8.27}
q_{15}(\varepsilon,x):= \int_{(1-\varepsilon) x}^x dy \, e^{-2 \tilde\kappa_{N+1} (x-y)}\, |V(y)|.
\end{equation}

\item[\text{\rm(c)}] The $x$-derivative $g'(i\tilde\kappa_{N+1},x)$ has the asymptotics
\begin{equation}\label{8.28}
g'(i\tilde\kappa_{N+1},x )=\tilde\kappa_{N+1}\, e^{\tilde\kappa_{N+1}x} \left[ I+  O\left( q_6(x)+ q_{14}(a,\varepsilon,x)+q_{15}(\varepsilon,x)
\right)\right],\qquad x\to+\infty.
\end{equation}

\item[\text{\rm(d)}] The quantity $\varphi(i\tilde\kappa_{N+1},x)$ has the representation
\begin{equation}\label{8.29}
\varphi(i\tilde\kappa_{N+1},x)= f(i\tilde\kappa_{N+1},x)\,K_{N+1}+  g(i\tilde\kappa_{N+1},x)\,L_{N+1},
\end{equation}
for some $n \times n$ matrices $K_{N+1}$ and $L_{N+1},$ where the matrix $L_{N+1}$ is invertible.

\end{enumerate}

\begin{proof}
We remark that (a) follows from Propositions~3.2.2 and 3.2.3 of \cite{AW2021}. 
For any $a\ge 0,$ from (3.2.39) and (3.2.63) of \cite{AW2021} we have
\begin{equation}\label{8.30}
 e^{-\tilde\kappa_{N+1} x}g(i\tilde\kappa_{N+1},x)= I+ O\left( q_6(x)
+q_{16}(a,x)\right), \qquad x \to +\infty,
\end{equation}
where we have defined
$q_{16}(a,x)$ as
\begin{equation*}
 q_{16}(a,x):= \int_a^x dy \, e^{-2 \tilde\kappa_{N+1} (x-y)} \, |V(y)| .
\end{equation*}
We remark that, even though the result of (3.2.63) in \cite{AW2021} is derived for a large positive $a,$
it also holds for any smaller $a$-value, and hence \eqref{8.30} holds for any nonnegative $a.$
We also remark that we have suppressed the $a$-dependence in $g(k,x)$ in our notation.
Using any $\varepsilon$ satisfying $0<\varepsilon<1,$ we can write \eqref{8.30} as
\begin{equation}\label{8.32}
 e^{-\tilde\kappa_{N+1} x}g(i\tilde\kappa_{N+1},x)=  I +O\left( q_6(x)
+q_{17}(a,\varepsilon,x)+
 q_{15}(\varepsilon,x)\right), \qquad x \to +\infty,
\end{equation}
where we have defined $q_{17}(a,\varepsilon,x)$ as
\begin{equation*}
q_{17}(a,\varepsilon,x):=  \int_a^{(1-\varepsilon) x} dy\, e^{-2\tilde\kappa_{N+1}(x-y)}|V(y)|.
\end{equation*}
We note that \eqref{8.25} follows from \eqref{8.32}
after using $y\le(1-\varepsilon)x,$ and hence the 
proof of (b) is complete. For the proof of (c), we use 
\begin{equation}\label{8.34}
 \frac{d}{dx}\left[  e^{-\tilde\kappa_{N+1} x} g(i\tilde\kappa_{N+1},x)\right]=
O\left(q_{16}(a,x)\right), \qquad x \to +\infty, 
\end{equation}
which is implied by (3.2.71) of \cite{AW2021}.
By proceeding as in the proof of \eqref{8.25}, with the help of \eqref{8.32} and \eqref{8.34} we obtain \eqref{8.28}, which completes the proof
of (c). For the proof of (d) we proceed as follows.
From (a) we know that $f(i\tilde\kappa_{N+1},x)$ and $g(i \tilde\kappa_{N+1},x)$ form a fundamental set of solutions to \eqref{2.1} 
when $k= i\tilde\kappa_{N+1}.$ Hence, $\varphi(i\tilde\kappa_{N+1},x)$ can be expressed as in \eqref{8.29}.
The invertibility of $L_{N+1}$ is proved by showing that the kernel of $L_{N+1}$ contains only
the zero vector in $\mathbb C^n,$ and for the proof we proceed as follows. For any column vector $v$ satisfying $L_{N+1}v=0,$ from \eqref{8.29} we get
\begin{equation}\label{8.35}
\varphi(i\tilde\kappa_{N+1},x)\,v= f(i\tilde\kappa_{N+1},x)\,K_{N+1}\,v.
\end{equation}
From the first equality of \eqref{2.9} we observe that the right-hand side
of \eqref{8.35} behaves as $O(e^{-\tilde\kappa_{N+1}x})$ as $x\to+\infty,$ and hence that right-hand side is
a square-integrable solution to \eqref{2.1} at $k=i\tilde\kappa_{N+1}.$
On the other hand, the left-hand side of \eqref{8.35} cannot be a square-integrable
solution to \eqref{2.1} because otherwise we would have a bound-state at $k=i\tilde\kappa_{N+1}$
for the unperturbed problem unless $v$ is the zero vector. Thus, the proof of (d) is complete. 
\end{proof}
\end{proposition}

The next theorem is the counterpart of Theorem~\ref{theorem6.3}.
It shows that the quantity $\tilde\varphi(k,x)$ in \eqref{5.18} with $\mathcal A(x,y)$ in \eqref{8.15} is the perturbed regular solution
to \eqref{2.1} with the perturbed potential
$\tilde V(x)$ in \eqref{8.22} and
with the appropriately determined boundary matrices $\tilde A$ and $\tilde B.$

\begin{theorem}\label{theorem8.5}
Consider the unperturbed Schr\"odinger operator with the potential
$V$ satisfying \eqref{2.2} and belonging to $L^1_1(\mathbb R^+),$   with the selfadjoint
boundary condition \eqref{2.5} described by the boundary matrices $A$ and $B$ satisfying
\eqref{2.6} and \eqref{2.7},  with the regular solution $\varphi(k,x)$ satisfying the initial conditions \eqref{2.10},
and with $N$ bound states with the energies $-\kappa_j^2$ and the Gel'fand--Levitan normalization matrices
$C_j$ for $1\le j\le N.$  
Let $\tilde\varphi(k,x)$ be the quantity defined in \eqref{5.18} with $\mathcal A(x,y)$ expressed as in \eqref{8.15}.  
Then, we have the following:

\begin{enumerate}

\item[\text{\rm(a)}] The quantity  
$\tilde\varphi(k,x)$ is a solution to the perturbed Schr\"odinger equation \eqref{2.1}
 with the potential $\tilde V(x)$ given in \eqref{8.22}. We can express
$\tilde\varphi(k,x)$ also as
 \begin{equation}\label{8.36}
\tilde\varphi(k,x)=\varphi(k,x)- \xi_{N+1}(x)\,\Omega_{N+1}(x)^+ \int_0^x dy\,\xi_{N+1}(y)^\dagger\, \varphi(k,y).
\end{equation}

\item[\text{\rm(b)}] The perturbed potential $\tilde V(x)$ appearing in \eqref{8.22} 
satisfies \eqref{2.2}. Moreover, for every nonnegative constant $a$ and
the parameter $\varepsilon$ satisfying $0<\varepsilon<1,$ the potential increment $\tilde V(x)-V(x)$ has the
asymptotic behavior
\begin{equation*}
\tilde V(x)- V(x)=O\left(q_{18}(a,\varepsilon,x)\right),\qquad x\to+\infty,
\end{equation*}
with $q_{18}(a,\varepsilon,x)$ defined as
\begin{equation*}
q_{18}(a,\varepsilon,x):=  x \, e^{-2\tilde\kappa_{N+1}x}+ q_6(x)
+  q_{14}(a,\varepsilon,x)
+q_{15}(\varepsilon,x),
\end{equation*}
and where we recall that $q_6(x),$ $q_{14}(a,\varepsilon,x),$ and $q_{15}(\varepsilon,x)$
are the quantities in \eqref{6.43}, \eqref{8.26}, and \eqref{8.27}, respectively.
If the unperturbed potential $V$ is further restricted to  $L^1_{1+\epsilon}(\mathbb R^+)$ for some fixed $\epsilon \ge 0,$ then 
the perturbed potential $\tilde V$ belongs to $L^1_\epsilon(\mathbb R^+).$

\item[\text{\rm(c)}] If the unperturbed potential $V(x)$ is further restricted to satisfy
\begin{equation*}
|V(x)| \le c\,e^{-\alpha x}, \qquad x \ge x_0,
\end{equation*}
for some positive constants $\alpha$ and $x_0$ and with $c$
denoting a generic constant, then for every $0<\varepsilon<1$  the potential increment
$\tilde V(x)-V(x)$ has the asymptotic behavior as $x\to+\infty$ given by
\begin{equation}\label{8.40}
\tilde V(x) -V(x) =\begin{cases}
O\left(e^{-\alpha x}+ e^{-2 \varepsilon\tilde\kappa_{N+1}x}\right),\qquad  \alpha\le 2 \tilde\kappa_{N+1}, \\
\noalign{\medskip}
O\left( e^{-2\varepsilon\tilde\kappa_{N+1} x}\right), \qquad \alpha > 2  \tilde\kappa_{N+1}.
\end{cases}
\end{equation}

\item[\text{\rm(d)}]
If the unperturbed potential $V(x)$ has compact support, then the perturbed potential $\tilde V(x)$ has the asymptotic behavior
\begin{equation*}
\tilde V(x)= O\left(x\, e^{-2\tilde\kappa_{N+1} x}\right), \qquad x \to +\infty.
\end{equation*}

\item[\text{\rm(e)}]
 For $ k\ne i\tilde\kappa_{N+1},$ the perturbed quantity $\tilde\varphi(k,x)$ can be expressed as
\begin{equation}
\label{8.42}
\tilde\varphi(k,x)=\varphi(k,x)-\ds\frac{1}{k^2+\tilde\kappa_{N+1}^2}
\,\xi_{N+1}(x) \,\Omega_{N+1}(x)^+\left[\xi'_{N+1}(x)^\dagger\,\varphi(k,x)-\xi_{N+1}(x)^\dagger
\,\varphi'(k,x)\right].
\end{equation}

\item[\text{\rm(f)}] The perturbed quantity $\tilde\varphi(k,x)$ satisfies the initial conditions \eqref{5.17}
where the matrices $\tilde A$ and $\tilde B$
are expressed in terms of the unperturbed
boundary matrices $A$ and $B$ and
the Gel'fand--Levitan normalization matrix $\tilde C_{N+1}$  in \eqref{8.2} as
\begin{equation}\label{8.43}
\tilde A=A, \quad \tilde B=B-A\,\tilde C_{N+1}^2 A^\dagger A.
\end{equation}

\item[\text{\rm(g)}] 
The matrices $\tilde A$ and $\tilde B$ appearing in \eqref{8.43} satisfy \eqref{2.6} and \eqref{2.7}. Hence, 
as a consequence of (a) and (f), 
the quantity $\tilde\varphi(k,x)$ corresponds to the regular solution 
to the matrix Schr\"odinger equation \eqref{2.1} with the potential $\tilde V(x)$ in \eqref{8.22} and
with the selfadjoint
boundary condition \eqref{2.5} with $A$ and $B$ there replaced with
$\tilde A$ and $\tilde B,$ respectively.

\end{enumerate}

\end{theorem}

\begin{proof}
The proof of (a) is obtained 
by using Theorems~\ref{theorem5.2} and \ref{theorem8.3}. The proof of (d) is obtained by using (b).
The proofs of (e), (f), (g) are similar to the proofs of (e), (f), (g) of Theorem~\ref{theorem6.3}.
In order to prove (b) and (c), we first obtain some estimates by proceeding as follows.
With the help of \eqref{6.48}, \eqref{6.49}, \eqref{8.16}, \eqref{8.25}, \eqref{8.28}, and \eqref{8.29}, 
we get the asymptotics for $\xi_{N+1}(x)$ defined in \eqref{8.16} as
\begin{equation}\label{8.44}
\xi_{N+1}(x)= e^{\tilde\kappa_{N+1}x}  \left[ L_{N+1}+ O\left(q_{19}(a,\varepsilon,x)
\right)\right] \tilde C_{N+1},\qquad x\to+\infty,
\end{equation}
\begin{equation}\label{8.45}
\xi'_{N+1}(x)=\tilde\kappa_{N+1} e^{\tilde\kappa_{N+1}x}  \left[ L_{N+1}+ O\left(q_{19}(a,\varepsilon,x)\right)\right] \tilde C_{N+1},\qquad x\to+\infty,
\end{equation}
where we have defined
\begin{equation*}
q_{19}(a,\varepsilon,x):=e^{-2\tilde\kappa_{N+1}x}+ q_6(x)
+  q_{14}(a,\varepsilon,x)
+q_{15}(\varepsilon,x),
\end{equation*}
by recalling that 
$L_{N+1}$ is the invertible matrix appearing in \eqref{8.29}.
In terms of the parameters $a$ and $\varepsilon,$ let us choose a positive constant $b$ satisfying
$b \ge a/(1-\varepsilon)$ and
\begin{equation}\label{8.46}
q_{19}(a,\varepsilon,x)
< 1, \qquad x \ge b.
\end{equation}
Using \eqref{8.44}--\eqref{8.46}, for the quantity $\Omega_{N+1}(x)$ given in
\eqref{8.17} we obtain
the asymptotics as $x\to+\infty$ given by
\begin{equation}\label{8.47}
\Omega_{N+1}(x)-\Omega_{N+1}(b)= \frac{1}{2\tilde\kappa_{N+1}}\left(e^{2\tilde\kappa_{N+1}x} -e^{2\tilde\kappa_{N+1}b}\right)\tilde C_{N+1}
 L_{N+1}^\dagger L_{N+1}\tilde C_{N+1}
+O\left(q_{20}(b,\varepsilon,x) \right),
\end{equation}
where $q_{20}(b,\varepsilon,x)$ is some quantity with the asymptotic behavior
\begin{equation}\label{8.48}
O\left(q_{20}(b,\varepsilon,x) \right)=
 \tilde C_{N+1}\int_b^x dy\, e^{2\tilde\kappa_{N+1}y}\left[  O\left(q_{19}(a,\varepsilon,y)
\right)\right] \tilde C_{N+1}
,\qquad x\to+\infty.
\end{equation}
Using integration by parts, we obtain
\begin{equation}\label{8.49}
\int_b^x dy\, e^{2\tilde\kappa_{N+1}y}  q_6(y)= \frac{1}{2\tilde\kappa_{N+1}}\left[ e^{2\tilde\kappa_{N+1}x} q_6(x)
  -  e^{2\tilde\kappa_{N+1}b}\, q_6(b) +q_{21}(b,x)
 \right],
\end{equation}
\begin{equation}\label{8.50}
\begin{split}
\int_b^x dy\, e^{2\tilde\kappa_{N+1}y} &q_{14}(a,\varepsilon,y)=-\frac{1}{2\tilde\kappa_{N+1}}q_{22}(b,\varepsilon,x)\\
&+\frac{1}{2(1-\varepsilon)\tilde\kappa_{N+1}}
\left[  e^{2\tilde\kappa_{N+1}x} q_{14}(a,\varepsilon,x)-e^{2\tilde\kappa_{N+1}b} q_{14}(a,\varepsilon,b)\right],
\end{split}
\end{equation}
\begin{equation}\label{8.51}
\begin{split}
\int_b^x dy\, &e^{2\tilde\kappa_{N+1}y} q_{15}(\varepsilon,y)\\
&=
  \frac{1}{2\tilde\kappa_{N+1}} \left[e^{2\tilde\kappa_{N+1}x}  q_{15}(\varepsilon,x)-e^{2\tilde\kappa_{N+1}b} 
  q_{15}(\varepsilon,b)-q_{21}(b,x)+(1-\varepsilon)\,q_{22}(b,\varepsilon,x)\right],
  \end{split}
\end{equation}
where we have defined
\begin{equation*}
q_{21}(b,x):=\int_b^x dy\, e^{2\tilde\kappa_{N+1}y} \,|V(y)|,
\end{equation*}
\begin{equation*}
q_{22}(b,\varepsilon,x):=\int_b^x dy\, e^{2(1-\varepsilon)\tilde\kappa_{N+1}y} \,|V((1-\varepsilon)y)|.
\end{equation*}
Using \eqref{8.47}--\eqref{8.51}, we get the asymptotics for $\Omega_{N+1}(x)$ as
\begin{equation}\label{8.55}
\Omega_{N+1}(x)= \frac{1}{2\tilde\kappa_{N+1}}e^{2\tilde\kappa_{N+1}x} \tilde C_{N+1}\left[ L_{N+1}^\dagger L_{N+1}+
O\left(q_{18}(a,\varepsilon,x)\right)\right]\tilde C_{N+1},\qquad x\to+\infty.
\end{equation}
From \eqref{8.2} we see that we can decompose $\tilde C_{N+1} L_{N+1}^\dagger L_{N+1} \tilde C_{N+1}$
as in the direct sum given in \eqref{8.9}, where
we use $[\tilde C_{N+1} L_{N+1}^\dagger L_{N+1}
 \tilde C_{N+1}]_1$ to denote the restriction of $\tilde C_{N+1} L_{N+1}^\dagger L_{N+1} \tilde C_{N+1}$ to $\tilde Q_{N+1}\mathbb C^n.$ 
Thus, we have
\begin{equation}\label{8.56}
\tilde C_{N+1} L_{N+1}^\dagger L_{N+1} \tilde C_{N+1}=[\tilde C_{N+1} L_{N+1}^\dagger L_{N+1} \tilde C_{N+1}]_1\oplus 0.
\end{equation}
We remark that $[\tilde C_{N+1} L_{N+1}^\dagger L_{N+1} \tilde C_{N+1}]_1$ is invertible, and this is proved by showing that the
only vector $v$ in $\mathbb C^n$ satisfying  
\begin{equation}\label{8.57}
[\tilde C_{N+1} L_{N+1}^\dagger L_{N+1} \tilde C_{N+1}]_1\,v=0,
\end{equation}
is the zero vector. For the proof we proceed as follows. From \eqref{8.1} and \eqref{8.57}, we see that
we equivalently need to show that
\begin{equation}\label{8.58}
\tilde C_{N+1} L_{N+1}^\dagger L_{N+1} \tilde C_{N+1} \,v=0,\qquad v\in \tilde Q_{N+1}\,\mathbb C^n,
\end{equation}
is possible only when $v$ is the zero vector. From \eqref{8.58}, with the help of \eqref{8.1} and \eqref{8.2} we obtain
\begin{equation}\label{8.59}
 v^\dagger \tilde C_{N+1} L_{N+1}^\dagger L_{N+1} \tilde C_{N+1} v= \left| L_{N+1} \tilde{\mathbf H}_{N+1}^{-1/2} v\right|^2.
\end{equation}
As a result of \eqref{8.58}, the left-hand side of \eqref{8.59} is zero. Hence, the right-hand side of \eqref{8.59} implies
$ L_{N+1} \tilde{\mathbf H}_{N+1}^{-1/2} v=0.$ Since the matrices $L_{N+1}$ and
$\tilde{\mathbf H}_{N+1}^{-1/2}$ are both invertible, we conclude that $v$ must be the zero vector.
With the help of \eqref{3.9} and \eqref{8.56}, we get
\begin{equation}\label{8.60}
\left(\tilde C_{N+1} L_{N+1}^\dagger L_{N+1} \tilde C_{N+1}\right)^+=\left([\tilde C_{N+1} L_{N+1}^\dagger L_{N+1} 
\tilde C_{N+1}]_1\right)^{-1} \oplus 0.
\end{equation}
Using \eqref{3.9}, \eqref{8.55}, and \eqref{8.60}, we obtain the asymptotics for $\Omega_{N+1}(x)^+$ as $x\to+\infty$ as
\begin{equation*}
\Omega_{N+1}(x)^+= 2\tilde\kappa_{N+1} e^{-2\tilde\kappa_{N+1}x}\tilde Q_{N+1}\left[\left(\tilde C_{N+1} L_{N+1}^\dagger L_{N+1}  \tilde C_{N+1}\right)^+
+
O\left(q_{18}(a,\varepsilon,x)\right)\right]\tilde Q_{N+1}\oplus 0,
\end{equation*}
which is equivalent to the asymptotics as $x\to+\infty$ given by
\begin{equation}\label{8.62}
\Omega_{N+1}(x)^+= 2\tilde\kappa_{N+1} e^{-2\tilde\kappa_{N+1}x}\tilde Q_{N+1}\left[\left(\tilde C_{N+1} L_{N+1}^\dagger L_{N+1}  \tilde C_{N+1}\right)^+
+ 
O\left(q_{18}(a,\varepsilon,x)\right)\right]\tilde Q_{N+1}.
\end{equation}
Next, by using \eqref{8.23}, \eqref{8.44}, \eqref{8.62}, and the second equality of \eqref{3.3}, we obtain the asymptotics as $x\to+\infty$ for
the $x$-derivative of $\Omega_{N+1}(x)^+$ as
\begin{equation}\label{8.63}
 \left[   \Omega(x)^+_{N+1}\right]' =
- 4 \tilde\kappa_{N+1}^2 e^{-2 \tilde\kappa_{N+1}} \tilde Q_{N+1}
 \left[\left(\tilde C_{N+1} L_{N+1}^\dagger L_{N+1}  \tilde C_{N+1}\right)^+ + O\left(q_{18}(a,\varepsilon,x)\right)
\ds \right] \tilde Q_{N+1}.
 \end{equation}
Finally, we use \eqref{8.44}, \eqref{8.45}, \eqref{8.55}, \eqref{8.62}, \eqref{8.63}, and we proceed as in the proof of 
Theorem~\ref{theorem6.3} and complete the proofs of (b) and (c).
\end{proof}

The next theorem is the counterpart of Theorem~\ref{theorem6.4}. When a new bound state is added to the spectrum, it describes
how the Jost matrix, the scattering matrix, and the Jost solution transform and
it also shows that the existing bound-state energies and their multiplicities remain unchanged.
As in Theorem~\ref{theorem6.4} we use the stronger assumption that the unperturbed potential $V$ belongs to
$L^1_2(\mathbb R^+)$ rather than $L^1_1(\mathbb R^+)$ so that the perturbed potential
$\tilde V$ belongs to $L^1_1(\mathbb R^+),$ which is assured by Theorem~\ref{theorem8.5}(b).
We remark that this is a sufficiency assumption because the asymptotic estimates used
in Theorem~\ref{theorem8.5}(b)
on the potentials
are not necessarily sharp.

\begin{theorem}
\label{theorem8.6} Consider
the unperturbed Schr\"odinger operator with the potential
$V$ satisfying \eqref{2.2} and belonging to $L^1_2(\mathbb R^+),$  with the selfadjoint
boundary condition \eqref{2.5} described by the boundary matrices $A$ and $B$ satisfying
\eqref{2.6} and \eqref{2.7},
the Jost solution $f(k,x)$ satisfying \eqref{2.9}, the Jost matrix $J(k)$ defined as in \eqref{2.11}, 
the scattering matrix $S(k)$ defined in \eqref{2.12}, and with 
$N$ bound states with the energies $-\kappa_j^2$ and
the Gel'fand--Levitan normalization matrices $C_j$ 
for $1\le j\le N.$ 
Let us use a tilde to identify the quantities
associated with the perturbed Schr\"odinger operator
so that we have
the perturbed potential
$\tilde V(x)$ expressed as in \eqref{8.22},
the boundary condition \eqref{2.5} where
the boundary matrices $A$ and $B$ are replaced with
$\tilde A$ and $\tilde B,$ respectively, given in \eqref{8.43},
the perturbed Jost solution $\tilde f(k,x)$ satisfying the asymptotics \eqref{2.9},
the perturbed regular solution $\tilde\varphi(k,x)$ satisfying the initial
conditions \eqref{5.17}, 
the perturbed Jost matrix $\tilde J(k)$ defined as in \eqref{6.89}, and
the perturbed scattering matrix $\tilde S(k)$ defined in \eqref{6.90}.
We then have the following:

\begin{enumerate}

\item[\text{\rm(a)}] The unperturbed Jost matrix $J(k)$ is transformed into the perturbed Jost matrix $\tilde J(k)$ as
 \begin{equation}
\label{8.64}
\tilde J(k)=\left[I-\ds\frac{2i\tilde\kappa_{N+1}}{k+i\tilde\kappa_{N+1}}\,\tilde P_{N+1}\right] J(k),\qquad
k\in\overline{\mathbb C^+},
\end{equation}
where we let
\begin{equation}\label{8.65}
\tilde P_{N+1}:= L_{N+1}\, \tilde C_{N+1}\, [ \tilde C_{N+1}\, L_{N+1}^\dagger\, L_{N+1}\, \tilde C_{N+1}]^+ \tilde C_{N+1}\, L_{N+1}^\dagger,
\end{equation}
with $L_{N+1}$ being the invertible matrix defined in \eqref{8.29}. Moreover, $\tilde P_{N+1}$ is the orthogonal projection onto the
 kernel of $\tilde J(i\tilde\kappa_{N+1})^\dagger$ and it has rank $\tilde m_{N+1}.$ We recall
that $\tilde m_{N+1}$ is the rank of  $\tilde C_{N+1}.$

\item[\text{\rm(b)}] The matrix product $J(k)^\dagger J(k)$ for $k\in\mathbb R$ does not change under
the perturbation,
i.e. we have
 \begin{equation}
\label{8.66}
\tilde J(k)^\dagger \tilde J(k)=J(k)^\dagger J(k),\qquad k\in\mathbb R.\end{equation}
Hence, the continuous part of the spectral measure $d\rho$ does not change
under the perturbation.

\item[\text{\rm(c)}] Under the perturbation, the determinant of the Jost matrix $J(k)$ is transformed as
 \begin{equation}
\label{8.67}
\det[\tilde J(k)]=\left(\ds\frac{k-i\tilde\kappa_{N+1}}{k+i\tilde\kappa_{N+1}}\right)^{\tilde m_{N+1}} \det[J(k)],\qquad
k\in\overline{\mathbb C^+}.\end{equation}

\item[\text{\rm(d)}] The perturbation adds
a new bound state with the energy $- \tilde\kappa_{N+1}^2$ and multiplicity $\tilde m_{N+1}$ in such a way that
the remaining bound states with the energies $-\kappa_j^2$ and multiplicities $m_j$ for $1\le j\le N$
are unchanged.

\item[\text{\rm(e)}] Under the perturbation, the scattering matrix $S(k)$ undergoes the transformation
 \begin{equation*}
\tilde S(k)=\left[I+\ds\frac{2i\tilde\kappa_{N+1}}{k-i\tilde\kappa_{N+1}}\,\tilde P_{N+1}\right] S(k)\left[I+\ds\frac{2i\tilde\kappa_N}{k-i\tilde\kappa_{N+1}}\,\tilde P_{N+1}\right],
\qquad k\in\mathbb R,
\end{equation*}
where $\tilde P_{N+1}$ is the orthogonal projection in \eqref{8.65}.

\item[\text{\rm(f)}] The perturbation changes the determinant of the scattering matrix as 
 \begin{equation*}
\det[\tilde S(k)]=\left(\ds\frac{k +i\tilde\kappa_{N+1}}{k-i\tilde\kappa_{N+1}}\right)^{2 \tilde m_{N+1}} \det[S(k)],
\qquad k\in\mathbb R.
\end{equation*}

\item[\text{\rm(g)}] Under the perturbation the Jost solution $f(k,x)$ is transformed into $\tilde f(k,x)$ as
\begin{equation}
\label{8.70}
\tilde f(k,x)=\left[f(k,x)-\ds\frac{1}{k^2+\tilde\kappa_{N+1}^2}
\,\xi_{N+1}(x) \,\Omega_{N+1}(x)^+
q_{23}(k,x)
\right]
\left[I- \frac{2i \tilde\kappa_{N+1}}{k+i \tilde\kappa_{N+1} } \tilde P_{N+1}\right],
\end{equation}
where we have defined $q_{23}(k,x)$ as
\begin{equation}\label{8.71}
q_{23}(k,x):=\left[\xi'_{N+1}(x)^\dagger\,f(k,x)-\xi_{N+1}(x)^\dagger
\,f'(k,x)\right],\end{equation}
which is the analog of $q_7(k,x)$ defined in \eqref{6.97}.

\end{enumerate}
\end{theorem}

\begin{proof} We establish \eqref{8.64} by proceeding as in the proof of Theorem~\ref{theorem6.4}(a), and for this
we use \eqref{2.9} for $\tilde f(k,x),$ the analog of \eqref{2.15} for $\tilde\varphi(k,x),$
the analog of \eqref{6.101}, as well
as \eqref{8.42}, \eqref{8.44}, \eqref{8.45}, and \eqref{8.62}.
To prove that $\tilde P_{N+1}$ is an orthogonal projection, we need to show that
$\tilde P_{N+1}$ is selfadjoint and it satisfies
\begin{equation}
\label{8.72}
\tilde P_{N+1}^2=\tilde P_{N+1}.
\end{equation}
Using \eqref{3.3}
and the fact that $\tilde C_{N+1}$ is selfadjoint, we prove that
$\tilde P_{N+1}$ given in \eqref{8.65} is selfadjoint. We establish \eqref{8.72}
by proceeding as follows. Using \eqref{8.60} we obtain
\begin{equation}\label{8.73}
\left[ \tilde C_{N+1}\, L_{N+1}^\dagger\, L_{N+1} \,\tilde C_{N+1}\right]^+ \left[ \tilde C_{N+1}\, L_{N+1}^\dagger\, L_{N+1}\, \tilde C_{N+1}\right]=
I\oplus 0= \tilde Q_{N+1}.
\end{equation}
Then, using \eqref{8.1}, \eqref{8.65}, and \eqref{8.73}, we confirm that \eqref{8.72} holds.
Since $L_{N+1}$ and $\tilde{\mathbf H}_{N+1}^{-1/2}$ are invertible matrices, 
using \eqref{8.60} and\eqref{8.65} we prove that the rank of $\tilde P_{N+1}$ is $\tilde m_{N+1}.$ 
To show that $\tilde P_{N+1}$ is the orthogonal projection onto the kernel of $\tilde J(i\tilde\kappa_{N+1})^\dagger,$
we proceed as follows.  Using
\eqref{8.64} and the fact that
$\tilde J(i\tilde\kappa_{N+1})$ is invertible, we obtain
\begin{equation*}
\text{\rm{Ker}}[\tilde J(i\tilde\kappa_{N+1})^\dagger]= \text{\rm{Ker}}[I-\tilde P_{N+1}],
\end{equation*}
which implies that the kernel of $\tilde J(i\tilde\kappa_{N+1})^\dagger$ is precisely the range of $ \tilde P_{N+1}.$ 
Hence, $\tilde P_{N+1}$ is the orthogonal projection onto the kernel of $\tilde J(i\tilde\kappa_{N+1})^\dagger.$
Thus, the proof of (a) is complete.
By using \eqref{8.64} and by proceeding as in the proof of 
Theorem~\ref{theorem6.4}, we establish the stated results in
(b), (c), (e), and (f). We note that (d) follows from \eqref{8.67} and Theorems~3.11.1 and 3.11.6 of \cite{AW2021}.
Finally, we prove (g) as in the proof of Theorem~\ref{theorem6.4}(g).
\end{proof}

The following remark is the counterpart of Remark~\ref{remark6.5}. It stresses the fact that the perturbed Jost solution $\tilde f(k,x)$ appearing in
\eqref{8.70} does not have a singularity at $k=i\tilde\kappa_{N+1}$ despite the factor $k^2+\tilde\kappa_{N+1}^2$
in a denominator on the right-hand side of \eqref{8.70}.

\begin{remark}\label{remark8.7}{\rm
In this remark we show that the right-hand side of \eqref{8.70} has a removable singularity at $k=i\tilde\kappa_{N+1},$
and hence the factor $k^2+\tilde\kappa_{N+1}^2$
in a denominator there does not cause a singularity for the perturbed Jost solution $\tilde f(k,x).$
For the proof we proceed as follows.
From \eqref{8.71} we get
\begin{equation}\label{8.75}
q_{23}(i\tilde\kappa_{N+1},x)=\xi'_{N+1}(x)^\dagger\,f(i\tilde\kappa_{N+1},x)-\xi_{N+1}(x)^\dagger
\,f'(i\tilde\kappa_{N+1},x).
\end{equation}
Since $\xi_{N+1}(x)$ and $f(i\tilde\kappa_{N+1},x)$ are matrix-valued solutions to \eqref{2.1} with $ k= i\tilde\kappa_{N+1},$ 
the Wronskian appearing on the right-hand side in \eqref{8.75} is independent of $x.$
Thus, the constant value of $q_{23}(i\tilde\kappa_{N+1},x)$ can be evaluated by letting $x\to+\infty$
in \eqref{8.75}.
Using  \eqref{2.9}, \eqref{8.44}, and \eqref{8.45} on the right-hand side of \eqref{8.75}, we obtain
\begin{equation}\label{8.76}
q_{23}(i\tilde\kappa_{N+1},x)= 2 \,\tilde\kappa_{N+1}\,  \tilde C_{N+1} \,L_{N+1}^\dagger.
\end{equation}
We can isolate the part of the term containing $k^2+\tilde\kappa_{N+1}^2$ on the right-hand side of \eqref{8.70} by defining
$q_{24}(k,x)$ as
\begin{equation}\label{8.77}
q_{24}(k,x):=\ds\frac{q_{23}(k,x)}{k^2+\tilde\kappa_{N+1}^2}
\left[I- \frac{2i \tilde\kappa_{N+1}}{k+i \tilde\kappa_{N+1} } \tilde P_{N+1}\right].
\end{equation}
Comparing \eqref{8.70} and \eqref{8.77} we see that $\tilde f(k,x)$ does not have a singularity
at $k=i\tilde\kappa_{N+1}$ provided the singularity of $q_{24}(k,x)$ at $k=i\tilde\kappa_{N+1}$ is removable.
We can write the right-hand side of \eqref{8.77} as a sum of two terms by letting
  \begin{equation*}
q_{24}(k,x)=q_{25}(k,x)+q_{26}(k,x),
\end{equation*}
where we have defined
\begin{equation}\label{8.79}
q_{25}(k,x):=\ds\frac{q_{23}(k,x)-q_{23}(i\tilde\kappa_{N+1},x)}{k^2+\tilde\kappa_{N+1}^2}
\,
\left[I- \frac{2i \tilde\kappa_{N+1}}{k+i \tilde\kappa_{N+1} } \tilde P_{N+1}\right],
\end{equation}
\begin{equation}\label{8.80}
q_{26}(k,x):=\ds\frac{q_{23}(i\tilde\kappa_{N+1},x)}{k^2+\tilde\kappa_{N+1}^2}
\,
\left[I- \frac{2i \tilde\kappa_{N+1}}{k+i \tilde\kappa_{N+1} } \tilde P_{N+1}\right].
\end{equation}
Since $f(k,x)$ and $f'(k,x)$ are analytic in $k$ at $k=i\tilde\kappa_{N+1},$ the quantity
$q_{23}(k,x)-q_{23}(i\tilde\kappa_{N+1},x)$ has the behavior
$O(k-i\tilde\kappa_{N+1})$ as $k\to i\tilde\kappa_{N+1},$ and hence the right-hand side of \eqref{8.79} has a removable
singularity at $k=i\tilde\kappa_{N+1}.$
On the other hand, using \eqref{8.76} in \eqref{8.80} we get
\begin{equation}\label{8.81}
q_{26}(k,x)=\ds\frac{2 \,\tilde\kappa_{N+1}\,  \tilde C_{N+1} \,L_{N+1}^\dagger}{k^2+\tilde\kappa_{N+1}^2}
\,
\left[I- \frac{2i \tilde\kappa_{N+1}}{k+i \tilde\kappa_{N+1} } \tilde P_{N+1}\right].
\end{equation}
From \eqref{8.1}, \eqref{8.65}, and \eqref{8.73}, we conclude that
\begin{equation}\label{8.82}
\tilde C_{N+1} L_{N+1}^\dagger \tilde P_{N+1}= \tilde C_{N+1} L_{N+1}^\dagger.
\end{equation}
Using \eqref{8.82} on the right-hand side of \eqref{8.80}, we obtain
\begin{equation*}
q_{26}(k,x)= \frac{2 \tilde\kappa_{N+1}\, \tilde C_{N+1}\, L_{N+1}^\dagger}{(k +i\tilde\kappa_{N+1})^2},
\end{equation*}
which indicates that $q_{26}(k,x)$ is analytic at $k=i\tilde\kappa_{N+1}.$
Thus, $q_{24}(k,x)$ has only a removable singularity, confirming that 
$\tilde f(k,x)$ does not have a singularity at $k=i\tilde\kappa_{N+1}.$
}
\end{remark}

The next theorem is the counterpart of Theorem~\ref{theorem6.6}. When we add a new bound state to the
Schr\"odinger operator, it shows that the existing  orthogonal projection matrices and Gel'fand--Levitan normalization matrices
remain unchanged. It also shows how the spectral measure 
changes when the bound state is added.
As in Theorem~\ref{theorem8.6} we use the stronger assumption that the unperturbed potential $V$ belongs to
$L^1_2(\mathbb R^+)$ rather than $L^1_1(\mathbb R^+)$ so that the perturbed potential
$\tilde V$ belongs to $L^1_1(\mathbb R^+),$ which is assured by Theorem~\ref{theorem8.5}(b).
We recall that that this is a sufficiency assumption  because the asymptotic estimates
in Theorem~\ref{theorem8.5}(b)
on the potentials
are not necessarily sharp.

\begin{theorem}
\label{theorem8.8} Consider
the unperturbed Schr\"odinger operator with the potential
$V$ satisfying \eqref{2.2} and belonging to $L^1_2(\mathbb R^+),$  
 with the selfadjoint
boundary condition \eqref{2.5} described by the boundary matrices $A$ and $B$ satisfying
\eqref{2.6} and \eqref{2.7}, the regular solution $\varphi(k,x)$ satisfying the initial conditions
\eqref{2.10}, 
the Jost matrix $J(k)$ defined as in \eqref{2.11},
the spectral measure $d\rho$ as in \eqref{5.1}, and
containing
$N$ bound states with the energies $-\kappa_j^2,$ 
the Gel'fand--Levitan normalization matrices $C_j,$ and the orthogonal projections
$Q_j$ onto $\text{\rm{Ker}}[J(i\kappa_j)]$ 
for $1\le j\le N.$
Let us use a tilde to identify the quantities
associated with the perturbed Schr\"odinger operator
so that
$\tilde\varphi(k,x)$ is the perturbed regular solution satisfying the initial
conditions \eqref{5.17}, $\tilde J(k)$ is the perturbed Jost matrix
defined as in \eqref{6.89}, 
$\tilde Q_j$ is the orthogonal projection onto $\text{\rm{Ker}}[\tilde J(i\tilde\kappa_j)],$
$\tilde C_j$ is the Gel'fand--Levitan normalization matrix defined
as in \eqref{3.33}, \eqref{3.34}, and \eqref{3.36} but
by using $\tilde Q_j$ instead of $Q_j$ and
by using $\tilde\varphi(i\tilde\kappa_j,x)$ instead of $\varphi(i\kappa_j,x)$ for $1\le j\le N+1,$ and $d\tilde\rho$ is the
perturbed spectral measure defined
as
\begin{equation}
\label{8.84}
d\tilde\rho=\begin{cases}
\ds\frac{\sqrt{\lambda}}{\pi}\,\left(\tilde J(k)^\dagger\,\tilde J(k)\right)^{-1}\,d\lambda,\qquad \lambda\ge 0,
\\
\noalign{\medskip}
\ds\sum_{j=1}^{N+1} \tilde C_j^2\,\delta(\lambda-\tilde\lambda_j)\,d\lambda,
\qquad \lambda<0,\end{cases}
\end{equation}
where $\tilde\lambda_j:= -\tilde\kappa_j^2$ with $\tilde\lambda_j$ for $1\le j\le N+1$ denoting the eigenvalues of the perturbed operator.
It is understood that $\tilde\kappa_j=\kappa_j$ for $1\le j\le N.$
We then have the following:

\begin{enumerate}

\item[\text{\rm(a)}] Under the perturbation,
the projection matrices $Q_j$ for $1\le j\le N $ remain unchanged, i.e. we have
 \begin{equation*}
\tilde Q_j=Q_j,\qquad 1\le j\le N.
\end{equation*}

\item[\text{\rm(b)}] Under the perturbation, the Gel'fand--Levitan normalization matrices  for $1\le j\le N$ remain
 unchanged, i.e. we have
 \begin{equation*}
\tilde C_j=C_j,\qquad 1\le j\le N.\end{equation*}

\item[\text{\rm(c)}] 
The orthogonal projection onto the kernel of $\tilde J(i\tilde\kappa_{N+1})$ is the matrix $\tilde Q_{N+1}$ appearing in \eqref{8.1}.
The matrix $\tilde C_{N+1}$ appearing in
\eqref{8.1} is indeed
the Gel'fand--Levitan normalization matrix for the bound-state with the energy $-\tilde\kappa_{N+1}^2$ for
the perturbed Schr\"odinger operator.

\item[\text{\rm(d)}] The perturbed spectral measure $d\tilde\rho$ is related to the unperturbed spectral measure $d\rho$ as
\begin{equation*}
d\tilde\rho=d\rho+\tilde C_{N+1}^2\,\delta(\lambda-\tilde\lambda_{N+1})\,d\lambda.
\end{equation*}
\end{enumerate}

\end{theorem}

\begin{proof}  The proof of (a) is similar to the proof of Theorem~\ref{theorem6.6}(a).
 The proof of (b) is obtained by proceeding as in the proof of Theorem~\ref{theorem6.6}(b)
and by using 
\eqref{6.50} and \eqref{6.51}, the counterparts for
\eqref{6.50} and \eqref{6.51} for the adjoint quantities,
the analogs of those four asymptotics but by replacing $N$ in them with
$j$ for $1\le j\le N+1,$ and by also using 
\eqref{8.44}, \eqref{8.45}, and \eqref{8.62}.
For the proof of (c) we proceed as follows.
From Theorem~\ref{theorem8.6}(d) we already know that the multiplicity of
the added bound state at $k=i\tilde\kappa_{N+1}$ is equal to $\tilde m_{N+1}.$
From \eqref{8.1}, \eqref{8.16}, \eqref{8.17}, \eqref{8.21}, and \eqref{8.36}, we obtain
\begin{equation}\label{8.88}
\tilde\varphi(i\tilde\kappa_{N+1},x) \,\tilde C_{N+1} = \xi_{N+1}(x)\, \Omega(x)^+_{N+1}.
\end{equation}
Using \eqref{8.17}, \eqref{8.23}, \eqref{8.88}, and the hermitian property of $\Omega_{N+1}(x)^+,$ we get
\begin{equation}\label{8.89}
\int_0^x dy\,[\tilde\varphi(i\tilde\kappa_{N+1},y)\, \tilde C_{N+1}]^\dagger \,\tilde\varphi(i\tilde\kappa_{N+1},y)\, \tilde C_{N+1}= \tilde Q_{N+1}- \Omega_{N+1}(x)^+.
\end{equation}
Using \eqref{8.62} in \eqref{8.89} we have
\begin{equation}\label{8.90}
\int_0^\infty dx\,[\tilde\varphi(i\tilde\kappa_{N+1},x)\, \tilde C_{N+1}]^\dagger\, \tilde\varphi(i\tilde\kappa_{N+1},x) \,\tilde C_{N+1}= \tilde Q_{N+1}.
\end{equation}
On the other hand, using \eqref{8.2}, \eqref{8.8}, and \eqref{8.90}, we obtain
\begin{equation}\label{8.91}
\int_0^\infty dx \,[\tilde\varphi(i\tilde\kappa_{N+1},x) \,\tilde Q_{N+1}]^\dagger\, \tilde\varphi(i\tilde\kappa_{N+1},x)\, \tilde Q_{N+1}
 = \tilde{\mathbf H}_{N+1} \,\tilde Q_{N+1}.
\end{equation}
From \eqref{8.91} we conclude that each column of $\tilde\varphi(i\tilde\kappa_{N+1},x)\, \tilde Q_{N+1}$ is square integrable, and hence  
$\tilde\varphi(i\tilde\kappa_{N+1},x) \,\tilde Q_{N+1}$ is a matrix-valued
bound-state solution to the perturbed Schr\"odinger equation
with the energy $-\tilde\kappa_{N+1}^2.$ Since the rank of $\tilde Q_{N+1}$ is $\tilde m_{N+1},$ any basis
for $\tilde Q_{N+1}\,\mathbb C^n$ contains exactly $\tilde m_{N+1}$ linearly independent vectors.
 Let $\{\tilde w^{(l)}_{N+1}\}_{l=1}^{\tilde m_{N+1}}$ be an orthonormal basis for $\tilde Q_{N+1}\,\mathbb C^n.$  Using
the analog of the second equality of \eqref{3.14}, we get
\begin{equation*}
\tilde Q_{N+1}=\ds\sum_{l=1}^{\tilde m_{N+1}} \tilde w^{(l)}_{N+1}\,(\tilde w^{(l)}_{N+1})^\dagger,
\end{equation*}
which yields
\begin{equation}\label{8.93}
\tilde\varphi(i\tilde\kappa_{N+1},x) \,\tilde Q_{N+1}= \ds\sum_{l=1}^{\tilde m_{N+1}}
[\tilde\varphi(i\tilde\kappa_{N+1},x) \,  \tilde w^{(l)}_{N+1}]\, (\tilde w^{(l)}_{N+1})^\dagger.
\end{equation}
From \eqref{8.93} it follows that each column vector $\tilde\varphi(i\tilde\kappa_{N+1},x)\, \tilde w^{(l)}_{N+1}$ 
for $1\le l\le \tilde m_{N+1}$ is a bound state for the
perturbed Schr\"odinger operator
with the eigenvalue $-\tilde\kappa_{N+1}^2.$ Then, from Theorem~3.11.1(e) of \cite{AW2021} we conclude that
each column vector $\tilde w^{(l)}_{N+1}$ for
$1\le l\le \tilde m_{N+1}$ belongs to the kernel of $\tilde J(i\tilde\kappa_{N+1}).$ 
Recall that from Theorem~\ref{theorem8.6}(d) we already know that the multiplicity of
the added bound state at $k=i\tilde\kappa_{N+1}$ is equal to $\tilde m_{N+1}.$ Furthermore, by Theorem~3.11.1(d) of \cite{AW2021} 
we know that the dimension of the kernel of $\tilde J(i\tilde\kappa_{N+1})$ is precisely $\tilde m_{N+1}.$
This shows that the map $\tilde Q_{N+1}$ is an orthogonal projection onto the kernel of $\tilde J(i\tilde\kappa_{N+1}).$
Moreover, from \eqref{3.33}, \eqref{3.34}, \eqref{3.36}, and \eqref{8.91} it follows that $\tilde C_{N+1}$ is the Gel'fand--Levitan normalization matrix for the bound state with the energy 
$-\tilde\kappa_{N+1}^2$ for the perturbed Schr\"odinger operator. This completes the proof of (c).
We obtain the proof of (d) by using (b) and (c) with the help of \eqref{8.66}, Theorem~\ref{theorem8.6}(d), and \eqref{8.84}.
\end{proof}

In (IV.1.30) of \cite{CS1989} it is claimed that, when adding a bound state of energy $-\kappa^2$ to the spectrum of
the scalar-valued Schr\"odinger operator with the Dirichlet boundary condition, the potential increment 
$\tilde V(x)- V(x)$ decays as $F e^{-2\kappa x}$ as $x\to+\infty,$ where $F$ is a nonzero constant. This is incorrect, as we demonstrate with an explicitly solved
example below.

\begin{example}
\label{example8.9} 
\normalfont
Consider the scalar case, i.e. when $n=1,$ with the unperturbed potential
\begin{equation}
\label{8.94}
V (x)=- \ds\frac{8 \,e^{2 x}}
{\left(1+e^{2  x}\right)^2},\qquad x\in\mathbb R^+,
 \end{equation}
 and the boundary matrices $A$ and $B,$ which are scalars in this case, are chosen as
\begin{equation}
\label{8.95}
A=0,\quad B=-1.
 \end{equation}
Using \eqref{8.95} in \eqref{2.5}, we observe that the
corresponding boundary condition is the Dirichlet boundary condition $\psi(0)=0.$
Using \eqref{8.94} and \eqref{8.95} in \eqref{2.1}, \eqref{2.9}, \eqref{2.11}, and \eqref{2.12}, we obtain
the corresponding Jost solution $f(k,x),$ regular solution $\varphi(k,x),$ Jost matrix $J(k),$ and the scattering matrix
 $S(k)$ as
 \begin{equation*}
f(k,x)= e^{ikx}\left[1-\ds\frac{2\,i }{(k+i)
\left(1+e^{2  x}\right)}\right],
 \end{equation*}
  \begin{equation*}
\varphi(k,x)=-\ds\frac{k\,\sin(kx)+\cos(kx)\, \tanh x}
{k^2+1},
 \end{equation*}
  \begin{equation}
\label{8.98}
J(k)=-\ds\frac{k}{k+i},
 \end{equation}
 \begin{equation*}
S(k)=-\ds\frac{k+i}{k-i},
 \end{equation*}
  which are all scalar valued.
Since $J(k)$ in \eqref{8.98} does not have any zeros on the positive
imaginary axis in the complex $k$-plane, we see that the unperturbed Schr\"odinger operator does not have any bound states.
Using the method of Section~\ref{section8}, we
add a bound state to the unperturbed Schr\"odinger operator
at $k=i\kappa_1$ with the Gel'fand--Levitan normalization constant $C_1,$ where we choose
 \begin{equation*}
\kappa_1=1,\quad 
C_1=4.
 \end{equation*}
The corresponding perturbed Jost matrix $\tilde J(k),$ perturbed boundary matrices $\tilde A$ and $\tilde B,$
and perturbed potential $\tilde V(x)$
 are obtained as
 \begin{equation*}
\tilde J(k)=-\ds\frac{k(k-i)}{(k+i)^2},
 \end{equation*}
 \begin{equation*}
\tilde A=0,\quad \tilde B=-1,
 \end{equation*}
 \begin{equation}
\label{8.103}
\tilde V(x)=\ds\frac{q_{27}(x)+q_{28}(x)}{\left[ (1-2x)\,\cosh x+\left(1+4 x^2+\cosh(2x)\right)\sinh x\right]^2},
\end{equation}
where we have defined
 \begin{equation*}
q_{27}(x):=7-24 x+32 x^4+64 x^2 \cosh(2x)-(16+32x)\,\sinh(2x),
\end{equation*}
 \begin{equation*}
q_{28}(x):=-(9+8x^2)\,\cosh(4x)+(-2+20x)\,\sinh(4x).
\end{equation*}
We observe that the unperturbed potential $V(x)$ in \eqref{8.94} behaves
as $O(e^{-2 x})$ and the perturbed potential $\tilde V(x)$ in \eqref{8.103}
behaves as  $O(x^2 e^{-2 x})$ as $x\to+\infty,$ and hence the potential difference
$\tilde V(x)-V(x)$ also behaves as
 $O(x^2 e^{-2 x})$ as $x\to+\infty.$ This is in contradiction to
 the result stated in (IV.1.30) of \cite{CS1989}, which incorrectly predicts
 $\tilde V(x)-V(x)=O(e^{-2x})$ as $x\to+\infty.$ We remark that the decay rate for $\tilde V(x)-V(x)$
 as $x\to+\infty$ in this example is consistent with our general estimate
 \eqref{8.40}. 
\end{example}

\section{The transformation to increase the multiplicity of a bound state}
\label{section9}

This section is complementary to Section~\ref{section8}.
In Section~\ref{section8} we have started with the unperturbed Schr\"odinger operator and added a new bound
state to the discrete spectrum without changing the continuous spectrum. We have developed the transformation formulas
for all relevant quantities under the aforementioned perturbation.
In this complementary section, we still do not change the continuous spectrum, we do not add any new bound states or remove any
existing bound states, but we only increase the multiplicity of one of the existing bound states. Under such a perturbation, we
show how all relevant quantities transform. We omit the proofs as those proofs are similar to the
proofs of the results in Section~\ref{section8}.

The changing of the multiplicity of a bound states does not apply in the scalar case, i.e. when $n=1,$ because 
in the scalar case all bound states have multiplicity one. Hence, we assume that
we are in the matrix case with $n\ge 2$ and that there exists at least one bound states, i.e. we have $N\ge 1.$
As in the previous sections, we use a tilde to identify the perturbed
quantities obtained after the multiplicity of the bound state is increased. We assume that our
unperturbed potential $V$
satisfies \eqref{2.2} and belongs
to $L^1_1(\mathbb R^+).$ 
In our notation, our unperturbed regular solution is
$\varphi(k,x),$ unperturbed Jost solution is $f(k,x),$ 
unperturbed Jost matrix is $J(k),$ unperturbed scattering matrix is $S(k),$ unperturbed 
matrix $Q_j$ corresponds to the
orthogonal projection onto the kernel of
$J(i\kappa_j)$ for $1\le j\le N,$
and 
unperturbed
boundary matrices are given by $A$ and $B.$
Our perturbed potential is $\tilde V(x),$ perturbed regular solution is
$\tilde\varphi(k,x),$ perturbed Jost solution is $\tilde f(k,x),$ 
perturbed Jost matrix is $\tilde J(k),$ perturbed scattering matrix is $\tilde S(k),$ 
and perturbed
boundary matrices are given by $\tilde A$ and $\tilde B.$

Without loss of generality, we increase the multiplicity of the bound state at $k=i\kappa_N$ from $m_N$ to $\tilde m_N$ by
the positive integer $m_{N\text{\rm{i}}},$ which is defined as
\begin{equation}
\label{9.1}
m_{N\text{\rm{i}}}:=\tilde m_N-m_N.
\end{equation}
Thus, $m_{N\text{\rm{i}}}$ satisfies the restriction
$1\le m_{N\text{\rm{i}}}\le n-m_N.$
By using the subscript $N{\text{\rm{i}}}$ in \eqref{9.1}, we indicate that we refer to the $N$th bound state at $k=i\kappa_N$ and that we 
increase its multiplicity.
We recall that $Q_N$ denotes the orthogonal projection onto
the kernel of $J(i\kappa_N),$ the positive integer $m_N$ corresponds to the dimension of
$\text{\rm{Ker}}[J(i\kappa_N)],$ and we use $C_N$ to denote the Gel'fand--Levitan normalization matrix
associated with the bound state at $k=i\kappa_N$ in the unperturbed case.
We introduce the $n\times n$ matrix $\tilde Q_{N\text{\rm{i}}}$
with the rank $m_{N\text{\rm{i}}}$  in order to denote the orthogonal projection 
onto a subspace of the orthogonal complement in 
$\mathbb C^n$ of $Q_N \,\mathbb C^n.$ Thus, we have
\begin{equation}
\label{9.2}
\tilde Q_{N\text{\rm{i}}}\, \mathbb C^n \subset \left(Q_N \,\mathbb C^n\right)^\perp.
\end{equation}

We use an $n\times n$ nonnegative hermitian matrix $\tilde{\mathbf G}_{N\text{\rm{i}}}$ so that we have the equalities
\begin{equation*}
\tilde{\mathbf G}_{N\text{\rm{i}}}\,\tilde Q_{N\text{\rm{i}}}=\tilde Q_{N\text{\rm{i}}}\,\tilde{\mathbf G}_{N\text{\rm{i}}} = \tilde{\mathbf G}_{N\text{\rm{i}}},
\end{equation*}
which is the analog of \eqref{8.3}.
We further assume that the restriction of  $\tilde{\mathbf G}_{N\text{\rm{i}}}$ to the subspace $\tilde Q_{N\text{\rm{i}}}\, \mathbb C^n$ is invertible.
 We define the $n\times n$ matrix $\tilde{\mathbf H}_{N\text{\rm{i}}}$ as
\begin{equation*}
\tilde{\mathbf H}_{N\text{\rm{i}}}:= I- \tilde Q_{N\text{\rm{i}}}+\tilde{\mathbf G}_{N\text{\rm{i}}},
\end{equation*}
which is analogous to \eqref{8.4}. By proceeding as in Section~\ref{section8}, 
we prove that $\tilde{\mathbf H}_{N\text{\rm{i}}}$ is a nonnegative hermitian matrix and that it
satisfies
\begin{equation*}
\tilde{\mathbf H}_{N\text{\rm{i}}}\, \tilde Q_{N\text{\rm{i}}}= \tilde Q_{N\text{\rm{i}}}\,\tilde{\mathbf H}_{N\text{\rm{i}}},
\end{equation*}
which is analogous to \eqref{8.8}.
By proceeding as in \eqref{8.6}, we
prove that the matrix $\tilde{\mathbf H}_{N\text{\rm{i}}}$ is invertible, and hence  it is a positive matrix.
Thus, the positive matrix $\tilde{\mathbf H}_{N\text{\rm{i}}}^{1/2}$ and its inverse
$\tilde{\mathbf H}_{N\text{\rm{i}}}^{-1/2}$ are well defined.
Analogous to \eqref{8.2}, we define the $n\times n$ matrix $\tilde C_{N\text{\rm{i}}}$ as
\begin{equation}\label{9.6}
\tilde C_{N\text{\rm{i}}}:=\tilde{\mathbf H}_{N\text{\rm{i}}}^{-1/2}\, \tilde Q_{N\text{\rm{i}}}.
\end{equation}
It follows that $\tilde C_{N\text{\rm{i}}}$ is hermitian and nonnegative, its rank is equal to $m_{N\text{\rm{i}}},$ and
it satisfies the equalities
\begin{equation*}
\tilde C_{N\text{\rm{i}}}\,\tilde Q_{N\text{\rm{i}}}= \tilde Q_{N\text{\rm{i}}} \,\tilde C_{N\text{\rm{i}}}=  \tilde C_{N\text{\rm{i}}}.
\end{equation*}

Proceeding as in the proof of
Proposition~\ref{proposition8.4}(d) and using the analog of \eqref{8.29},
we can express the matrix solution $\varphi(i\kappa_N, x) \,\tilde Q_{N\text{\rm{i}}}$ to the Schr\"odinger  equation 
\eqref{2.1} with $k= i\kappa_N$ as
\begin{equation}\label{9.8}
\varphi(i\kappa_N,x) \,\tilde Q_{N\text{\rm{i}}}= f(i\kappa_N,x)\,K_{N\text{\rm{i}}} \,
\tilde Q_{N\text{\rm{i}}} +  g(i\kappa_N,x)\, L_{N\text{\rm{i}}}\,\tilde Q_{N\text{\rm{i}}},
\end{equation}
for some $n \times n$ matrices $K_{N\text{\rm{i}}}$ and $L_{N\text{\rm{i}}}.$ We claim that the restriction of  
$L_{N\text{\rm{i}}}\,\tilde Q_{N\text{\rm{i}}}$ to 
the subspace $\tilde Q_{N\text{\rm{i}}} \,\mathbb C^n$ is invertible. 
Otherwise, we would have $L_{N\text{\rm{i}}}\,\tilde Q_{N\text{\rm{i}}}\, v=0$ 
for some column vector $v\in \tilde Q_{N\text{\rm{i}}} \, \mathbb C^n.$
Then, \eqref{9.8} would imply that
\begin{equation}
\label{9.9}
\varphi(i\kappa_N,x)\, \tilde Q_{N\text{\rm{i}}}\,v= f(i\kappa_N,x)\,K_{N\text{\rm{i}}}\,v.
\end{equation}
Using \eqref{6.48} on the right-hand side of \eqref{9.9}, we would get
\begin{equation}
\label{9.10}
\varphi(i\kappa_N,x) \,\tilde Q_{N\text{\rm{i}}}\,v=
e^{-\kappa_N x} \left[I+O\left(\int_{x}^\infty dy\,|V(y)| \right)\right] Q_{N\text{\rm{i}}}\,v,\qquad x\to+\infty.
\end{equation}
Thus, from \eqref{9.10} we would conclude that the column vector
$\varphi(i\kappa_N,x) \,\tilde Q_{N\text{\rm{i}}}\,v$
would be a square-integrable solution to \eqref{2.1} and satisfying the boundary condition \eqref{2.5}, and that would 
imply that $\varphi(i\kappa_N,x) \,\tilde Q_{N\text{\rm{i}}}\,v$ would be a bound-state solution
to \eqref{2.1} at $k=i\kappa_N.$
However, \eqref{9.2} would imply that $v\not\in Q_N \mathbb C^n,$
which would contradict Theorem~3.11.1(e) of \cite{AW2021}. This
completes the proof that the restriction of  $L_{N\text{\rm{i}}}\,\tilde Q_{N\text{\rm{i}}}$ to $ \tilde Q_{N\text{\rm{i}}} \,\mathbb C^n$  is indeed invertible.

We recall that the unperturbed potential $V$ is assumed to satisfy (2.2) and belong to $L_1^1(\mathbb R^+).$
By proceeding as in Section~\ref{section8}, we establish the following:

\begin{enumerate}

\item[\text{\rm(1)}] 
We solve  the  Gel'fand--Levitan system of integral equations \eqref{5.49}  with the kernel $G(x,y)$  chosen as
\begin{equation}
\label{9.11}
G(x,y)= \varphi(i\kappa_N,x)\,\tilde C_{N\text{\rm{i}}}^2\,
\varphi(i\kappa_N,y)^\dagger,
\end{equation}
where $\tilde C_{N\text{\rm{i}}}$ is the $n\times n$ nonnegative hermitian matrix defined in \eqref{9.6}.
We remark that \eqref{9.11} is analogous to \eqref{8.14}.
By proceeding as in the proof of 
Theorem~\ref{theorem8.2}, we show that the solution $\mathcal A(x,y)$ to \eqref{5.49} is given by 
\begin{equation}
\label{9.12}
\mathcal A(x,y)=-\xi_{N\text{\rm{i}}}(x)\,\Omega_{N\text{\rm{i}}}(x)^+\,\xi_{N\text{\rm{i}}}(y)^\dagger, \qquad 0 \le y < x,
 \end{equation}
 which is the analog of \eqref{8.15}.
We note that the $n\times n$ matrix $\xi_{N\text{\rm{i}}}(x)$  appearing in \eqref{9.12} is defined as
\begin{equation}
\label{9.13}
\xi_{N\text{\rm{i}}}(x):=\varphi(i\kappa_N,x)\,\tilde C_{N\text{\rm{i}}},
\end{equation}
which is the analog of \eqref{8.16}, and the $n\times n$ matrix $\Omega_{N\text{\rm{i}}}(x)$ there is defined as
\begin{equation}
\label{9.14}
\Omega_{N\text{\rm{i}}}(x):=
\tilde Q_{N\text{\rm{i}}}+\int_0^x dy\,\xi_{N\text{\rm{i}}}(y)^\dagger\,
\xi_{N\text{\rm{i}}}(y),
\end{equation}
which is the analog of \eqref{8.17}.
We recall that $\Omega_{N\text{\rm{i}}}(x)^+$ denotes the Moore--Penrose inverse of $\Omega_{N\text{\rm{i}}}(x).$

\item[\text{\rm(2)}] 
We define the perturbed potential $\tilde V(x)$ as
\begin{equation}\label{9.15}
\tilde V(x):=V(x)-2\,\ds\frac{d}{dx}\left[ \xi_{N\text{\rm{i}}}(x)\,\Omega_{N\text{\rm{i}}}(x)^+\,\xi_{N\text{\rm{i}}}(x)^\dagger\right].\end{equation}
By proceeding as in the proof of Theorem~\ref{theorem8.3},
we show that $\mathcal A(x,y)$ and $\tilde V(x)$ satisfy (a)--(d) of Theorem~\ref{theorem8.3}. 

\item[\text{\rm(3)}] 
Similarly, by proceeding as in the proof of Theorem~\ref{theorem8.5}, we establish the following:

\begin{enumerate}

\item[\text{\rm(a)}] The quantity $\tilde\varphi(k,x)$ expressed in terms of the unperturbed regular solution
$\varphi(k,x)$ and the quantity $\mathcal A(x,y)$ in \eqref{9.12} as
\begin{equation*}
\tilde\varphi(k,x)=\varphi(k,x)+\int_0^x dy\,\mathcal A(x,y)\,\varphi(k,y),
\end{equation*}
   is a solution to the perturbed Schr\"odinger equation \eqref{2.1}
 with the potential $\tilde V(x)$ given in \eqref{9.15}. Furthermore, with the help of \eqref{9.13} and \eqref{9.14}, we prove that
 the quantity $\varphi(k,x)$ can also be expressed as
  \begin{equation}\label{9.17}
\tilde\varphi(k,x)=\varphi(k,x)- \xi_{N\text{\rm{i}}}(x)\,\Omega_{N\text{\rm{i}}}(x)^+\, \int_0^x dy\,\xi_{N\text{\rm{i}}}(y)^\dagger\, \varphi(k,y).
\end{equation}

\item[\text{\rm(b)}] The perturbed potential $\tilde V(x)$ appearing in \eqref{9.15} 
satisfies \eqref{2.2}. Moreover, for any $a$ and $\varepsilon$ with
$a\ge 0$ and $0 < \varepsilon <1,$ the corresponding potential increment $\tilde V(x)-V(x)$ has the asymptotic behavior
\begin{equation*}
\tilde V(x)- V(x)=O\left(q_{29}(x,a,\varepsilon)\right), \qquad x\to+\infty,
\end{equation*}
where we have defined
\begin{equation*}
\begin{split}
q_{29}(x,a,\varepsilon):=&
 x \, e^{-2\,\kappa_N x} +  \int_x^\infty dy \,|V(y)|
+  
e^{-2 \,\varepsilon\, \kappa_N\,  x} \int_a^{(1-\varepsilon)x} dy\, |V(y)| \\
&+  \int_{(1-\varepsilon) x}^x dy \, e^{-2\, \kappa_N\,(x-y)}\,|V(y)|.
\end{split}
\end{equation*}
Furthermore, if the unperturbed potential $V$ belongs  $L^1_{1+\epsilon}(\mathbb R^+)$ for some fixed $\epsilon \ge 0,$ then 
the perturbed potential $\tilde V$ belongs to $L^1_\epsilon(\mathbb R^+).$

\item[\text{\rm(c)}] Assume that the unperturbed potential $V$ is further restricted to satisfy
\begin{equation*}
|V(x)| \le c\,e^{-\alpha x}, \qquad x \ge x_0,
\end{equation*}
for some positive constants $\alpha$ and $x_0,$ where $c$
denotes a generic constant.
Then, for every $0< \varepsilon <1,$ the potential increment
$\tilde V(x)-V(x)$ has the asymptotic behavior as $x\to+\infty$ given by
\begin{equation*}
\tilde V(x) -V(x) =
\begin{cases}
O\left( e^{-\alpha x}+ e^{-2  \varepsilon \kappa_N x}\right), \qquad \alpha\le 2  \kappa_N, \\
\noalign{\medskip}
O\left( e^{-2 \varepsilon \kappa_N x}\right),  \qquad \alpha > 2 \kappa_N.
\end{cases}
\end{equation*}

\item[\text{\rm(d)}]
If the unperturbed potential $V$ has compact support, then the perturbed potential $\tilde V$ has the asymptotic behavior
\begin{equation*}
 \tilde V(x)= O\left( x\, e^{-2\kappa_N x}\right), \qquad x \to +\infty.
 \end{equation*} 

\item[\text{\rm(e)}] For $ k \neq i \kappa_N,$  the perturbed quantity $\tilde\varphi(k,x)$ given in \eqref{9.17} can be expressed as
\begin{equation*}
\tilde\varphi(k,x)=\varphi(k,x)-\ds\frac{1}{k^2+\kappa_N^2}
\,\xi_{N\text{\rm{i}}}(x) \,\Omega_{N\text{\rm{i}}}(x)^+\left[\xi'_{N\text{\rm{i}}}(x)^\dagger\,\varphi(k,x)-\xi_{N\text{\rm{i}}}(x)^\dagger
\,\varphi'(k,x)\right].
\end{equation*}

\item[\text{\rm(f)}] The perturbed quantity $\tilde\varphi(k,x)$ satisfies the initial conditions \eqref{5.17},
where the matrices $\tilde A$ and $\tilde B$
are expressed in terms of the unperturbed
boundary matrices $A$ and $B$ and the nonnegative matrix $\tilde C_{N\text{\rm{i}}}$ defined in \eqref{9.6}, and we have
\begin{equation}\label{9.24}
\tilde A=A, \quad \tilde B=B-A\,\tilde C_{N\text{\rm{i}}}^2 \,A^\dagger A.
\end{equation}

\item[\text{\rm(g)}] 
The matrices $\tilde A$ and $\tilde B$ appearing in \eqref{9.24} satisfy \eqref{2.6} and \eqref{2.7}. Hence, 
as a consequence of (a) and (f), 
the quantity $\tilde\varphi(k,x)$ is the regular solution 
to the matrix Schr\"odinger equation with the potential $\tilde V(x)$ in \eqref{9.15}
and with the selfadjoint
boundary condition \eqref{2.5} with $A$ and $B$ there replaced with
$\tilde A$ and $\tilde B,$ respectively.

\end{enumerate}

\item[\text{\rm(4)}]  Under the additional assumption $V\in L^1_{2}(\mathbb R^+),$ 
by proceeding as in the proof of Theorem~\ref{theorem8.6} we establish the following:

\begin{enumerate}
\item[\text{\rm(a)}] 

The unperturbed Jost matrix $J(k)$ is transformed into the perturbed Jost matrix as
 \begin{equation*}
\tilde J(k)=\left[I-\ds\frac{2i\kappa_N}{k+i\kappa_N}\,\tilde P_{N\text{\rm{i}}}\right] J(k),\qquad
k\in\overline{\mathbb C^+},
\end{equation*}
where we have defined
\begin{equation*}
\tilde P_{N\text{\rm{i}}}:= L_{N\text{\rm{i}}}\, \tilde C_{N\text{\rm{i}}} \left(\tilde C_{N\text{\rm{i}}}\,L_{N\text{\rm{i}}}^\dagger \,
L_{N\text{\rm{i}}} \,\tilde C_{N\text{\rm{i}}}\right)^+ \tilde C_{N\text{\rm{i}}} \,L_{N\text{\rm{i}}}^\dagger,
\end{equation*}
with $L_{N\text{\rm{i}}}$ being the $n\times n$ matrix appearing 
 in \eqref{9.8}. Moreover, $\tilde P_{N\text{\rm{i}}}$ is an orthogonal projection with the rank $m_{N\text{\rm{i}}},$ where we recall
that $m_{N\text{\rm{i}}}$ is also equal to the rank of the orthogonal projection $\tilde Q_{N\text{\rm{i}}}$ appearing in \eqref{9.2}.  We also prove that
$\tilde P_{N\text{\rm{i}}}$ projects onto a subspace of the kernel of $\tilde J(i\kappa_N)^\dagger,$ where that
subspace has dimension $m_{N\text{\rm{i}}}.$

\item[\text{\rm(b)}] The matrix product $J(k)^\dagger J(k)$ for $k\in\mathbb R$ does not change under
the perturbation,
i.e. we have
 \begin{equation*}
\tilde J(k)^\dagger \tilde J(k)=J(k)^\dagger J(k),\qquad k\in\mathbb R,
\end{equation*}
and hence the continuous part of the spectral measure $d\rho$ does not change
under the perturbation.

\item[\text{\rm(c)}] The perturbation changes the determinant of the Jost matrix as
 \begin{equation}
\label{9.28}
\det[\tilde J(k)]=\left(\ds\frac{k-i\kappa_N}{k+i\kappa_N}\right)^{m_{N\text{\rm{i}}}} \det[J(k)],\qquad
k\in\overline{\mathbb C^+}.\end{equation}

\item[\text{\rm(d)}] Under the perturbation,
the bound state with the energy $-\kappa_N^2$  remains, but its multiplicity is increased from $m_N$ to $m_N+m_{N\text{\rm{i}}}.$ No new bound states are added, and the  bound states 
with the energies $-\kappa_j^2$ and multiplicities $m_j$ for $1\le j\le N-1$
are unchanged. We confirm that the multiplicity of the bound state with the energy
$-\kappa_N^2$ is increased by $m_{N\text{\rm{i}}}$ by using the following argument. The multiplicity of
the bound state at $k=i\kappa_N$
for the unperturbed problem is $m_N.$ It follows from Theorem~3.11.6 of \cite{AW2021} that $\det[J(k)]$ has a 
zero of order $m_N$ at $k=i\kappa_N.$ From \eqref{9.28} we see that $\det[\tilde J(k)]$ has a zero of order $m_N+m_{N\text{\rm{i}}}$ at $k= i\kappa_N.$ Thus, the multiplicity of the bound state
with the energy $-\kappa_N^2$ for the perturbed problem is $m_N+m_{N\text{\rm{i}}}.$ 

\item[\text{\rm(e)}] Under the perturbation, the scattering matrix $S(k)$ undergoes the transformation
 \begin{equation*}
\tilde S(k)=\left[I+\ds\frac{2i \kappa_N}{k-i\kappa_N}\,\tilde P_{N\text{\rm{i}}}\right] S(k)\left[I+\ds\frac{2i\kappa_N}{k-i\kappa_N}\,\tilde P_{N\text{\rm{i}}}\right],
\qquad k\in\mathbb R.
\end{equation*}

\item[\text{\rm(f)}] The perturbation changes the determinant of the scattering matrix as 
  \begin{equation*}
\det[\tilde S(k)]=\left(\ds\frac{k+i\kappa_N}{k-i\kappa_N}\right)^{2\, m_{N\text{\rm{i}}}} \det[S(k)],
\qquad k\in\mathbb R.
\end{equation*}

\item[\text{\rm(g)}] Under the perturbation, the Jost solution $f(k,x)$ is transformed into $\tilde f(k,x)$ as
\begin{equation}
\label{9.31}
\tilde f(k,x)=\left[f(k,x)-\ds\frac{1}{k^2+\kappa_N^2}
\,\xi_{N\text{\rm{i}}}(x) \,\Omega_{N\text{\rm{i}}}(x)^+ q_{30}(x) \right] \left[I- \ds\frac{2i \kappa_N}{k+i \kappa_N} \tilde P_{N\text{\rm{i}}}\right],
\end{equation}
where we have defined
\begin{equation*}
q_{30}(x):=\xi'_{N\text{\rm{i}}}(x)^\dagger\,f(k,x)-\xi_{N\text{\rm{i}}}(x)^\dagger\,f'(k,x).
\end{equation*}
 We prove that the singularity at $k=i\kappa_N$ appearing
 on the right-hand side of \eqref{9.31} is actually a removable singularity, as discussed in Remark~\ref{remark8.7}.
 \end{enumerate}

\item[\text{\rm(5)}] 
Finally,  assuming further that $V\in L^1_2(\mathbb R^+),$ by proceeding as in the
proof of Theorem~\ref{theorem8.8} we establish the following:

\begin{enumerate}
\item[\text{\rm(a)}] Under the perturbation,
the projection matrices $Q_j$ for $1\le j\le N-1$ remain unchanged, i.e. we have
 \begin{equation*}
\tilde Q_j=Q_j,\qquad 1\le j\le N-1.
\end{equation*}

\item[\text{\rm(b)}] Under the perturbation, the Gel'fand--Levitan normalization matrices $C_j$ remain
 unchanged for $1\le j\le N-1,$ i.e. we have
 \begin{equation*}
\tilde C_j=C_j,\qquad 1\le j\le N-1.
\end{equation*}

\item[\text{\rm(c)}] 
The orthogonal projection  $\tilde Q_{N\text{\rm{i}}}$ projects into a subspace of dimension $m_{N\text{\rm{i}}}$ of the kernel of $\tilde J(i\kappa_N).$
\end{enumerate}
\end{enumerate}


\begin{thebibliography}{16}



\bibitem{AM1963} Z. S. Agranovich and V. A. Marchenko, \textit{The inverse problem of
scattering theory,} Gordon and Breach, New York, 1963.


\bibitem{AE2022} T. Aktosun and R. Ercan,
\textit{Direct and inverse scattering problems for the first-order discrete system associated with the derivative NLS system,}
Stud. Appl. Math. {\bf 148}, 270--339 (2022).



\bibitem{AEU2023} 
T. Aktosun, R. Ercan, and M. Unlu,
\textit{The generalized Marchenko method in the inverse scattering problem for a first-order linear system with energy-dependent potentials,}
J. Math. Phys. Anal. Geom. {\bf 19}, 3--58 (2023).


\bibitem{AEU2023a} 
 T. Aktosun, R. Ercan, and M. Unlu,
\textit{The Marchenko method to solve the general system of derivative nonlinear Schr\"odinger equations,}
J. Math. Phys.  {\bf 46}, 073502 (2023).



\bibitem{AK2001}
T. Aktosun and M. Klaus, \textit{Chapter 2.2.4: Inverse theory: problem on the line,} In:
E. R. Pike and P. C. Sabatier (eds.), \textit{Scattering,} Academic Press, London, 2001, pp.
770--785.





\bibitem{AU2022}T. Aktosun and M. Unlu,
\textit{A generalized method for the Darboux transformation,}
J. Math. Phys. {\bf 63}, 103501 (2022).

\bibitem{AW2018}
T. Aktosun and R. Weder, \textit{Inverse scattering for the half-line matrix Schr\"odinger equation,}
J. Math. Phys. Anal. Geom. {\bf 14},  237--269 (2018).


\bibitem{AW2021}
T. Aktosun and R. Weder, \textit{Direct and inverse scattering for the matrix Schr\"odinger equation}, Springer, 2021.

\bibitem{BG2003} A. Ben-Israel and T. N. E. Greville, \textit{Generalized inverses: Theory and applications,} Springer, 2nd ed., 2003.

\bibitem{BK2006} G. Berkolaiko and P. Kuchment, \textit{Introduction to quantum graphs,} Am. Math. Soc., Providence, R.I., 2006. 

\bibitem{CM2009} S. L. Campbell and C. D. Meyer, \textit{Generalized inverses of linear transformations,} SIAM, Philadelphia, 2009.



\bibitem{CS1989} K. Chadan and P. C. Sabatier, \textit{Inverse problems in quantum scattering theory,} 2nd ed., Springer, New York, 1989.

\bibitem{C1955} M. M. Crum, \textit{Associated Sturm--Liouville systems,} Q. J. Math. {\bf 6},  121--127 (1955).

\bibitem{D1882} G. Darboux, \textit{Sur une proposition relative aux \'equations lin\'eaires},  C. R. Acad. Sci. (Paris) {\bf 94},  1456--1459 (1882).

\bibitem{D1978} P.  A. Deift, \textit{Applications of a commutation formula,} Duke Math. J. {\bf 45}, 267--310 (1978).

\bibitem{DT1979} P. A. Deift and E. Trubowitz, \textit{Inverse scattering on the line,} Comm. Pure Appl. Math. {\bf 32},
121--251 (1979).

\bibitem{EK1982} M. S. P. Eastham and H. Kalf, \textit{Schr\"odinger-type operators with continuous spectra}, Pitman, Boston, 1982.   

\bibitem{GL1955} I. M. Gel'fand  and B. M. Levitan, {\textit On the determination of a differential equation from its spectral function,}
Am. Math. Soc. Transl. (ser. 2) {\bf 1}, 253--304 (1955)  [Izv. Akad. Nauk SSSR Ser. Mat. {\bf 15}, 309--360 (1951) (Russian)]  

\bibitem{G1993} F. Gesztesy, \textit{A complete spectral characterization of the double commutation method,} J. Funct. Anal. {\bf 117}, 401--446 (1993).

\bibitem{GST1996} F. Gesztesy,  B. Simon, and G. Teschl, \textit{Spectral deformations of one-dimensional Schr\"odinger operators},
J. Anal. Math. {\bf 70}, 267--324 (1996).

\bibitem{GT1996} F. Gesztesy and G. Teschl, \textit{On the double commutation method,} Proc. Am. Math. Soc. {\bf 124}, 1831--1840 (1996).



\bibitem{J1837} C. G. J. Jacobi, \textit{Zur Theorie der Variations-Rechnung und der Differential-Gleichungen}, J. Reine Angew. Math. {\bf 17},  68--82 (1837). 

\bibitem{K1980} T. Kato, \textit{Perturbation theory for linear operators}, 2nd ed., Springer, Berlin, 1976.

\bibitem{K2024} P. Kurasov, \textit{Spectral geometry of graphs,} Birkh\"auser, Berlin, 2024.




\bibitem{LL1989} L. D. Landau and E. M. Lifschitz, \textit{Quantum mechanics: non-relativistic theory,} 3rd ed., Pergamon Press, New York, 1989.

\bibitem{L1987} B. N. Levitan, \textit{Inverse Sturm--Liouville problems}, VNU Science Press, Utrecht, 1987.
 
\bibitem{LG1964} B. M. Levitan and M. G. Gasymov, \textit{ Determination of a differential operator by two of its  spectra}, Russian Math. Surveys {\bf 19}, 1--6 (1964).

\bibitem{M2011} V. A. Marchenko, \textit{Sturm--Liouville operators and applications,} Rev. ed.,  Am. Math. Soc. Chelsea
Publishing, Providence, RI, 2011.  

\bibitem{MS1991} V. B. Matveev and M. A. Salle, \textit{Darboux transformations and solitons}, Springer, Berlin, 1991.

\bibitem{NW2023} I. Naumkin and R. Weder, \textit{The matrix nonlinear Schr\"odinger equation with a potential,} J. Math. Pures Appl. {\bf 172}, 1--104 (2023).

\bibitem{NJ1955} R. G. Newton and R. Jost, \textit{ The construction of potentials from the $S$--matrix for systems of differential equations,} Nuovo Cimento {\bf 1}, 590--622 (1955).

\bibitem{RS2002} C. Rogers and W. K. Schief, \textit{B\"acklund  and Darboux transformations: geometry and modern applications in soliton theory,} Cambridge University Press, Cambridge, 2002.  

\bibitem{R1987} W. Rudin, \textit{Real and complex analysis,} Mc-Graw Hill, New York, 1987.

\bibitem{S1979} U.-W. Schmincke, \textit{On Schr\"odinger factorization method for Sturm--Liouville operators,} Proc. Royal Soc. Edinburgh  {\bf 80 A},  67--84 (1979).  

\bibitem{S2003} U.-W. Schmincke, \textit{On a paper by Gesztesy, Simon, and Teschl concerning isospectral deformations of ordinary Schr\"odinger operators,} J. Math. Anal Appl. {\bf 277}, 51--78 (2003). 


\end{thebibliography}
\end{document}